\documentclass[12pt,reqno]{amsart}
\usepackage{fullpage,graphicx,subfigure,mathpazo,color}
\usepackage{amsmath,amscd,tikz,mathrsfs}
\usepackage[normalem]{ulem}
\usepackage{amsmath}
\usepackage{epstopdf}
\usepackage{tikz}
\usepackage{setspace}
\usepackage{hyperref}
\newcommand{\ii}{\mathrm{i}}

\newtheorem{thm}{Theorem}
\newtheorem{lem}{Lemma}
\newtheorem{prop}{Proposition}

\newtheorem{definition}{Definition}

\newtheorem{rem}{Remark}

\usepackage{graphicx}      % insert graphic
\usepackage{titlesec}

\titleformat{\section}{\centering\LARGE\bfseries}{\thesection}{1em}{}
\titleformat{\subsection}{\Large\bfseries}{\thesubsection}{1em}{}

\begin{document}

\title{Stability analysis of breathers for coupled nonlinear Schr\"odinger equations}

 \author{Liming Ling}
 \address{School of Mathematics, South China University of Technology, Guangzhou, China 510641}
\email{linglm@scut.edu.cn}
\author{Dmitry E. Pelinovsky}
\address{Department of Mathematics and Statistics, McMaster University,
Hamilton, Ontario, Canada, L8S 4K1}
\email{dmpeli@math.mcmaster.ca}
\author{Huajie Su}
\address{School of Mathematics, South China University of Technology, Guangzhou, China 510641}
\email{747443327@qq.com}

\begin{abstract}
    We investigate the spectral stability of non-degenerate vector soliton solutions and the nonlinear stability of breather solutions for the 
    coupled nonlinear Schrödinger (CNLS) equations. The non-degenerate vector solitons are spectrally stable despite the linearized operator admits either embedded or isolated eigenvalues of negative Krein signature. The nonlinear stability of breathers is obtained by the Lyapunov method with the help of 
    the squared eigenfunctions due to integrability of the CNLS equations.
\end{abstract}

\date{\today}

\maketitle



% Possible journals: Analysis & PDEs

\section{Introduction}
    In this work, we study the spectral stability of the non-degenerate vector soliton solutions and the 
    nonlinear stability of breather solutions in the 
    coupled nonlinear Schr\"odinger (CNLS) equations
    \cite{berkhoe_self_1970,manako_theory_1974,roskes_nonlinear_1976} on the real line (also known as Manakov system):
    \begin{equation}\label{CNLS}
        \begin{split}
            {\rm i} q_{1,t} +\frac{1}{2}q_{1,xx}+(|q_{1}|^{2}+|q_{2}|^{2})q_{1}&=0, \\
            {\rm i}q_{2,t}+\frac{1}{2}q_{2,xx}+(|q_{1}|^{2}+|q_{2}|^{2})q_{2}&=0,
        \end{split}
    \end{equation}
    where $(q_{1},q_{2})=(q_{1}(x,t),q_{2}(x,t))$ are complex-valued functions and $(x,t) \in \mathbb{R} \times \mathbb{R}$. The initial-value problem 
    for the CNLS equations \eqref{CNLS} is globally well-posed in Sobolev 
    space $H^k(\mathbb{R})$ for $k \in \mathbb{N}$, see    \cite{cazenave_introduction_1989}. The CNLS equations \eqref{CNLS} have important applications in Bose-Einstein condensates \cite{qin_nondegenerate_2019} 
    and birefringent fibers
    \cite{agrawal_nonlinear_2019}. 
    The nonlinear terms in CNLS equations \eqref{CNLS} couple two components to describe self-phase and
    cross-phase modulation phenomenon. 
    As an extension of the nonlinear Schr\"odinger equation (NLS), 
    the CNLS equations \eqref{CNLS} can be used to study the dynamics of vector solitons \cite{APT2004}. 
    
    As an example of integrable equations, the CNLS equations \eqref{CNLS} admit Lax pair \cite{manako_theory_1974, wang_integrable_2010},
    bi-Hamiltonian structure \cite{yang_nonlinear_2010},  
    and an infinite set of conservation laws \cite{yang_nonlinear_2010}. 
    The Lax representation is a powerful tool to 
    analyze the properties of equations. 
    The Lax pair 
    for CNLS equations \eqref{CNLS} is
    \begin{equation}
    \label{CNLS-lax}
        \begin{split}
            \mathbf{\Phi}_{x}(\lambda;x,t)=\mathbf{U}(\lambda,\mathbf{q})\mathbf{\Phi}(\lambda;x,t),\quad
            \mathbf{\Phi}_{t}(\lambda;x,t)=\mathbf{V}(\lambda,\mathbf{q})\mathbf{\Phi}(\lambda;x,t)
        \end{split}
    \end{equation}
    where 
    \begin{equation}
    \label{U-V-lax}
            \mathbf{U}(\lambda;x,t)={\rm i}\lambda\sigma_{3}+{\rm i} \mathbf{Q}, \quad 
            \mathbf{V}(\lambda;x,t)={\rm i}\lambda^{2}\sigma_{3}+{\rm i}\lambda\mathbf{Q}
            -\frac{1}{2}({\rm i}\sigma_{3}\mathbf{Q}^{2}-\sigma_{3}\mathbf{Q}_{x}), 
    \end{equation}
    with 
\begin{align*}
\mathbf{Q}(x,t)=\begin{pmatrix}
                0& \mathbf{r}^{T}\\
                \mathbf{q}& 0
            \end{pmatrix},\,\,\,\, 
            \sigma_{3}=\mathrm{diag}(1,-1,-1),
            \quad
            \mathbf{q}=(q_{1},q_{2})^{T}, \quad \mathbf{r} = \mathbf{q}^*,
\end{align*}
and the complex spectral parameter $\lambda$. 
The CNLS equations \eqref{CNLS} have the zero curvature representation 
$$
\mathbf{U}_t-\mathbf{V}_x+[\mathbf{U},\mathbf{V}]=0,
$$ 
where $[\mathbf{U},\mathbf{V}]\equiv \mathbf{U}\mathbf{V}-\mathbf{V}\mathbf{U}$ is the commutator. The Lax pair (\ref{CNLS-lax})--(\ref{U-V-lax}) can be derived from the $3\times 3$ system \cite{APT2004,yang_nonlinear_2010}.  
The infinite set of conservation laws can be obtained by trace formula in the inverse scattering transform for integrable equations with Hamiltonian structure \cite{faddeev1987hamiltonian, koch2018conserved}. In particular, we will use the following conserved quantities of the CNLS equations (\ref{CNLS}):
   \begin{align}
   	\label{con-0}
H_{0}(\mathbf{q})
&= \frac{1}{2}\int_{\mathbb{R}}\left| \mathbf{q} \right|^{2}\mathrm{d}x, \\
\label{con-02}
H_{1}(\mathbf{q})
&=\frac{1}{2}\int_{\mathbb{R}}{\rm i} \mathbf{q}^{\dagger}\mathbf{q}_{x} \mathrm{d}x, \\
\label{con-03}
H_{2}(\mathbf{q}) &= \frac{1}{2}\int_{\mathbb{R}} \left( |\mathbf{q}_{x}|^{2}
-\left|\mathbf{q}\right|^{4} \right) \mathrm{d}x, \\
\label{con-04}
H_{3}(\mathbf{q})
&= \frac{1}{2}\int_{\mathbb{R}}{\rm i} \left(\mathbf{q}_{x}^{\dagger}\mathbf{q}_{xx}+ 3|\mathbf{q}|^{2}\mathbf{q}^{\dagger}_{x}\mathbf{q}\right)
\mathrm{d}x, \\
\label{con-95}
H_{4}(\mathbf{q})
&= \frac{1}{2}\int_{\mathbb{R}}\left( |\mathbf{q}_{xx}|^{2}-4\left|\mathbf{q}^{\dagger}\mathbf{q}_{x}\right|^{2}
-\left( \mathbf{q}^{\dagger}\mathbf{q}_{x} \right)^{2}-\left( \mathbf{q}_{x}^{\dagger}\mathbf{q} \right)^{2}
-4\left| \mathbf{q} \right|^{2}\left| \mathbf{q}_{x} \right|^{2}
+2\left|\mathbf{q}\right|^{6} \right) \mathrm{d}x.
\end{align}
We note that $H_0(\mathbf{q}) = H_0^{(1)}(q_1) + H_0^{(2)}(q_2)$, 
where $H_0^{(j)}(q_j) = \frac{1}{2}\int_{\mathbb{R}} |q_j|^2 \mathrm{d}x$, $j = 1,2$  are independent conserved quantities. Another relevant observation is that the values of $H_1(\mathbf{q})$ and $H_3(\mathbf{q})$ are real due to integration by parts.
 
Various solutions to the CNLS equations \eqref{CNLS} have been derived through different methods. The degenerate vector solitons with the single-humped profiles were initially obtained through the inverse scattering method \cite{manako_theory_1974}. Bright and dark solitons had been discovered using Hirota bilinear method \cite{radhakrishnan1995bright, stalin2020nondegenerate}. The non-degenerate vector solitons with the double-humped profiles were derived by using the Hirota bilinear method in \cite{stalin2019prl}. The asymptotic behavior and collision dynamics of non-degenerate solitons have been investigated in \cite{ramakrishnan2020nondegenerate}. Furthermore, the Darboux transformation can be used to construct non-degenerate vector solitons and breather solutions \cite{ling_darboux_2016,qin_nondegenerate_2019}. 

This paper focuses on investigating the stability of non-degenerate vector soliton and breather solutions from \cite{qin_nondegenerate_2019} in the time evolution of the CNLS equations (\ref{CNLS}). The stability analysis is a fundamental problem of mathematical physics, which is particularly important 
for applications of solitons and breathers in physics.

\subsection{Review on stability results for CNLS equations}
    
The history of the nonlinear stability of solitary waves takes place from the first study on the Korteweg-de Vries (KdV) equation by using the Lyapunov method in \cite{benjamin1972stability}. Regarding the NLS equation, nonlinear stability of ground states was obtained by utilizing concentration compactness principle in  \cite{cazenave1983stable,cazenave_orbital_1982} and \cite{weinstein_lyapunov_1986}. Furthermore, the Lyapunov method was extended to a general class of Hamiltonian nonlinear evolution equations in \cite{grillakis_stability_1987,grillakis_stability_1990}, 
with more results concerning the spectral stability in \cite{grillakis_analysis_1990}. The Lyapunov method incorporates coercivity of the second variation of the Lyapunov functional in the proof of 
nonlinear stability of the orbit of solitary waves. This approach 
has been used in many works, see recent papers in \cite{alejo2021stability, killip_orbital_2022, koch2024multisolitons,laurens_multisolitons_2023}. 
    
For the particular case of the CNLS equations \eqref{CNLS} and their nonintegrable extensions, spectral stability of degenerate vector solitons with the single-humped profiles was obtained in \cite{mesentsev1992stability}. The proofs of their nonlinear stability and instability was developed in \cite{li_structural_1998,li_mechanism_2000}.  Variational characterizations of such vector solitons was developed in \cite{ohta_stability_1996,cipolatti2000orbitally,nguyen2011orbital,nguyen_existence_2015}. The stability analysis of more general vector solitons which include 
multi-humped profiles in one component and the single-humped profiles in another 
component was developed in \cite{pelinovsky_inertia_2005,pelinovsky_instabilities_2005}. Finally, 
bifurcations and stability of such vector solitons was developed recently in 
\cite{yagasaki_bifurcations_2023}. 
    
A difficulty in the nonlinear stability analysis of solitary waves with the multi-humped profiles by using the Lyapunov method arises due to a high number of negative eigenvalues of the second variation of the Lyapunov functional 
and a low number of symmetries that characterize its kernel \cite{kapitula_spectral_2013}. In this work, we consider 
the double-humped profiles of the nondegenerate vector solitons and breather profiles
and use integrability of the CNLS equations \eqref{CNLS} to obtain the squared eigenfunctions of the Lax pair, see \cite{deconinck_orbital_2020, upsal_real_2020}. The squared eigenfunctions satisfy the linearized operators which solve the spectral stability problem and contribute to the nonlinear stability analysis. The closure relation  \cite{gerdjikov_generating_1981,yang_squared_2009,kaup_inverse_2009} provides tools to prove the completeness of squared eigenfunctions in the spectrum of the linearized operator and to compute the inner product between the squared eigenfunctions and adjoint squared eigenfunctions. It also allows us to 
relate the spectrum of the linearized operator with the spectrum of the 
second variation of the Lyapunov functional 
\cite{kapitula_counting_2004,pelinovsky_inertia_2005,haragus_spectra_2008}. 
    
For the nonlinear stability analysis, we also use an additional tool 
that exists due to integrability of the CNLS equations (\ref{CNLS}). 
The same non-degenerate vector solitons can be characterized variationally 
with several Lyapunov functions due to the existence of the higher-order 
conserved quantities. This tool was pioneered in the proof of nonlinear stability of multi-solitons in \cite{Kap07,MS93} and applied 
to studies of nonlinear stability of breathers in \cite{Al2,Al1}, Dirac solitons 
in \cite{PY14}, and periodic waves and black solitons in \cite{GP-15,GP-dark}. Further results on the linear and nonlinear stability of multi-solitons were found recently in \cite{LeCoz,Wang22}. This approach is also useful to characterize linear transverse stability of solitary and periodic waves even if the higher-order conserved quantities do not form coercive Lyapunov functionals, see \cite{LHP-17}. Here we will extend this tool to the CNLS equations \eqref{CNLS}.

\subsection{Main result}

Recall the Galilean transformation for CNLS equations \eqref{CNLS} given by
    \begin{equation}
    \label{speed-CNLS}
        T(a)\mathbf{q}(x,t)={\rm e}^{-2{\rm i} a(x+at)}\mathbf{q}(x+2at,t). 
    \end{equation}
If $\mathbf{q}$ is a solution of \eqref{CNLS}, so is $T(a) \mathbf{q}$ for any $a \in \mathbb{R}$. Without loss of generality, one can consider the standing waves with zero speed and obtain traveling waves with nonzero speed by using (\ref{speed-CNLS}). Consequently, we will consider the non-degenerate vector solitons in the form of the standing waves:
\begin{equation}
\mathbf{q}(x,t) =    \begin{pmatrix}
{\rm e}^{2{\rm i} b_{1}^2 t} & 0 \\
0 & {\rm e}^{2{\rm i} b_{2}^2 t}
\end{pmatrix}     \begin{pmatrix}
 u_1(x) \\
u_2(x)
\end{pmatrix},
\label{standing-wave}
\end{equation}
where $u_1(x),u_2(x) : \mathbb{R} \mapsto \mathbb{C}$ represent 
spatial profiles of the two components with parameters $b_{1}, b_{2} > 0$ satisfying $b_{1} \ne b_{2}$.  Without loss of generality, we assume \(0 < b_{2} < b_{1} \). 

The exact expression for the non-degenerate vector solitons exist in the form, see \cite{qin_nondegenerate_2019},
\begin{eqnarray}
\label{u-1}
u_1(x;b_1,b_2,c_{11},c_{22}) &= \frac{4b_1 c_{11}}{M_{non}(x)}  \left(|c_{22}|^{2}{\rm e}^{2b_{2}x}+\frac{b_{1}-b_{2}}{b_{1}+b_{2}}{\rm e}^{-2b_{2}x}\right), \\
\label{u-2}
u_2(x;b_1,b_2,c_{11},c_{22}) &= \frac{4b_2 c_{22}}{M_{non}(x)} \left(|c_{11}|^{2}{\rm e}^{2b_{1}x}+\frac{b_{2}-b_{1}}{b_{1}+b_{2}}{\rm e}^{-2b_{1}x}\right),
\end{eqnarray}
where 
\begin{equation*}
M_{non}(x)=
\frac{(b_{1}-b_{2})^{2}}{(b_{1}+b_{2})^{2}}{\rm e}^{-2 (b_{1}+b_{2}) x}+
|c_{11}|^{2}{\rm e}^{2 (b_{1}-b_{2}) x}+|c_{22}|^{2}{\rm e}^{-2 (b_{1}-b_{2}) x} + |c_{11}c_{22}|^{2}{\rm e}^{2 (b_{1}+b_{2}) x}
\end{equation*}
and $c_{11}, c_{22} \in \mathbb{C}$ are arbitrary parameters.

The CNLS equations (\ref{CNLS}) are invariant under the four-parameter group of translations:
    \begin{equation}
    \label{sym-CNLS}
S(\theta_1,\theta_2,x_0,t_0) \mathbf{q}(x,t) = \begin{pmatrix}
{\rm e}^{{\rm i} \theta_1} & 0  \\
0 & {\rm e}^{{\rm i} \theta_2} 
\end{pmatrix} \mathbf{q}(x+x_0,t+t_0).
\end{equation}
If $\mathbf{q}$ is a solution of \eqref{CNLS}, so is $S(\theta_1,\theta_2,x_0,t_0)  \mathbf{q}$ for any $\theta_1,\theta_2,x_0,t_0 \in \mathbb{R}$. Without loss of generality, one can take $c_{11}$, $c_{22}$ 
as real parameters by defining $\theta_1$ and $\theta_2$ in the transformation (\ref{sym-CNLS}). Furthermore, if we parameterize 
$$
c_{ii} = \sqrt{\frac{b_{1}-b_{2}}{b_1 + b_2}} {\rm e}^{-2 b_{i} \tau_{i}}, \quad \tau_i \in \mathbb{R},
$$
the non-degenerate vector solitons in the form (\ref{non-sol-1}) reduce to the equivalent form found in \cite{akhmediev1995phase,Sipe}:
    \begin{equation*}
        \begin{split}
            u_{1}=\frac{2 b_{1}\sqrt{b_{1}^{2}-b_{2}^{2}}\ \mathrm{cosh}(2 b_{2} (x-\tau_{2}))}
            {b_{1} \mathrm{cosh}(2b_{1}(x-\tau_{1})) \mathrm{cosh}(2 b_{2} (x-\tau_{2})) -b_{2} \mathrm{sinh}(2 b_{1} (x-\tau_{1})) \mathrm{sinh}(2 b_{2}(x-\tau_{2}))}, \\
            u_{2}=\frac{2 b_{2} \sqrt{b_{1}^{2}-b_{2}^{2}}\ \mathrm{sinh}(2 b_{1}(x-\tau_{1}))}
            {b_{1} \mathrm{cosh}(2 b_{1}(x-\tau_{1})) \mathrm{cosh}(2 b_{2}(x-\tau_{2})) -b_{2} \mathrm{sinh}(2 b_{1}(x-\tau_{1})) \mathrm{sinh}(2 b_{2}(x-\tau_{2}))}.
        \end{split}
    \end{equation*}
Only one parameter from $\tau_1$, $\tau_2$ can be set to zero by the translational symmetry (\ref{sym-CNLS}), the other parameter determines 
the profile of the non-degenerate vector soliton. The spatial profile of $u_1$ is positive, whereas the spatial profile of $u_2$ has a single zero at $x = \tau_1$.  

Figure \ref{Example-non-vector soliton} displays profiles $|u_1|$ and $|u_2|$ of some non-degenerate vector soliton solutions. The positive profile $u_1$ can be either single-humped or  double-humped. The single-zero profile $u_2$ is always a superposition of two solitons of opposite polarity, hence $|u_2|$ is always double-humped. 

\begin{figure}[htbp]
	\centering
	\includegraphics[scale=0.20]{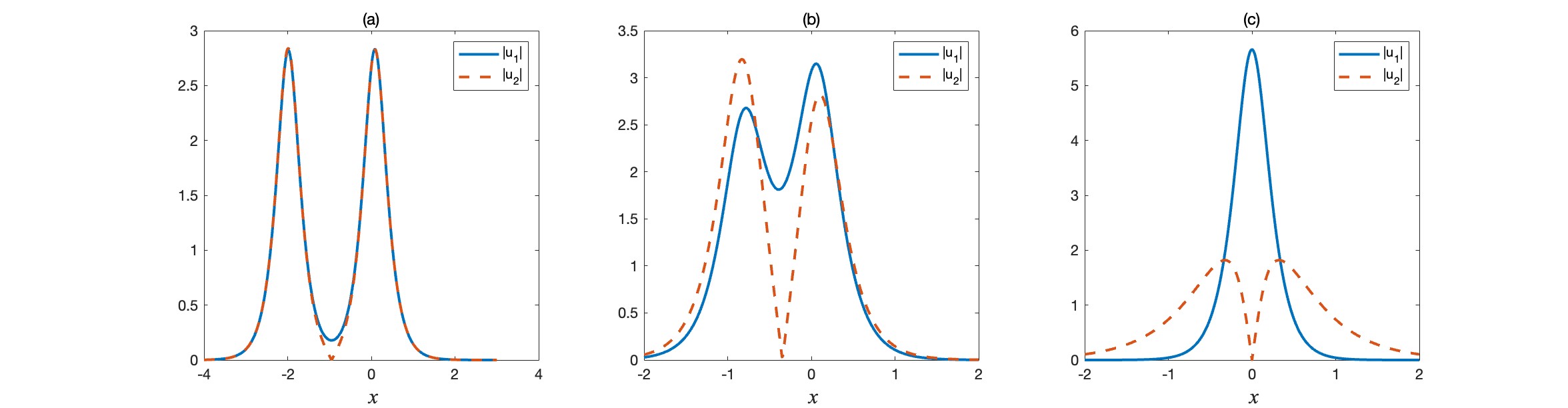}
	\caption{Examples of the non-degenerate vector soliton solutions (\ref{u-1})--(\ref{u-2}) with (a) $b_{1} = 2$, $b_{2}=2.002$, 
		$c_{11}=c_{22}=1$, (b) $b_{1}=2$, $b_{2}=2.2$,
		$c_{11}=c_{22}=1$, and (c) $b_{1}=3$, $b_{2}=1$, $c_{11}=c_{22}=\frac{\sqrt{2}}{2}$. The solid blue line represents $|u_1|$, while the dashed red line represents $|u_2|$. }
	\label{Example-non-vector soliton}
\end{figure}

\begin{rem}
The CNLS equations (\ref{CNLS}) are also invariant under the two-parameter group of rotations:
\begin{equation}
\label{rot-CNLS}
R(\alpha_1,\alpha_2) \mathbf{q}(x,t) = \begin{pmatrix}
\cos(\alpha_1) & \sin(\alpha_1)  \\
-\sin(\alpha_1) & \cos(\alpha_1)
\end{pmatrix} \begin{pmatrix}
\cos(\alpha_2) & {\rm i} \sin(\alpha_2)  \\
{\rm i} \sin(\alpha_2) & \cos(\alpha_2)
\end{pmatrix}  \mathbf{q}(x,t).
\end{equation}
If $\mathbf{q}$ is a solution of \eqref{CNLS}, so is $R(\alpha_1,\alpha_2)  \mathbf{q}$ for any $\alpha_1,\alpha_2 \in \mathbb{R}$.	The non-degenerate vector solitons (\ref{standing-wave}) are written in the form with zero rotational parameters $\alpha_1$, $\alpha_2$.
\end{rem}

For notational clarity, we define
\begin{equation}
\label{non-sol-1}
\mathbf{q}^{[2]}_{non}(x,t;a,b_{1},b_{2};c_{11},c_{22}) := T(a) 
\begin{pmatrix}
{\rm e}^{2{\rm i} b_{1}^{2} t} u_1(x;b_1,b_2,c_{11},c_{22}) \\
{\rm e}^{2{\rm i} b_{2}^{2} t} u_2(x;b_1,b_2,c_{11},c_{22})
\end{pmatrix},
\end{equation}
where $T(a)$ is given by (\ref{speed-CNLS}) and the spatial profiles are given by (\ref{u-1}) and (\ref{u-2}). 

The spectral stability of the standing waves (\ref{standing-wave}) is examined by considering the perturbative solution in the separable form
    \begin{equation}\label{perb-q}
        \mathbf{q}
        =\mathbf{q}^{[2]}_{non} + 
        \varepsilon T(a)\begin{pmatrix}
            {\rm e}^{2{\rm i} b_{1}^2 t} & 0 \\
            0 & {\rm e}^{2{\rm i} b_{2}^2 t}
        \end{pmatrix}
        \left(\mathbf{p}_{1}(x){\rm e}^{\Omega t}+\mathbf{p}_{2}^{*}(x){\rm e}^{\Omega^{*} t}\right)
    \end{equation}
where $\varepsilon$ is a small perturbation, 
$\mathbf{p}_{1}, \mathbf{p}_{2} \in L^{2}(\mathbb{R},\mathbb{C}^2)$, 
and $\Omega \in \mathbb{C}$. Substituting \eqref{perb-q} into \eqref{CNLS}  
and linearizing at the order of $\mathcal{O}(\varepsilon)$ yields the spectral problem for the linearized operator
    \begin{equation}\label{spectral-JL1}
        \mathcal{J}\mathcal{L}_{1}\mathbf{P}=
        \Omega\mathbf{P}, 
    \end{equation}
    where $\mathbf{P} := (\mathbf{p}_{1}^{T}, \mathbf{p}_{2}^{T})^{T}$ and the expressions for $\mathcal{J}$ and $\mathcal{L}_{1}$ can be found in Section 3. In Hilbert space $L^{2}(\mathbb{R},\mathbb{C}^4)$, $\mathcal{J}$ is skew-adjoint and $\mathcal{L}_{1}$ is self-adjoint. 
    
The self-adjoint operator $\mathcal{L}_{1}$ is related to the first 
variational characterization of the non-degenerate vector solitons
\begin{equation}
\label{Lyp-first}
\mathcal{I}_1 := H_2 - 4a H_1 + 4(a^2 + b_1^2) H_0^{(1)} + 4 (a^2 + b_2^2) H_0^{(2)},
\end{equation}
where $H_0 = H_0^{(1)} + H_0^{(2)}$, $H_1$, and $H_2$ are given by 
(\ref{con-0}), (\ref{con-02}), and (\ref{con-03}). 
Euler--Lagrange equations for $\mathcal{I}_1$ with fixed parameters $a \in \mathbb{R}$, $b_1, b_2 > 0$ are given by 
\begin{equation}
\label{CNLS-trav}
\begin{split}
2 {\rm i} a q_{1,x} +\frac{1}{2}q_{1,xx}+(|q_{1}|^{2}+|q_{2}|^{2})q_{1}&= 2(a^2+b_1^2) q_1, \\
2 {\rm i} a q_{2,x}+\frac{1}{2}q_{2,xx}+(|q_{1}|^{2}+|q_{2}|^{2})q_{2}&=2(a^2+b_2^2) q_2,
\end{split}
\end{equation}
which are satisfied by $\mathbf{q} =   \mathbf{q}^{[2]}_{non}(\cdot,t;a,b_{1},b_{2};c_{11},c_{22})$ 
for every $t \in \mathbb{R}$, $c_{11}, c_{22} \in \mathbb{C}$. 
Adding a perturbation $\mathbf{p}$ to $\mathbf{q}^{[2]}_{non}$
in $\mathcal{I}_1$ and expanding it near $\mathbf{q}^{[2]}_{non}$ yields 
\begin{equation}
\label{def-L1}
\mathcal{I}_1(\mathbf{q}^{[2]}_{non} + \mathbf{p}) = \mathcal{I}_1(\mathbf{q}^{[2]}_{non}) + (\mathcal{L}_1 
\mathbf{P},  \mathbf{P}) + \mathcal{O}(\| \mathbf{p} \|_{H^1}^3),
\end{equation}
with the same self-adjoint operator $\mathcal{L}_1$ in $L^{2}(\mathbb{R},\mathbb{C}^4)$ and with $\mathbf{P} := (\mathbf{p}^{T}, \mathbf{p}^{\dagger})^{T}$.
    
The method of squared eigenfunctions allows us to construct 
all solutions of the spectral problem \eqref{spectral-JL1}.
Denote the stability spectrum
    \begin{equation*}
        \sigma_{s}(\mathcal{J}\mathcal{L}_{1})=\{
        \Omega\in\mathbb{C} : \quad \mathbf{P}\in L^{\infty}(\mathbb{R},\mathbb{C}^4)
        \}.
    \end{equation*} 
The following theorem gives the spectral stability of the non-degenerate vector solitons by using the squared eigenfunctions.

    \begin{thm}\label{spectral-stability}
        The non-degenerate vector solitons $\mathbf{q}_{non}^{[2]}$ for the CNLS equations \eqref{CNLS} are spectrally stable, as indicated by 
        $$
        \sigma_{s}(\mathcal{J}\mathcal{L}_{1})\subset {\rm i}\mathbb{R}.
        $$ 
        If $\frac{b_{1}}{b_{2}} \leq \frac{1}{\sqrt{2}}$ or $\frac{b_{1}}{b_{2}} \geq \sqrt{2}$, there exist embedded eigenvalues of $\mathcal{J}\mathcal{L}_{1}$
        that are limit points of the essential spectrum of $\mathcal{J}\mathcal{L}_{1}$. For 
        $\frac{1}{\sqrt{2}} < \frac{b_{1}}{b_{2}} < \sqrt{2}$, the closure of the essential spectrum has no intersection with the point spectrum.
    \end{thm}

\begin{rem}
The result of Theorem \ref{spectral-stability} agrees with the theory based on the dimension of the negative subspace 
$$
\mathcal{N}_1 := \{ \mathbf{P} \in H^1(\mathbb{R},\mathbb{C}^4) : \quad (\mathcal{L}_1 \mathbf{P}, \mathbf{P}) < 0 \}, 
$$
see \cite[Section 3]{pelinovsky_instabilities_2005}. Eigenfunctions for embedded eigenvalues were found explicitly, see equation (36) in \cite{pelinovsky_instabilities_2005}, and these eigenfunctions attain the negative values of $(\mathcal{L}_1 \mathbf{P},  \mathbf{P})$. If $\frac{1}{\sqrt{2}} < \frac{b_{1}}{b_{2}} < \sqrt{2}$, the same eigenvalues with negative values of $(\mathcal{L}_1 \mathbf{P},  \mathbf{P})$ are isolated from the essential spectrum of $\mathcal{J}\mathcal{L}_{1}$. Such eigenvalues on $i \mathbb{R}$ with negative values of $(\mathcal{L}_1 \mathbf{P},  \mathbf{P})$ are referred to as eigenvalues of negative Krein signature.
\end{rem}
    
It is natural to ask whether the non-degenerate vector solitons 
remain stable under nonlinear perturbation in $H^1(\mathbb{R},\mathbb{C}^4)$. 
However, the nonlinear orbital stability theory from     \cite{grillakis_stability_1987,grillakis_stability_1990} does not 
hold for such vector solitons with multi-humped profiles because 
the negative subspace of $\mathcal{N}_1$ restricted 
at the tangent space of fixed $H_0^{(1)}$, $H_0^{(2)}$, and $H_1$ 
is two-dimensional, see equation (41) in \cite{pelinovsky_instabilities_2005}.

As a result, we provide a second variational characterization 
of the non-degenerate vector solitons by using the higher-order 
conserved quantities $H_3$ and $H_4$ in (\ref{con-03}) and (\ref{con-04}). However, the second characterization 
does not distinguish between the non-degenerate vector solitons and 
the breather solutions, to which they are particular cases \cite{silberberg1995rotating}. The nonlinear orbit of the breather solutions is constructed by using the scattering parameters in the Darboux transformation as in the following definition.

    \begin{definition}
    	\label{def-breather}
Fix the scattering parameters $(c_{11},c_{21},c_{12},c_{22})\in \mathbb{C}^{4}$ and the spectral parameters $a\in \mathbb{R}$, $b_{1},b_{2}>0$ such that $b_{1}\ne b_{2}$. The breather solutions of CNLS equations \eqref{CNLS} are given by 
    \begin{equation}\label{breather}
\mathbf{q}^{[2]}(x,t;a,b_{1},b_{2};c_{11},c_{22},c_{12},c_{21}) := 4 T(a) \frac{N(x,t)}{M(x,t)},
\end{equation}
where
\begin{align*}
N(x,t) &:= \left[ 
                {\rm e}^{2{\rm i} b_{1}^{2} t}
                b_{1}
            \left( |\mathbf{c}_2|^2 {\rm e}^{2b_{2}x}+\frac{b_{1}-b_{2}}{b_{1}+b_{2}}{\rm e}^{-2b_{2}x}\right)
-{\rm e}^{2{\rm i} b_{2}^{2} t} 
                \frac{2b_{1}b_{2}}{b_{1}+b_{2}} \mathbf{c}_1^{\dag}\mathbf{c}_2 {\rm e}^{2b_{1}x}
\right] \mathbf{c}_1
            \\
            & \qquad 
                +\left[
                    {\rm e}^{2{\rm i} b_{2}^{2}t}
                    b_{2}
                \left( |\mathbf{c}_1|^2 {\rm e}^{2b_{1}x}+\frac{b_{2}-b_{1}}{b_{1}+b_{2}}{\rm e}^{-2b_{1}x}\right)
                -{\rm e}^{2{\rm i} b_{1}^{2} t} 
\frac{2b_{1}b_{2}}{b_{1}+b_{2}} \mathbf{c}_2^{\dag}\mathbf{c}_1 {\rm e}^{2b_{2}x}           \right] 
 \mathbf{c}_2,
\end{align*}
and
\begin{align*}
M(x,t) &:= \frac{(b_{1}-b_{2})^{2}}{(b_{1}+b_{2})^{2}}{\rm e}^{-2 (b_{1}+b_{2}) x} + |\mathbf{c}_1|^2 {\rm e}^{2 (b_{1}-b_{2}) x} + |\mathbf{c}_2|^2 {\rm e}^{-2 (b_{1}-b_{2}) x} \\
& \qquad 
        +\left(
       |\mathbf{c}_{1}|^{2}|\mathbf{c}_{2}|^{2}
                    -\frac{4b_{1}b_{2}}{(b_{1}+b_{2})^{2}}\mathbf{c}_{1}^{\dagger}\mathbf{c}_{2}\mathbf{c}_{2}^{\dagger}\mathbf{c}_{1} \right)
        {\rm e}^{2 (b_{1}+b_{2}) x} -\frac{8b_{1}b_{2}}{(b_{1}+b_{2})^{2}}\mathrm{Re}\left(
            \mathbf{c}_1^{\dag}\mathbf{c}_2 {\rm e}^{2{\rm i}(-b_{1}^{2}+b_{2}^{2})t}
        \right),
\end{align*}
for $\mathbf{c}_{i}=(c_{1i}, c_{2i})^{T}$, $i=1,2$.
Non-degenerate vector soliton solutions \eqref{non-sol-1} are given by the breather solutions \eqref{breather} with $c_{12}=c_{21}=0$.  
    \end{definition}

    \begin{figure}[htbp]
        \centering
        \includegraphics[scale=0.19]{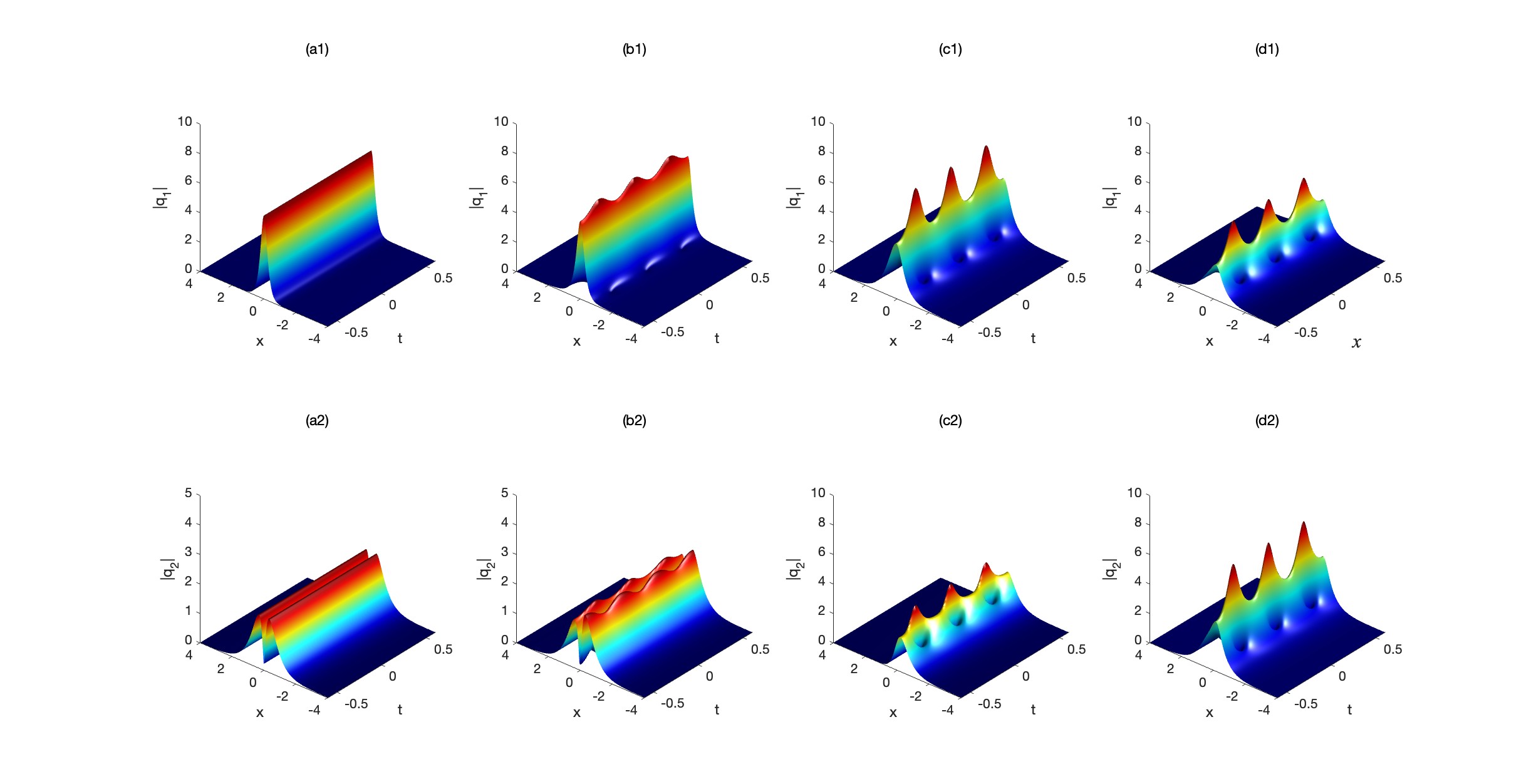}
        \caption{Examples of the breather solutions (\ref{breather}) with 
        	$a = 0$, $b_{1}=3$, $b_{2}=1$, $c_{11}=c_{22}=\frac{\sqrt{2}}{2}$, and 
        (a1, a2) $c_{12}=c_{21}=0$,  
        (b1, b2) $c_{12}=c_{21}=0.1$, 
        (c1, c2) $c_{12}=c_{21}=0.5$, 
        (d1, d2) $c_{12}=0.5$, $c_{21}=1$. 
        Top (bottom) panels show the norm of the first (second) component.}
        \label{Example-breather}
        \end{figure}

Examples of the profiles of breather solutions of Definition \ref{def-breather} are shown in Figure \ref{Example-breather}. Panels (a1,a2) display the time-independent dynamics of a non-degenerate vector soliton since $c_{12} = c_{21} = 0$, where other panels display the time-periodic dynamics of the breather solutions with different nonzero values of $c_{12}$, $c_{21}$.

The following theorem presents the main result of this paper on the nonlinear stability of the breather solutions of Definition \ref{def-breather}.

    \begin{thm}\label{nonlinear-stability}
        The breather solutions \eqref{breather} are nonlinearly stable in the following sense. For any initial 
        condition $\mathbf{u}_{0} \in H^{2}(\mathbb{R},\mathbb{C}^2)$, denote the global solution of the CNLS equations \eqref{CNLS} by  $\mathbf{u}$.
        Given any positive 
        constant $\epsilon>0$, there exist $\delta>0$ such that if 
        \begin{equation*}
            \|\mathbf{u}_{0} -\mathbf{q}^{[2]}(\cdot,0;a,b_1,b_2;c_{11}(0),c_{22}(0),c_{12}(0),c_{21}(0))\|_{H^{2}}< \delta, 
        \end{equation*}
        for some breather solutions $\mathbf{q}^{[2]}$ with parameters $a \in \mathbb{R}$, $b_1,b_2 >0$ such that $b_1 \neq b_2$, and 
        $c_{ij}(0) \in \mathbb{C}$, $1 \leq i,j \leq 2$,  
        then there exist $c_{ij}(t)\in C^1(\mathbb{R}, \mathbb{C}), 1\leq i,j \leq 2$ such that 
        \begin{equation*}
            \|\mathbf{u}(\cdot,t)-\mathbf{q}^{[2]}(\cdot,t;a,b_1,b_2;c_{11}(t),c_{22}(t),c_{12}(t),c_{21}(t))\|_{H^{2}}< \epsilon
        \end{equation*}
        for any $t\in\mathbb{R}$.  The rate of change of $c_{ij}(t)$ are controlled by
        \begin{equation}\label{derivative-cij}
            \sum_{i,j}|\partial_{t}c_{ij}(t)|\leq C \epsilon
        \end{equation} 
        for all $t\in \mathbb{R}$ and some constant $C$. 
    \end{thm}
    
    \begin{rem}
    Compared to the standard applications of the Lyapunov method, it is not sufficient to use the translation and phase symmetries \eqref{sym-CNLS} of the CNLS equations \eqref{CNLS}, as well as the rotational symmetry \eqref{rot-CNLS}, for the proof of the nonlinear stability of the non-degenerate vector solitons because the second variation of the corresponding Lyapunov functional is not coercive in the tangent plane related to these symmetries. The choice of the scattering parameters $c_{11},c_{22},c_{12},c_{21}$, which can be viewed as generalizations of symmetries in the Darboux transformation \cite{koch2024multisolitons}, addresses this issue. When the four complex-valued scattering parameters are perturbed, the non-degenerate vector solitons are not only translated according to the symmetries \eqref{sym-CNLS} and \eqref{rot-CNLS} but also transform into more general breather solutions. Therefore, it is natural to consider the nonlinear stability of the set of breather solutions. Evidences that breather solutions rather than the non-degenerate vector solitons are more appropriate objects in the construction of the nonlinearly stable orbits can be seen from various soliton interaction scenarios in a related complex short-pulse equation \cite{Prinari,Ling}.
    \end{rem}

\subsection{Main steps of the proof} 

The CNLS equations \eqref{CNLS} are given by the second flow of the Hamiltonian equations
  \begin{equation}
  \label{Ham-str}
  \mathbf{q}_{t} = \frac{1}{2} \mathcal{J} \frac{\delta H_{2}}{\delta \mathbf{q}},
  \end{equation}
where $\mathcal{J}=-{\rm i}$ is a skew-adjoint operator. The Hamiltonian formulation \eqref{Ham-str} is used in the construction of the first variational characterization of the non-degenerate vector solitons, as in (\ref{Lyp-first}). However, since this is not sufficient for the proof of their nonlinear stability, we proceed with the second variational characterization, which extends to the entire family of breather solutions of Definition \ref{def-breather}. We define the following Lyapunov functional
    \begin{align}
  \mathcal{I}_{2} &:= H_{4} - 8aH_{3} +
  4(6a^{2}+b_{1}^{2}+b_{2}^{2})H_{2} - 16 a(2a^{2}+b_{1}^{2}+b_{2}^{2})H_{1} 
  \notag \\
  \label{Ly-non}
& \qquad + 16(a^{2}+b_{1}^{2})(a^{2}+b_{2}^{2}) H_{0}.
    \end{align}
We show in Appendix \ref{A.ISM} from the trace formula that  $\mathbf{q}^{[2]}(x,t;a,b_{1},b_{2};c_{11},c_{22},c_{12},c_{21})$ 
is a critical point of $\mathcal{I}_2$.

For the solution 
    $\mathbf{u} = \mathbf{u}(x,t)$ in Theorem \ref{nonlinear-stability}, the continuity and conservation 
    of the Lyapunov functional lead to
    \begin{equation*}
        \mathcal{I}_{2}(\mathbf{u}(\cdot,t))-\mathcal{I}_{2}(\mathbf{q}^{[2]}(\cdot,t))=
        \mathcal{I}_{2}(\mathbf{u}_{0})-\mathcal{I}_{2}(\mathbf{q}^{[2]}(\cdot,0)) \leq C
        \|\mathbf{u}_{0} - \mathbf{q}^{[2]}(\cdot,0)\|_{H^{2}},
    \end{equation*}
for some $C > 0$. On the other hand, for the perturbation  $\mathbf{p}$ 
added to $\mathbf{q}^{[2]}$, we can expand $\mathcal{I}_2$ near 
$\mathbf{q}^{[2]}$ and obtain 
\begin{equation}
\label{L2}
\mathcal{I}_2(\mathbf{q}^{[2]} + \mathbf{p}) = \mathcal{I}_2(\mathbf{q}^{[2]}) + \langle \mathcal{L}_2 
\mathbf{P},  \mathbf{P} \rangle_{L^2} + \mathcal{O}(\| \mathbf{p} \|_{H^2}^3),
\end{equation}
defined by a self-adjoint operator $\mathcal{L}_2$ in $L^{2}(\mathbb{R},\mathbb{C}^4)$ acting on $\mathbf{P} := (\mathbf{p}^{T}, \mathbf{p}^{\dagger})^{T}$.

If $\mathcal{L}_2$ were coercive in $H^{2}(\mathbb{R},\mathbb{C}^4)$,  then 
we could obtain the stability of the non-degenerate vector solitons.  
However, this is not the case. The self-adjoint operator $\mathcal{L}_{2}$ admits a nontrivial kernel of dimension {\em eight} and a negative subspace of dimension {\em two}. Introducing the nonlinear orbit is a method to overcome this difficulty and obtain coercivity at the tangent plane defined by appropriate orthogonality conditions. 
    
To characterize the spectrum of $\mathcal{L}_{2}$, it suffices to study eigenfunctions of the operator $\mathcal{J}\mathcal{L}_{2}$. Using the higher flow in Hamiltonian equations, the spectrum of $\mathcal{J}\mathcal{L}_{2}$ can be determined by the method of squared eigenfunctions. The closure relation can then be used to calculate the inner product $(\mathcal{L}_{2}\cdot,\cdot)$ through these squared eigenfunctions. This approach allows us to determine the number of negative eigenvalues and the kernel of $\mathcal{L}_{2}$. 

The spectrum of  $\mathcal{J}\mathcal{L}_{2}$ is very similar to the spectrum of 
 $\mathcal{J}\mathcal{L}_{1}$ in Theorem \ref{spectral-stability} but includes no additional eigenvalues on $i \mathbb{R}$ (neither isolated nor embedded), see Theorem \ref{thm-J-L2} in Section \ref{JL2-L2}. Similarly, the spectrum of $\mathcal{L}_2$ contains {\em two} negative eigenvalues, see Theorem \ref{spectrum-L2-1}, compared to the spectrum of $\mathcal{L}_1$ which contains {\em four} negative eigenvalues, see Remark \ref{rem-spectrum-L1}.
    
The main difficulty in the proof of nonlinear stability of breather solutions lies in evaluating the inner product $(\mathcal{L}_{2}\cdot,\cdot)$ on the subspace 
$$
\mathrm{gKer}(\mathcal{J}\mathcal{L}_{2}) \backslash \mathrm{Ker}(\mathcal{J}\mathcal{L}_{2}) 
$$ 
in $L^2(\mathbb{R},\mathbb{C}^4)$, which corresponds to the zero eigenvalue of $\mathcal{J}\mathcal{L}_{2}$ with high algebraic multiplicity. By defining the skew-symmetric differential form
    \begin{equation}\label{diff-form}
        \omega(\mathbf{h},\mathbf{f})=(\mathbf{h}^{\dagger}\mathcal{J}\mathbf{f})\mathrm{d}x ,\quad 
        \mathbf{f}\in\mathrm{Ker}(\mathcal{J}\mathcal{L}_{2}),\,\, 
    \mathbf{h}\in \mathrm{gKer}(\mathcal{J}\mathcal{L}_{2})\backslash 
    \mathrm{Ker}(\mathcal{J}\mathcal{L}_{2})
    \end{equation}
    and transforming the inner product $(\mathcal{L}_{2}\cdot,\cdot)$ to the integral of the 
    differential form $\omega(\mathbf{h},\mathbf{f})$, we transform the 
   corresponding inner product $(\mathcal{L}_{2}\cdot,\cdot)$ into the inner product between squared eigenfunctions and adjoint squared eigenfunctions. 
    We present a new method to calculate the inner product between squared eigenfunctions  and adjoint squared eigenfunctions, which relies on the 
properties of integrability rather than on the explicit expressions for eigenfunctions.

    In our approach, the inner product between squared eigenfunctions and adjoint squared eigenfunctions can be determined by analyzing the behavior of squared eigenfunction matrices 
    at infinity. The connection arises from the symmetry of the potential 
    in the Lax pair \eqref{CNLS-lax} and the differential equations 
    satisfied by the squared eigenfunction matrices.

The kernel $\mathrm{Ker}(\mathcal{L}_{2})$ can be obtained 
    through derivatives of the four complex scattering parameters 
    $(c_{11},c_{12},c_{21},c_{22})$ in the orbit of breather solutions (\ref{breather}). Moreover, functions corresponding to the negative subspace
    $$
\mathcal{N}_2 :=    \{ \mathbf{P} \in H^2(\mathbb{R},\mathbb{C}^4) : \quad (\mathcal{L}_{2} \mathbf{P}, \mathbf{P}) < 0 \}
    $$ 
    can be derived from variations of conserved quantities 
    that are linear combinations of conservation laws. 
    By considering modulation and constructing the reduced Hamiltonian 
    \cite{grillakis_stability_1987,grillakis_stability_1990}, we derive the nonlinear stability result of Theorem \ref{nonlinear-stability}.

\subsection{Notations}

Let us consider the real Hilbert space 
    \begin{equation*}
        \mathrm{X}=\left\{
            \begin{pmatrix}
                u_{1} &
                u_{2} &
                u_{1}^{*} &
                u_{2}^{*} 
            \end{pmatrix}^{T}: \quad u_{1},u_{2}\in L^{2}(\mathbb{R},\mathbb{C})
        \right\}\subset L^{2}(\mathbb{R},\mathbb{C}^{4}), 
    \end{equation*}
    equipped with the inner product 
    \begin{equation*}
        (\mathbf{f},\mathbf{g})=\mathrm{Re}\int_{\mathbb{R}}\mathbf{f}^{\dagger}\mathbf{g}
        \mathrm{d}x. 
    \end{equation*}
    The decomposition for real Hilbert space $L^2(\mathbb{R},\mathbb{C}^4) = \mathrm{X}\oplus {\rm i}\mathrm{X}$ holds. 
    
Notation $\cdot^{*}$ represents the complex conjugate and $\cdot^{\dagger}$ represents the operator adjoint or conjugate and transpose for matrices. 
For example, 
$$
({\rm i}\partial_{x})^{*}=(-{\rm i})\partial_{x} \quad 
\mbox{\rm and} \quad  
({\rm i}\partial_{x})^{\dagger}=(-\partial_{x})(-{\rm i})={\rm i}\partial_{x}.
$$ 
    Then we can view any vector $\mathbf{u}\in L^{2}(\mathbb{R},\mathbb{C}^2)$ in the space $\mathrm{X}$ through the 
    bijection
    \begin{equation*}
        \iota (\mathbf{u})=
        \begin{pmatrix}
            \mathbf{u}\\
            \mathbf{u}^{*}
        \end{pmatrix}\in \mathrm{X}.
    \end{equation*}
    The bijection $\iota$ is also an isomorphism between the Hilbert space $L^{2}(\mathbb{R},\mathbb{C}^{2})$ and $\mathrm{X}$
    since 
    \begin{equation*}
        (\iota (\mathbf{u}),\iota (\mathbf{v}))=2 (\mathbf{u},\mathbf{v}). 
    \end{equation*}
    Then any operator $\mathcal{A}$ acting on $L^{2}(\mathbb{R},\mathbb{C}^2)$ can naturally be extended on
    $\mathrm{X}$:
    \begin{equation*}
        \mathcal{A}'\begin{pmatrix}
            \mathbf{u} \\
            \mathbf{u}^{*}
        \end{pmatrix}=\iota \mathcal{A}(\mathbf{u})=
        \begin{pmatrix}
            \mathcal{A}\mathbf{u} \\
            (\mathcal{A}\mathbf{u})^{*}
        \end{pmatrix}.
    \end{equation*}
For a functional \(\mathcal{K}(\mathbf{q})\) on $L^{2}(\mathbb{R},\mathbb{C}^2)$, the first variation $\frac{\delta \mathcal{K}}{\delta \mathbf{q}}(\mathbf{q})$ is given by
        \begin{equation*}
            \left(\mathbf{v}, \frac{\delta \mathcal{K}}{\delta \mathbf{q}}(\mathbf{q})\right)=
            \lim_{\epsilon\to 0}\frac{\mathcal{K}(\mathbf{q}+\epsilon \mathbf{v})-\mathcal{K}(\mathbf{q})}{\epsilon}
        \end{equation*}
        for $\mathbf{v}\in L^{2}(\mathbb{R},\mathbb{C}^2)$. 
The second variation is given by 
\begin{equation*}
\frac{\delta^{2} \mathcal{K}}{\delta^{2} \mathbf{q}}(\mathbf{q})[\mathbf{v}]=
\lim_{\epsilon\to 0}\frac{\frac{\delta \mathcal{K}}{\delta \mathbf{q}}(\mathbf{q}+\epsilon \mathbf{v})-\frac{\delta \mathcal{K}}{\delta \mathbf{q}}(\mathbf{q})}{\epsilon}. 
\end{equation*}

    In the following, we abuse the notation $\mathcal{A}$ to represent $\mathcal{A}'$ without ambiguity 
    and sometimes we mean 
    $\mathbf{u}$ be $\iota(\mathbf{u})$ in $\mathrm{X}$. 
    The reason why we consider $\mathrm{X}$ instead of $L^{2}(\mathbb{R},\mathbb{C}^2)$ is that it makes the expression 
    and calculation of functions more convenient compared to the separation of
    the real and imaginary parts of functions. 
    The functions in $\mathrm{X}$ can be obtained by the functions
    $(\mathbf{q}^{T}, \mathbf{r}^{T})^{T} \in L^{2}(\mathbb{R},\mathbb{C}^4)$ with symmetries 
    $\mathbf{r}=\mathbf{q}^{*}$ which are the symmetries of potentials in the Lax pair (\ref{CNLS-lax}). Hence 
    it is natural to consider the space $\mathrm{X}$ in the dynamics of the CNLS equations (\ref{CNLS}).
    
    It is important to note that both the space $\mathrm{X}$ and $L^{2}(\mathbb{R},\mathbb{C}^{4})$
    can be interpreted as complex linear spaces. When we examine the eigenspace of an operator, 
    we consider complex linear combinations of eigenfunctions. 

\subsection{Main contributions}

Our main contributions can be listed as follows.

    \begin{itemize}
        \item The spectral stability for non-degenerate vector solitons is obtained by using the method of squared eigenfunctions. The spectral problem of the linearized operator $\mathcal{J}\mathcal{L}_{1}$ is solved by squared eigenfunctions and the spectrum of the linearized operator $\mathcal{J}\mathcal{L}_{1}$ is obtained with the dimension of the negative subspace $\mathcal{N}_1$. The non-degenerate vector solitons are spectrally stable. In addition, the linearized operator $\mathcal{J}\mathcal{L}_{1}$ admits either embedded or isolated eigenvalues of negative Krein signature. \\
            
        \item We connect the inner product $(\mathcal{L}_{2}\cdot,\cdot)$ for 
              eigenfunctions in the generalized kernel of the operator $\mathcal{J}\mathcal{L}_{2}$ with 
              the integral in the skew-symmetric differential form. This integral can be computed from the 
              squared eigenfunction matrix due to integrability of the CNLS equations (\ref{CNLS}). The dimension of the negative 
              subspace $\mathcal{N}_2$ is calculated.  Our method can be readily extended to the studies of stability of solitons or breathers in other integrable equations by using the integrability tools.\\
              
        \item The nonlinear stability of the breather solutions in the CNLS equations (\ref{CNLS}) is derived for the first time. The nonlinear 
        stability of the non-degenerate vector solitons with the spectral parameters $a, b_1, b_2$ holds when they are 
        included in the orbit of the breather solutions of Definition \ref{def-breather} given by the scattering parameters $c_{11},  c_{22}, c_{12}, c_{21}$ in the relevant Darboux transformation. 
    \end{itemize}

\subsection{Outline}

This paper is organized as follows. 

\begin{itemize}
\item In Section 2, we construct the
non-degenerate vector soliton and breather solution by using 
Darboux transformation. We also introduce squared eigenfunctions and the squared eigenfunction matrix. The Lyapunov functional and the properties of its second variation are analyzed.  \\

\item In Section 3, we investigate the spectral stability of non-degenerate vector soliton solutions. The proof of Theorem \ref{spectral-stability} relies on analyzing the spectrum of the linearized operator $\mathcal{J}\mathcal{L}_{1}$ constructed by using squared eigenfunctions. \\

\item In Section 4, we construct the squared eigenfunctions in higher flow to prove that the squared eigenfunctions satisfy the spectral problem for $\mathcal{J}\mathcal{L}_{2}$. Subsequently, 
the closure relation induces the spectrum of $\mathcal{L}_{2}$. \\

\item In Section 5, we prove the nonlinear stability of breather solutions in Theorem \ref{nonlinear-stability}.
\end{itemize}

\section{Darboux transformation and breathers}

We introduce the Darboux transformation for constructing breather solutions and their squared eigenfunctions. We also 
compute values of the conserved quantities at the breather solutions. 

The $N$-fold Darboux transformation is applied to the Lax pair \eqref{CNLS-lax}--\eqref{U-V-lax}
in order to obtain the same Lax pair but with a new potential that gives a new 
solution to the coupled NLS equation (\ref{CNLS}). Compared 
to the Lax spectrum of the original solution, the Lax spectrum 
of the new solution contains $N$ additional isolated eigenvalues. Since the zero solution is inherently a solution for the CNLS equations \eqref{CNLS}, applying the $N$-fold Darboux transformation yields $N$-soliton solutions of the 
CNLS equations \eqref{CNLS}.

The trace formulas (see Appendix \ref{A.ISM}) suggest that $N$-soliton solutions satisfy ordinary differential equations (ODEs). The Lyapunov functional $\mathcal{I}_2$ in (\ref{Ly-non}) is derived from the corresponding system of two fourth-order ODEs as a linear combination of the conservation laws, each of which remains independent of time.

\subsection{Construction of breathers}

Denote for any matrix $\mathbf{M}=(m_{ij})_{1\leq i,j\leq N}$, 
\begin{equation*}
(\mathbf{M})^{off}=\begin{pmatrix}
0 & m_{12} & m_{13}\\
m_{21}& 0 & 0\\
m_{31}& 0 & 0\\
\end{pmatrix}.
\end{equation*}
To get general $N$-soliton solutions for CNLS equations \eqref{CNLS}, we 
pick $N$ distinct spectral parameters $\{\lambda_{i}\}_{i=1}^{N}\subset 
\mathbb{C}^{+}=\{z\in\mathbb{C}:\mathrm{Im}(z)>0\}$ such that $\lambda_i \neq \lambda_j$ for every $i \neq j$ and use the following $N$-fold Darboux transformation.

\begin{prop}\cite{ling_darboux_2015, ling_darboux_2016}
	\label{thm-DT}
	Let $\mathbf{\Phi}^{[0]}$ be a fundamental solution and $\mathbf{Q}^{[0]}$ be the corresponding potential for the Lax pair \eqref{CNLS-lax}--\eqref{U-V-lax}. Define the $N$-fold Darboux matrix
	\begin{equation}
	\label{Darboux}
	\mathbf{D}^{[N]}(\lambda;x,t)=\mathbb{I}_{3}-\sum_{i=1}^{N}
	\frac{\lambda_{i}-\lambda_{i}^{*}}{\lambda-\lambda_{i}^{*}}
	|\mathbf{x}_i\rangle \langle \mathbf{y}_i|.
	\end{equation}
Then, 
	\begin{equation*}
	\mathbf{\Phi}^{[N]}=\mathbf{D}^{[N]}\mathbf{\Phi}^{[0]}
	\end{equation*}
is a fundamental solution of the Lax pair \eqref{CNLS-lax}--\eqref{U-V-lax} 
with the new potential given by
	\begin{equation}\label{Ba}
	\mathbf{Q}^{[N]}=\mathbf{Q}^{[0]}+2\sigma_{3}\sum_{i=1}^{N}
	(\lambda_{i}-\lambda_{i}^{*})
	(|\mathbf{x}_i\rangle \langle \mathbf{y}_i|)^{off}, 
	\end{equation}
	where $|\mathbf{x}_i\rangle=(x_{1,i},x_{2,i},x_{3,i})^{T}$, 
	$|\mathbf{y}_i\rangle=(y_{1,i},y_{2,i},y_{3,i})^{T}$, $\langle x_i|=(|x_i\rangle)^{\dag}$ and $\langle y_i|=(|y_i\rangle)^{\dag}$. The vectors $|\mathbf{x}_i\rangle$ and 
	$|\mathbf{y}_i\rangle$ are related by 
	\begin{equation}
	(|\mathbf{y}_1\rangle, |\mathbf{y}_2\rangle,\cdots, |\mathbf{y}_N\rangle)
	=(|\mathbf{x}_1\rangle, |\mathbf{x}_2\rangle,\cdots, |\mathbf{x}_N\rangle)\mathbf{M}, \quad 
	\mathbf{M}=\left(
	\frac{\lambda_{i}-\lambda_{i}^{*}}{\lambda_{j}-\lambda_{i}^{*}}\langle \mathbf{y}_i|\mathbf{y}_j\rangle
	\right)_{1\leq i,j\leq N}.
	\label{matrix-M}
	\end{equation}
\end{prop}

Applying Proposition \ref{thm-DT} with $\mathbf{Q}^{[0]} = {\bf 0}$, 
$N = 2$, and $\mathbf{\Phi}^{[0]}={\rm e}^{{\rm i}\lambda\sigma_{3}(x+\lambda t)}$, we define two spectral parameters $\lambda_{1},\lambda_{2}$ 
such that $\lambda_1 \neq \lambda_2$ and the vector
\begin{equation*}
|\mathbf{y}_i\rangle=\mathbf{\Phi}^{[0]}(\lambda_{i};x,t)c^{[i]}
={\rm e}^{{\rm i}\lambda_{i}\sigma_{3}(x+\lambda_{i} t)}\begin{pmatrix}
1 \\ \ii c_{1i} \\ \ii c_{2i}
\end{pmatrix},\quad i=1,2,
\end{equation*}
where $(c_{ji})_{1 \leq i,j \leq 2}$ is a matrix of complex scattering parameters. The transformation (\ref{Ba}) produces the $2$-soliton solutions 
of the CNLS equations \eqref{CNLS}. 
To get breather solutions of Definition \ref{def-breather}, we take the spectral parameters in the form $\lambda_1 = a+{\rm i} b_1$ and $\lambda_2 = a+{\rm i} b_2$ with $b_1 \neq b_2$. Matrix $\mathbf{M}$ 
in (\ref{matrix-M}) is given explicitly as 
\begin{equation*}
\mathbf{M}=\begin{pmatrix}
{\rm e}^{2\mathrm{Re}\theta(\lambda_{1})}+d_{11}{\rm e}^{-2\mathrm{Re}\theta(\lambda_{1})} & 
\frac{2b_{1}}{b_{1}+b_{2}}({\rm e}^{\theta(\lambda_{1})^{*}+\theta(\lambda_{2})}
+d_{12}{\rm e}^{-\theta(\lambda_{1})^{*}-\theta(\lambda_{2})}) \\
\frac{2b_{2}}{b_{1}+b_{2}}({\rm e}^{\theta(\lambda_{1})+\theta(\lambda_{2})^{*}}
+d_{21}{\rm e}^{-\theta(\lambda_{1})-\theta(\lambda_{2})^{*}}) & 
{\rm e}^{2\mathrm{Re}\theta(\lambda_{2})}+d_{22}{\rm e}^{-2\mathrm{Re}\theta(\lambda_{2})}
\end{pmatrix}
\end{equation*}
where
    \begin{equation*}
        d_{ij}=\mathbf{c}_{i}^{\dagger}\mathbf{c}_{j},
    \end{equation*}
and 
\begin{align*}
\theta(\lambda_{i}) &= i \lambda_i (x+\lambda_i t) \\
&= -b_{i}(x+2at)+{\rm i}\left(a(x+2at)-(a^{2}+b_{i}^{2})t\right) \\
&= -b_{i}\xi+{\rm i}\left(a\xi-(a^{2}+b_{i}^{2})t\right)
\end{align*}
with $\xi=x+2at$. Then, $M := \det(\mathbf{M})$ is given in Definition \ref{def-breather}. Denote $\mathbf{M}=(m_{ij})_{1\leq i,j\leq 2}$. 
The Darboux matrix $\mathbf{D}^{[2]}$ in (\ref{Darboux}) is obtained in the form
\begin{equation}
\begin{split}
\mathbf{D}^{[2]} &=\mathbb{I}_{3}  -\frac{1}{M}\\
& \times \left(
\frac{2{\rm i}b_{1}\left(m_{22}|\mathbf{y}_1\rangle \langle \mathbf{y}_1|
	-m_{21}|\mathbf{y}_2\rangle \langle \mathbf{y}_1|\right)}{\lambda-a+{\rm i}b_{1}}
+\frac{2{\rm i}b_{2}    \left(-m_{12}|\mathbf{y}_1\rangle \langle \mathbf{y}_2|
	+m_{11}|\mathbf{y}_2\rangle \langle \mathbf{y}_2|\right)}{\lambda-a+{\rm i}b_{2}}
\right). 
\end{split}
\label{DT-breather}
\end{equation}
Since 
\begin{equation}\label{asy-yiyj}
|\mathbf{y}_i\rangle \langle \mathbf{y}_j|=
{\rm e}^{{\rm i}\left(a\xi-(a^{2}+b_{i}^{2})t\right) \sigma_{3}}
\begin{pmatrix}
{\rm e}^{-(b_{i}+b_{j})\xi} & -\ii c_{1j}^{*}{\rm e}^{(b_{i}-b_{j})\xi} & -\ii c_{2j}^{*}{\rm e}^{(b_{i}-b_{j})\xi} \\
\ii c_{1i}{\rm e}^{-(b_{i}-b_{j})\xi} & c_{1i}c_{1j}^{*}{\rm e}^{(b_{i}+b_{j})\xi} & c_{1i}c_{2j}^{*}{\rm e}^{(b_{i}+b_{j})\xi} \\
\ii c_{2i}{\rm e}^{-(b_{i}-b_{j})\xi} & c_{2i}c_{1j}^{*}{\rm e}^{(b_{i}+b_{j})\xi} & c_{2i}c_{2j}^{*} {\rm e}^{(b_{i}+b_{j})\xi}
\end{pmatrix}
{\rm e}^{-{\rm i}\left(a\xi-(a^{2}+b_{j}^{2})t\right) \sigma_{3}},
\end{equation}
the transformation \eqref{Ba} yields the breather solutions given by \eqref{breather}. Additionally, 
by setting \(c_{12} = c_{21} = 0\), we obtain the non-degenerate vector soliton solutions 
given by \eqref{non-sol-1}.

The fundamental matrix solution (FMS) of the Lax pair (\ref{CNLS-lax}) with the new potential $\mathbf{Q}^{[2]}$ is defined by  $\mathbf{\Phi}^{[2]}=\mathbf{D}^{[2]}\mathbf{\Phi}^{[0]}$. 
The asymptotic behavior of \(\mathbf{\Phi}^{[2]}\) is particularly useful
in Section \ref{JL2-L2}. It follows from the explicit expression for 
$\mathbf{M}$ that 
\begin{equation*}
\mathbf{M}\sim 
\begin{pmatrix}
d_{11}{\rm e}^{2b_{1}\xi} & 
\frac{2b_{1}}{b_{1}+b_{2}}d_{12}{\rm e}^{(b_{1}+b_{2})\xi-{\rm i}(b_{1}^{2}-b_{2}^{2})t} \\
\frac{2b_{2}}{b_{1}+b_{2}}d_{21}{\rm e}^{(b_{1}+b_{2})\xi+{\rm i}(b_{1}^{2}-b_{2}^{2})t} & 
d_{22}{\rm e}^{2b_{2}\xi}
\end{pmatrix},\quad x\to +\infty 
\end{equation*}
and
\begin{equation*}
\mathbf{M}\sim 
\begin{pmatrix}
{\rm e}^{-2b_{1}\xi} & 
\frac{2b_{1}}{b_{1}+b_{2}}{\rm e}^{-(b_{1}+b_{2})\xi+{\rm i}(b_{1}^{2}-b_{2}^{2})t} \\
\frac{2b_{2}}{b_{1}+b_{2}}{\rm e}^{-(b_{1}+b_{2})\xi-{\rm i}(b_{1}^{2}-b_{2}^{2})t} & 
{\rm e}^{-2b_{1}\xi}
\end{pmatrix},\quad x\to -\infty,
\end{equation*}
which implies that the determinant $M = \det(\mathbf{M})$ satisfies
\begin{equation}
\label{asy-M}
M  \sim \left\{ \begin{array}{ll} M^{+} {\rm e}^{2(b_{1}+b_{2})\xi}, \quad & x\to +\infty \\
M^{-} {\rm e}^{-2(b_{1}+b_{2})\xi} ,\quad & x\to -\infty, 
\end{array} \right.
\end{equation}
where 
\begin{equation*}
M^{+}=d_{11}d_{22}-\frac{4b_{1}b_{2}}{(b_{1}+b_{2})^{2}}d_{12}d_{21}, \quad 
M^{-}=\frac{(b_{1}-b_{2})^{2}}{(b_{1}+b_{2})^{2}}. 
\end{equation*}
Combing \eqref{asy-yiyj} and \eqref{asy-M}  with \eqref{DT-breather}, we define 
\begin{equation*}
    \begin{split}
    D_{i1}=d_{11}c_{i2}-\frac{2 b_{1}}{b_{1}+b_{2}}d_{12}c_{i1}, \quad
    D_{i2}=d_{22}c_{i1}-\frac{2 b_{2}}{b_{1}+b_{2}}d_{21}c_{i2}
    \end{split}
\end{equation*}
and 
\begin{equation*}
l_{ij}^{+}(\lambda)=2{\rm i}b_{1}D_{i2}c_{j1}^{*}(\lambda-\lambda_{2}^{*})
+2{\rm i}b_{2}D_{i1}c_{j2}^{*}(\lambda-\lambda_{1}^{*}),
\end{equation*}
which gives the following asymptotic representation of the 
Darboux matrix 
\begin{equation*}
\mathbf{D}^{[2]} \sim \mathbf{D}^{[2]}_{\pm \infty} \quad \mbox{\rm as} \;\; 
x \to \pm \infty,
\end{equation*} 
where
\begin{equation*} 
\mathbf{D}^{[2]}_{+\infty} =
\begin{pmatrix}
1 & 0 & 0\\
0 & 1-\frac{l_{11}^{+}(\lambda)}{M^{+}(\lambda-\lambda_{1}^{*})(\lambda-\lambda_{2}^{*})} & -\frac{l_{12}^{+}(\lambda)}{M^{+}(\lambda-\lambda_{1}^{*})(\lambda-\lambda_{2}^{*})} \\
0 & -\frac{l_{21}^{+}(\lambda)}{M^{+}(\lambda-\lambda_{1}^{*})(\lambda-\lambda_{2}^{*})} &1-\frac{l_{22}^{+}(\lambda)}{M^{+}(\lambda-\lambda_{1}^{*})(\lambda-\lambda_{2}^{*})}
\end{pmatrix}
\end{equation*}
and
\begin{equation*}
\mathbf{D}^{[2]}_{-\infty} =
\begin{pmatrix}
\frac{(\lambda-\lambda_{1})(\lambda-\lambda_{2})}{(\lambda-\lambda_{1}^{*})(\lambda-\lambda_{2}^{*})} & 0 & 0\\
0 & 1 & 0 \\
0 & 0 & 1
\end{pmatrix}. 
\end{equation*}
Let us define two FMSs $\mathbf{\Phi}^{\pm}$ of the Lax pair \eqref{CNLS-lax}--\eqref{U-V-lax} for the breather solutions by  
\begin{equation*}
\mathbf{\Phi}^{+}=\mathbf{\Phi}^{[2]}(\mathbf{D}^{[2]}_{+\infty})^{-1},\quad 
\mathbf{\Phi}^{-}=\mathbf{\Phi}^{[2]}(\mathbf{D}^{[2]}_{-\infty})^{-1}.
\end{equation*}
The two FMSs satisfy
\begin{equation*}
\mathbf{\Phi}^{\pm}\sim {\rm e}^{{\rm i}\lambda (x+\lambda t)\sigma_{3}}, \quad 
x\to \pm \infty 
\end{equation*}
and appear to be important in the inverse scattering transform 
(see Appendix \ref{A.ISM}). 
The transfer matrix between the two FMSs is given by
\begin{align}
\mathbf{S}(\lambda) &= \mathbf{D}^{[2]}_{+\infty}(\mathbf{D}^{[2]}_{-\infty})^{-1} \notag \\
&=\begin{pmatrix}
\frac{(\lambda-\lambda_{1}^{*})(\lambda-\lambda_{2}^{*})}{(\lambda-\lambda_{1})(\lambda-\lambda_{2})} & 0 & 0\\
0 & 1-\frac{l_{11}^{+}(\lambda)}{M^{+}(\lambda-\lambda_{1}^{*})(\lambda-\lambda_{2}^{*})} & \frac{l_{12}^{+}(\lambda)}{M^{+}(\lambda-\lambda_{1}^{*})(\lambda-\lambda_{2}^{*})} \\
0 & \frac{l_{21}^{+}(\lambda)}{M^{+}(\lambda-\lambda_{1}^{*})(\lambda-\lambda_{2}^{*})} &1-\frac{l_{22}^{+}(\lambda)}{M^{+}(\lambda-\lambda_{1}^{*})(\lambda-\lambda_{2}^{*})}
\end{pmatrix}. 
\label{S-matrix-non-2}
\end{align}
The (1,1) element of the transfer matrix \(\mathbf{S}(\lambda)\) is useful for calculating the closure relation of squared eigenfunctions.
 
\subsection{Construction of squared eigenfunctions}

Recall that the FMS matrix $\mathbf{\Phi}=(\phi_{ij})_{1\leq i,j \leq 3}$ satisfies the spectral problem 
\begin{equation}\label{spectral-problem}
\mathbf{\Phi}_{x}(\lambda;x,t)=\mathbf{U}(\lambda,\mathbf{q})\mathbf{\Phi}(\lambda;x,t), 
\end{equation}
with the symmetry $\mathbf{Q}=\mathbf{Q}^{\dagger}$. It follows from  Appendix \ref{A.ISM} that the inverse matrix 
$\mathbf{\Phi}^{-1}(\lambda)=\mathbf{\Phi}^{\dagger}(\lambda^{*})=(\hat{\phi}_{ij})_{1\leq i,j \leq 3}$ satisfies 
\begin{equation}\label{spectral-problem-inverse}
(\mathbf{\Phi}^{-1})_{x}(\lambda;x,t)=-(\mathbf{\Phi}^{-1})(\lambda;x,t)\mathbf{U}(\lambda,\mathbf{q}). 
\end{equation}
The squared eigenfunctions of the Lax pair \eqref{CNLS-lax}--\eqref{U-V-lax} are constructed from $\mathbf{\Phi}(\lambda)$ and $\mathbf{\Phi}^{\dagger}(\lambda^{*})$. 
Let us denote the $i$-th column vector of matrix $\mathbf{\Phi}$ by $(\mathbf{\Phi})_{i}$  and its $i$-th row vector by 
$(\mathbf{\Phi})^{i}$. The squared eigenfunctions are defined by 
\begin{equation}\label{squared-eigenfunctions}
\begin{split}
&s_{1}(\mathbf{\Phi})=
\begin{pmatrix}
\phi_{21}\hat{\phi}_{12} &
\phi_{31}\hat{\phi}_{12} &
-\phi_{11}\hat{\phi}_{22}&
-\phi_{11}\hat{\phi}_{32}
\end{pmatrix}^{T}, \\
&s_{2}(\mathbf{\Phi})=
\begin{pmatrix}
\phi_{21}\hat{\phi}_{13} &
\phi_{31}\hat{\phi}_{13} &
-\phi_{11}\hat{\phi}_{23}&
-\phi_{11}\hat{\phi}_{33}
\end{pmatrix}^{T},\\
&s_{-1}(\mathbf{\Phi})=\begin{pmatrix}
\phi_{22}\hat{\phi}_{11} &
\phi_{32}\hat{\phi}_{11} &
-\phi_{12}\hat{\phi}_{21} &
-\phi_{12}\hat{\phi}_{31}
\end{pmatrix}^{T}, \\
&s_{-2}(\mathbf{\Phi})=\begin{pmatrix}
\phi_{23}\hat{\phi}_{11} &
\phi_{33}\hat{\phi}_{11} &
-\phi_{13}\hat{\phi}_{21} &
-\phi_{13}\hat{\phi}_{31}
\end{pmatrix}^{T}
\end{split}
\end{equation}
and the squared eigenfunction matrices are defined by 
\begin{equation}\label{squared-eigenfunction-matrix}
\begin{split}
&p_{1}(\mathbf{\Phi})(\lambda)=(\mathbf{\Phi}(\lambda))_{1}(\mathbf{\Phi}^{\dagger}(\lambda^{*}))^{2}, \\
&p_{2}(\mathbf{\Phi})(\lambda)=(\mathbf{\Phi}(\lambda))_{1}(\mathbf{\Phi}^{\dagger}(\lambda^{*}))^{3}, \\
&p_{-1}(\mathbf{\Phi})(\lambda)=(\mathbf{\Phi}(\lambda))_{2}(\mathbf{\Phi}^{\dagger}(\lambda^{*}))^{1}, \\
&p_{-2}(\mathbf{\Phi})(\lambda)=(\mathbf{\Phi}(\lambda))_{3}(\mathbf{\Phi}^{\dagger}(\lambda^{*}))^{1}.
\end{split}
\end{equation}
The elements of the squared eigenfunctions \(s_{i}(\mathbf{\Phi})\) 
are found in the off-diagonal entries of the matrix \(p_{i}(\mathbf{\Phi})\) since we have for $i=1,2$, 
\begin{align*}
p_{i}(\mathbf{\Phi}) &= \begin{pmatrix}
\star & \phi_{11}\hat{\phi}_{2,i+1} &\phi_{11}\hat{\phi}_{3,i+1} \\
\phi_{21}\hat{\phi}_{1,i+1} & \star & \star \\
\phi_{31}\hat{\phi}_{1,i+1} & \star & \star
\end{pmatrix}, \\
p_{-i}(\mathbf{\Phi}) &= \begin{pmatrix}
\star & \phi_{1,i+1}\hat{\phi}_{21} &\phi_{1,i+1}\hat{\phi}_{21} \\
\phi_{2,i+1}\hat{\phi}_{11} & \star & \star \\
\phi_{3,i+1}\hat{\phi}_{11} & \star & \star 
\end{pmatrix}.
\end{align*}
The symmetry condition \(\mathbf{Q} = \mathbf{Q}^{\dagger}\) for the potential induces symmetry between different squared eigenfunction matrices
\begin{equation}
\label{symmetric-squared-eigenfunction-matrix}
p_{i}(\mathbf{\Phi})(\lambda)=p_{-i}(\mathbf{\Phi})^{\dagger}(\lambda^{*}), \quad 
i = 1,2.
\end{equation}
Hence the squared eigenfunctions satisfy the following symmetry,
\begin{equation}\label{symmetric-squared-eigenfunctions}
s_{i}(\mathbf{\Phi})(\lambda)=-\mathbf{\Sigma} (s_{-i}(\mathbf{\Phi})(\lambda^{*}))^{*}, \qquad 
\mathbf{\Sigma}=\begin{pmatrix}
\mathbf{0}_{2\times 2} & \mathbb{I}_{2}\\
\mathbb{I}_{2} & \mathbf{0}_{2\times 2}
\end{pmatrix}. 
\end{equation}
Now we consider the differential equations satisfied by the squared eigenfunction matrix. According to equations \eqref{spectral-problem} and \eqref{spectral-problem-inverse} and symmetry \eqref{symmetric-squared-eigenfunction-matrix}, 
the squared eigenfunction matrix \(p_{i}(\mathbf{\Phi})(\lambda)\)
satisfies for \(i \in \{\pm 1, \pm 2\}\), 
\begin{equation}\label{diff-L}
\mathbf{F}_{x}(\lambda)=[\mathbf{U}(\lambda),\mathbf{F}(\lambda)],
{\quad \mathbf{G}_{x}(\lambda)=-[\mathbf{G}(\lambda),\mathbf{U}(\lambda)]. }
\end{equation}
The differential equations (\ref{diff-L}) are useful 
for calculating the orthogonality conditions between squared eigenfunctions and adjoint squared eigenfunctions in Section 4.4.

For different solutions, we need to choose different FMS \(\mathbf{\Phi}\) to 
construct squared eigenfunctions. We require that \(\mathbf{\Phi}\) has no singularities 
at the point spectrum of the Lax pair and that the squared eigenfunctions constructed 
from \(\mathbf{\Phi}\) are non-zero at the point spectrum. The exact forms of the squared 
eigenfunctions for non-degenerate vector soliton solutions and breather solutions will be examined in Sections \ref{sec-3} and \ref{JL2-L2}.

\subsection{Conserved energies for Lyapunov functional $\mathcal{I}_2$}

The Lyapunov functional $\mathcal{I}_2$ given by (\ref{Ly-non}) is time-independent since it is a linear combination of the conserved quantities 
(\ref{con-0}), (\ref{con-02}), (\ref{con-03}), (\ref{con-04}), and (\ref{con-95}). It follows from Appendix \ref{A.ISM} that 
the conserved quantities for breather solutions can be expressed explicitly 
in terms of the spectral parameters $\lambda_1, \lambda_2\in \mathbb{C}$:
\begin{equation*}
H_{n-1}=\frac{2^{n}}{n}
\mathrm{Im} \left(\lambda_{1}^{n}+\lambda_{2}^{n}\right),\quad n\geq 1.
\end{equation*}
Variations of the conserved quantities with respect to \(\mathbf{q}\) are given by 
\begin{equation*}
    \frac{\delta H_{n-1}}{\delta \mathbf{q}}= -2^{n-1}{\rm i}\left(
        \lambda_{1}^{n-1}\frac{\delta \lambda_{1}}{\delta \mathbf{q}}-
        (\lambda_{1}^{*})^{n-1}\frac{\delta \lambda_{1}^{*}}{\delta \mathbf{q}}+
        \lambda_{2}^{n-1}\frac{\delta \lambda_{2}}{\delta \mathbf{q}}-
        (\lambda_{2}^{*})^{n-1}\frac{\delta \lambda_{2}^{*}}{\delta \mathbf{q}}
    \right)
\end{equation*}
and involve \(\frac{\delta \lambda_{i}}{\delta \mathbf{q}}\) and 
\(\frac{\delta \lambda_{i}^{*}}{\delta \mathbf{q}}\), with coefficients that are polynomials 
in \(\lambda_{i}\) and \(\lambda_{i}^{*}\). Consequently, a special linear 
combination of these variations can vanish. The following proposition follows from a more general Proposition \ref{prop-app-A} with $N = 2$ proven in Appendix \ref{app-A-1}. 

\begin{prop}
	Define the polynomial
	\begin{equation*}
	\mathcal{P}_{2}(\lambda)=(\lambda-\lambda_{1})(\lambda-\lambda_{1}^{*})
	(\lambda-\lambda_{2})(\lambda-\lambda_{2}^{*})=\sum_{n=0}^{4}2^{n-4}\mu_{n}
	\lambda^{n}
	\end{equation*}
	where $\mu_{n}$ are the elementary symmetric polynomials. 
	The breather solutions correspond to the critical points of the Lyapunov functional
	\begin{equation}\label{Ly-non-1}
	\mathcal{I}_{2}=\sum_{n=0}^{4}\mu_{n}H_{n}. 
	\end{equation}
\end{prop}

For the spectral parameters $\lambda_{1}=a+{\rm i}b_{1}, \lambda_{2}=a+{\rm i}b_{2}$, we have 
\begin{align*}
\mu_{0} &= 16(a^{2}+b_{1}^{2})(a^{2}+b_{2}^{2}), \\
\mu_{1} &= -16 a(2a^{2}+b_{1}^{2}+b_{2}^{2}), \\
\mu_{2} &=4(6a^{2}+b_{1}^{2}+b_{2}^{2}), \\
\mu_{3} &=-8a, \\
\mu_4 &= 1,
\end{align*} 
from which the Lyapunov functional \eqref{Ly-non-1} coincides with the explicit expression \eqref{Ly-non}. Moreover, the trace formula reveals that
\begin{equation}\label{con-law-spectral}
\begin{split}
&H_{0}
=2(b_{1}+b_{2}), \\
&H_{1}
=4a(b_{1}+b_{2}), \\
&H_{2}=8a^{2}(b_{1}+b_{2})-\frac{8}{3}(b_{1}^{3}+b_{2}^{3}), \\
&H_{3}=16a^{3}(b_{1}+b_{2})-16a(b_{1}^{3}+b_{2}^{3}), \\
&H_{4}=32a^{4}(b_{1}+b_{2})-64a^{2}(b_{1}^{3}+b_{2}^{3})+
\frac{32}{5}(b_{1}^{5}+b_{2}^{5}),
\end{split}
\end{equation}
from which we obtain
\begin{equation*}
\mathcal{I}_{2}=-\frac{64}{15}(b_{1}+b_{2})^{3}(b_{1}^{2}-3b_{1}b_{2}+b_{2}^{2}).
\end{equation*}

Now we turn to variation of conserved quantities in terms of the functions \(\mathbf{q}\). We will derive the explicit form 
of the ODE satisfied by the breather profile, 
as well as the operator \(\mathcal{L}_{2}\) defined in \eqref{L2}. 
The first variations are given by
\begin{equation*}
%\label{con-1}
\begin{split}
\frac{\delta H_{0}}{\delta \mathbf{q}}&=\mathbf{q}, \\
\frac{\delta H_{1}}{\delta \mathbf{q}}&={\rm i}\mathbf{q}_{x}, \\
\frac{\delta H_{2}}{\delta \mathbf{q}} &=-\left(\mathbf{q}_{xx}+2|\mathbf{q}|^{2}\mathbf{q}\right), \\
\frac{\delta H_{3}}{\delta \mathbf{q}} &=
-{\rm i}\left(\mathbf{q}_{(3x)}+3|\mathbf{q}|^{2}\mathbf{q}_{x}+
3\mathbf{q}\mathbf{q}^{\dagger}\mathbf{q}_{x}\right),\\
\frac{\delta H_{4}}{\delta \mathbf{q}}&=
\mathbf{q}_{(4x)}+4|\mathbf{q}|^{2}\mathbf{q}_{xx}+
2\mathbf{q}\mathbf{q}_{xx}^{\dagger}\mathbf{q}+
4\mathbf{q}\mathbf{q}^{\dagger}\mathbf{q}_{xx}+
2\mathbf{q}_{x}\mathbf{q}_{x}^{\dagger}\mathbf{q} +
6 \mathbf{q}_{x}\mathbf{q}^{\dagger}\mathbf{q}_{x}  \\ &\qquad +
2|\mathbf{q}_{x}|^{2}\mathbf{q}+
6|\mathbf{q}|^{4}\mathbf{q}.
\end{split}
\end{equation*}
Hence the breather profile satisfies the system of two fourth-order ODEs given by 
$$
\sum_{n=0}^{4}\mu_{n}\frac{\delta H_{n}}{\delta \mathbf{q}}=0.
$$
The second variations are given by
\begin{equation*}
%\label{con-2}
\begin{split}
\frac{\delta^{2} H_{0}}{\delta \mathbf{q}^{2}}&=\mathbb{I}_{2}, \\
\frac{\delta^{2} H_{1}}{\delta \mathbf{q}^{2}} &= {\rm i}\partial_{x}, \\
\frac{\delta^{2} H_{2}}{\delta \mathbf{q}^{2}} &=
-\left(\partial_{x}^{2}+2|\mathbf{q}|^{2}+
2\mathbf{q}\otimes \mathbf{q}^{*}+2\mathbf{q}\otimes \mathbf{q}\cdot^{*}\right), \\
\frac{\delta^{2} H_{3}}{\delta \mathbf{q}^{2}}&=
-{\rm i}\left(\partial_{x}^{3}+3|\mathbf{q}|^{2}\partial_{x}+
3\mathbf{q}_{x}\otimes \mathbf{q}^{*}+3\mathbf{q}_{x}\otimes \mathbf{q}\cdot^{*}
+3\mathbf{q}^{\dagger}\mathbf{q}_{x}
+3\mathbf{q}\otimes \mathbf{q}^{*}\partial_{x}+3\mathbf{q}\otimes \mathbf{q}_{x}\cdot^{*}\right),\\
\frac{\delta^{2} H_{4}}{\delta \mathbf{q}^{2}}&=
\partial_{x}^{4}+4|\mathbf{q}|^{2}\partial_{x}^{2}+
4\mathbf{q}_{xx}\otimes \mathbf{q}^{*}+
4\mathbf{q}_{xx}\otimes \mathbf{q}\cdot^{*}+
2\mathbf{q}_{xx}^{\dagger}\mathbf{q}+
2\mathbf{q}\otimes \mathbf{q}\partial_{x}^{2}\cdot^{*}+
2\mathbf{q}\otimes \mathbf{q}_{xx}^{*}
\\&+
4\mathbf{q}^{\dagger}\mathbf{q}_{xx}+
4\mathbf{q}\otimes \mathbf{q}^{*}\partial_{x}^2+
4\mathbf{q}\otimes \mathbf{q}_{xx} \cdot^{*}+
2\mathbf{q}_{x}^{\dagger}\mathbf{q}\partial_{x}+
2\mathbf{q}_{x}\otimes \mathbf{q}_{x}^{\dagger}+
2\mathbf{q}_{x}\otimes \mathbf{q} \partial_{x}\cdot^{*}
\\&  +
6 \mathbf{q}^{\dagger}\mathbf{q}_{x} \partial_{x}+
6  \mathbf{q}_{x}\otimes \mathbf{q}_{x} \cdot^{*}+
6  \mathbf{q}_{x}\otimes \mathbf{q}^{\dagger} \partial_{x}+
2|\mathbf{q}_{x}|^{2}+
2\mathbf{q}\otimes \mathbf{q}_{x}^{*}\partial_{x}+
2\mathbf{q}\otimes \mathbf{q}_{x}\partial_{x}\cdot^{*}\\&+
6|\mathbf{q}|^{4}+
12|\mathbf{q}|^{2}\mathbf{q}\otimes \mathbf{q}^{*}+
12|\mathbf{q}|^{2}\mathbf{q}\otimes \mathbf{q}\cdot^{*} 
,
\end{split}
\end{equation*}
where the operation $\otimes$ is defined by $\mathbf{q}\otimes\mathbf{r}\equiv \mathbf{q} \mathbf{r}^{{\rm T}}$ and $\cdot^{*}$ is a multiplication with complex conjugation. For instance, $\frac{\delta^{2} H_{2}}{\delta \mathbf{q}^{2}}$ acts on an element of $\mathrm{X}$ as follows:
\begin{equation*}
    \frac{\delta^{2} H_{2}}{\delta \mathbf{q}^{2}}\begin{pmatrix}
        \mathbf{v} \\
        \mathbf{v}^{*}
    \end{pmatrix}=
    -\begin{pmatrix}
        \partial_{x}^{2}\mathbf{v}+2|\mathbf{q}|^{2}\mathbf{v}+
2\mathbf{q}\otimes \mathbf{q}^{*}\mathbf{v}+2\mathbf{q}\otimes \mathbf{q}\mathbf{v}^{*}\\
\partial_{x}^{2}\mathbf{v}^{*}+2|\mathbf{q}|^{2}\mathbf{v}^{*}+
2\mathbf{q}^{*}\otimes \mathbf{q}\mathbf{v}^{*}+2\mathbf{q}^{*}\otimes \mathbf{q}^{*}\mathbf{v}
    \end{pmatrix}.
\end{equation*}
The linear operator $\mathcal{L}_2$ in (\ref{L2}) is expressed by the fourth-order operator
$$
\mathcal{L}_2 := \sum_{n=0}^{4}\mu_{n}\frac{\delta^2 H_{n}}{\delta \mathbf{q}^2}.
$$
Let \(\mathcal{D}(\mathcal{A})\) denote the domain of an operator \(\mathcal{A}\) in $\mathrm{X}$. Then, we have the following property.

\begin{prop}
	The linear operator $\mathcal{L}_{2}$ is self-adjoint operator acting on the space $\mathrm{X}\cap \mathcal{D}(\mathcal{L}_{2})$. 
\end{prop}

\begin{proof}
	It suffices to prove that \(\frac{\delta^{2} H_{n}}{\delta \mathbf{q}^{2}}\) is 
	self-adjoint. We demonstrate this by showing that \(\frac{\delta^{2} H_{2}}{\delta \mathbf{q}^{2}}\) 
	is self-adjoint, and the same method can be applied to the other cases. 
	It is evident that \(\partial_{x}^{2}\) and \(|\mathbf{q}|^{2}\) are self-adjoint. 
	We consider the operator \(\mathbf{q} \otimes \mathbf{q}^{*}\) in the space \(\mathrm{X}\), then
	\begin{equation*}
	\begin{split}
	\mathrm{Re}\int_{\mathbb{R}}\left(\begin{pmatrix}
	\mathbf{q}\otimes \mathbf{q}^{*} & 0 \\
	0 & \mathbf{q}^{*}\otimes \mathbf{q}
	\end{pmatrix}
	\mathbf{f}\right)^{\dagger}
	\mathbf{g}\mathrm{d}x=\mathrm{Re}\int_{\mathbb{R}}
	\mathbf{f}^{\dagger}
	\begin{pmatrix}
	\mathbf{q}\otimes \mathbf{q}^{*} & 0 \\
	0 & \mathbf{q}^{*}\otimes \mathbf{q}
	\end{pmatrix}\mathbf{g}\mathrm{d}x. 
	\end{split}
	\end{equation*}
	Similarly, the operator $\mathbf{q}\otimes \mathbf{q}\cdot^{*}$
	\begin{equation*}
	\begin{split}
	\mathrm{Re}\int_{\mathbb{R}}\left(\begin{pmatrix}
	0 & \mathbf{q}\otimes \mathbf{q} \\
	\mathbf{q}^{*}\otimes \mathbf{q}^{*} & 0
	\end{pmatrix}
	\mathbf{f}\right)^{\dagger}
	\mathbf{g}\mathrm{d}x=\mathrm{Re}\int_{\mathbb{R}}
	\mathbf{f}^{\dagger}
	\begin{pmatrix}
	0 & \mathbf{q}\otimes \mathbf{q} \\
	\mathbf{q}^{*}\otimes \mathbf{q}^{*} & 0
	\end{pmatrix}\mathbf{g}\mathrm{d}x. 
	\end{split}
	\end{equation*}
Similarly, one can show that \(\frac{\delta^{2} H_{3}}{\delta \mathbf{q}^{2}}\)  and \(\frac{\delta^{2} H_{4}}{\delta \mathbf{q}^{2}}\)  are self-adjoint. 
\end{proof}

Since \(\mathcal{L}_{2}\) is self-adjoint, the spectrum of \(\mathcal{L}_{2}\) in the space \(\mathrm{X}\) is a subset of \(\mathbb{R}\). 
The operator \(\mathcal{L}_{2}\) can be viewed as a perturbation of a linear operator
\begin{equation*}
\mathcal{L}_{\infty}=\sum_{n=0}^{4}\mu_{n}({\rm i}\partial_{x})^{n}=2^{4}\mathcal{P}_2\left(\frac{{\rm i}\partial_{x}}{2} \right).
\end{equation*}
The perturbation depends on \(\mathbf{q}\) and its derivatives, which belong for the breather solutions to the Schwartz class. Consequently, \(\mathcal{L}_{2}\) is a relatively compact perturbation of the operator 
\(\mathcal{L}_{\infty}\). The essential spectrum of \(\mathcal{L}_{2}\) coincides with the spectrum of \(\mathcal{L}_{\infty}\) by Weyl's essential spectrum theorem. Since 
\begin{equation*}
\mathcal{P}_2(\lambda)=|\lambda-\lambda_{1}|^{2}|\lambda-\lambda_{2}|^{2}>0,\quad \lambda\in\mathbb{R} 
\end{equation*}
for every $\lambda_{1},\lambda_{2} \in \mathbb{C} \backslash \{ \mathbb{R} \}$,  the essential spectrum of $\mathcal{L}_{2}$ is strictly positive and is bounded away from zero.

\section{Spectral stability of vector solitons}
\label{sec-3}

We examine the spectral stability of non-degenerate vector solitons and prove Theorem \ref{spectral-stability}. According to the expression \eqref{non-sol-1}, the variables in non-degenerate vector soliton solutions 
can be separated into functions of $\xi := x + 2 a t$ and $t$ respectively
\begin{equation}\label{q-non2-spearate}
\mathbf{q}^{[2]}_{non}={\rm e}^{\mathbf{\Pi} t}
\hat{\mathbf{q}}(\xi),
\end{equation}
where 
\begin{equation*}
\mathbf{\Pi}=2{\rm i}\begin{pmatrix}
|\lambda_{1}|^2 & 0 \\
0 & |\lambda_{2}|^2
\end{pmatrix},\quad
\hat{\mathbf{q}}(\xi)=
\frac{4}{\mathrm{det}(\mathbf{M})}
\begin{pmatrix}
b_{1}c_{11}{\rm e}^{-2{\rm i} a \xi}
\left(|c_{22}|^{2} {\rm e}^{2b_{2}\xi}+\frac{b_{1}-b_{2}}{b_{1}+b_{2}} 
 {\rm e}^{-2b_{2}\xi} \right) \\
b_{2}c_{22}{\rm e}^{-2{\rm i} a \xi}
\left(|c_{11}|^{2}{\rm e}^{2b_{1}\xi}+\frac{b_{2}-b_{1}}{b_{1}+b_{2}} 
 {\rm e}^{-2b_{1}\xi} \right)
\end{pmatrix},
\end{equation*}
where $\lambda_1 = a + {\rm i} b_1$ and $\lambda_2 = a + {\rm i} b_2$.
Substituting \eqref{q-non2-spearate} into \eqref{CNLS}, CNLS equations \eqref{CNLS} can 
be reduced to equations involving \(\hat{\mathbf{q}}(\xi)\) that are independent of \(t\)
\begin{equation}\label{ODE-non}
-\frac{1}{2}\hat{\mathbf{q}}_{xx}-2{\rm i} a \hat{\mathbf{q}}_{x}
+2
\begin{pmatrix}
|\lambda_{1}|^{2} & 0\\
0 & |\lambda_{2}|^{2}
\end{pmatrix}
\hat{\mathbf{q}}-|\hat{\mathbf{q}}|^2\hat{\mathbf{q}}=0,
\end{equation} 
which coincides with system (\ref{CNLS-trav}). The spectral stability of the non-degenerate vector soliton solutions is examined by considering 
the perturbation
\begin{equation*}
%\label{perturbation-2}
\mathbf{q}
=
{\rm e}^{\mathbf{\Pi} t}
(\hat{\mathbf{q}}(\xi)+\epsilon 
(\mathbf{p}_{1}(\xi){\rm e}^{\Omega t}+\mathbf{p}_{2}^{*}(\xi){\rm e}^{\Omega^{*} t})).
\end{equation*}
The perturbations $\mathbf{P}$ satisfy
\begin{equation}\label{linear-operator-2-1}
-{\rm i}\begin{pmatrix}
\mathbb{I}_{2} & 0\\
0 & -\mathbb{I}_{2}
\end{pmatrix}
\begin{pmatrix}
\mathcal{M}_{1} & \mathcal{M}_{2} \\
\mathcal{M}_{2}^{*} & \mathcal{M}_{1}^{*}
\end{pmatrix}
\mathbf{P}
=
\Omega
\mathbf{P}, \quad \mathbf{P}=\begin{pmatrix}
\mathbf{p}_{1}\\\mathbf{p}_{2}
\end{pmatrix}, 
\end{equation}
where 
\begin{equation*}
\mathcal{M}_{1}(\hat{\mathbf{q}})=\mathcal{L}_{0}(\hat{\mathbf{q}})-\hat{\mathbf{q}}\otimes \hat{\mathbf{q}}^{*},\quad 
\mathcal{M}_{2}(\hat{\mathbf{q}})=-\hat{\mathbf{q}}\otimes \hat{\mathbf{q}}
\end{equation*}
and 
\begin{equation*}
\mathcal{L}_{0}(\hat{\mathbf{q}}) = - \frac{1}{2} \partial_{x}^{2} - 
2 {\rm i} a \partial_{x}
+ 2 \begin{pmatrix}
|\lambda_{1}|^{2} & 0\\
0 & |\lambda_{2}|^{2}
\end{pmatrix} - |\mathbf{\hat{\mathbf{q}}}|^2. 
\end{equation*}
The variation of \eqref{ODE-non} lead to the operator $\mathcal{L}_{1}$ acting 
of $(\mathbf{v}^{T},\mathbf{v}^{\dagger})^{T} \in \mathbb{C}^4$ as follows:
\begin{equation*}
\mathcal{L}_{1}\begin{pmatrix}
    \mathbf{v} \\
    \mathbf{v}^{*}
\end{pmatrix}
=\begin{pmatrix}
    \mathcal{L}_{0}\mathbf{v}-
(\hat{\mathbf{q}}\otimes \hat{\mathbf{q}}^{*})\mathbf{v}-
(\hat{\mathbf{q}}\otimes \hat{\mathbf{q}})\mathbf{v}^{*}\\
\mathcal{L}_{0}^{*}\mathbf{v}^{*}-
(\hat{\mathbf{q}}^{*}\otimes \hat{\mathbf{q}})\mathbf{v}^{*}-
(\hat{\mathbf{q}}^{*}\otimes \hat{\mathbf{q}}^{*})\mathbf{v}
\end{pmatrix}. 
\end{equation*}
Note that $\mathcal{L}_{1}$ is self-adjoint since 
\begin{equation*}
    \mathcal{L}_{1}(\hat{\mathbf{q}})=\frac{1}{2} 
        \frac{\delta^{2} H_{2}}{\delta \mathbf{q}^{2}}
        -2a\, \frac{\delta^{2} H_{1}}{\delta \mathbf{q}^{2}}
        +2\, \mathrm{diag}(|\lambda_{1}|^{2},|\lambda_{2}|^{2},|\lambda_{1}|^{2},|\lambda_{2}|^{2})\frac{\delta^{2} H_{0}}{\delta \mathbf{q}^{2}}.
\end{equation*}
We consider $\mathcal{J}$ and $\mathcal{L}_{1}$ in the space 
$\mathrm{X}$, then the two operators have representation
\begin{equation*}
\mathcal{J}=-{\rm i}\begin{pmatrix}
\mathbb{I}_{2} & 0\\
0 & -\mathbb{I}_{2}
\end{pmatrix},\quad \mathcal{L}_{1}=
\begin{pmatrix}
\mathcal{M}_{1} & \mathcal{M}_{2} \\
\mathcal{M}_{2}^{*} & \mathcal{M}_{1}^{*}
\end{pmatrix}. 
\end{equation*}
The equation \eqref{linear-operator-2-1} becomes
\begin{equation}\label{spectral-JL}
\mathcal{J}\mathcal{L}_{1}\mathbf{P}=
\Omega\mathbf{P}. 
\end{equation}
The stability spectrum is defined by 
\begin{equation*}
\sigma_{s}(\mathcal{J}\mathcal{L}_{1})=\{
\Omega\in\mathbb{C}: \quad \mathbf{P}\in L^{\infty}
\}.
\end{equation*}
Spectral stability is defined by the condition 
\(\sigma_{s}(\mathcal{J}\mathcal{L}_{1}) \subset {\rm i}\mathbb{R}\). 
Our goal is to determine the spectrum of \(\mathcal{J}\mathcal{L}_{1}\) by constructing the squared eigenfunctions that satisfy the spectral problem \eqref{spectral-JL}. 

\begin{rem}
Since the range of the function \(\mathbf{q}\) is \(\mathbb{C}^2\), 
it is important to note that \(\mathbf{p}_{1}\) and \(\mathbf{p}_{2}\) 
are linearly independent. Unlike previous work, where the perturbation was considered as 
\(\mathbf{p} = (\mathbf{p}_{1}(\xi) + {\rm i} \mathbf{p}_{2}(\xi)){\rm e}^{\Omega t}\), 
we now consider 
\(\mathbf{p} = \mathbf{p}_{1}(\xi){\rm e}^{\Omega t} + \mathbf{p}_{2}^{*}(\xi){\rm e}^{\Omega^{*} t}\). 
This approach takes into account the linearized operator that the squared eigenfunctions satisfy.
\end{rem}

We define 
\begin{equation*}
\mathbf{\Phi}^{[2]}_{r,non}(\lambda;\xi) = \mathbf{\Phi}^{[2]}(\lambda;\xi,0) \left(\begin{matrix}
(\lambda-\lambda_{1}^{*})(\lambda-\lambda_{2}^{*}) & 0 & 0 \\
0 & \lambda-\lambda_{1}^{*} & 0 \\ 0 & 0 & \lambda-\lambda_{2}^{*} \end{matrix}
\right)
\end{equation*}
to eliminate the singularity with respect to $\lambda$ and consider the squared eigenfunctions
\begin{equation*}
\hat{\mathbf{P}}_{\pm 1}(\lambda;\xi) = s_{\pm 1}(\mathbf{\Phi}^{[2]}_{r,non}(\lambda;\xi)), \qquad  
\hat{\mathbf{P}}_{\pm 2}(\lambda;\xi) = s_{\pm 2}(\mathbf{\Phi}^{[2]}_{r,non}(\lambda;\xi)).
\end{equation*}
Denote $\mathbf{\Phi}^{[2]}_{r,non}=(Q_{ij})_{1\leq i,j\leq 3}$, the squared eigenfunctions have exact representation
\begin{equation}\label{eigenfunction-JL1}
\begin{split}
\hat{\mathbf{P}}_{i-1} =
\begin{pmatrix}
Q_{21}(\lambda)Q_{1i}(\lambda^{*})^{*} &
Q_{31}(\lambda)Q_{1i}(\lambda^{*})^{*} &
-Q_{11}(\lambda)Q_{2i}(\lambda^{*})^{*} &
-Q_{11}(\lambda)Q_{3i}(\lambda^{*})^{*}
\end{pmatrix}^{T}, 
\\
\hat{\mathbf{P}}_{-i+1}=
\begin{pmatrix}
Q_{2i}(\lambda)Q_{11}(\lambda^{*})^{*} &
Q_{3i}(\lambda)Q_{11}(\lambda^{*})^{*} &
-Q_{1i}(\lambda)Q_{21}(\lambda^{*})^{*} &
-Q_{1i}(\lambda)Q_{31}(\lambda^{*})^{*}
\end{pmatrix}^{T},
\end{split}
\end{equation}
for $i=2,3$.
The derivatives of $Q_{ij}$ with respect to $\lambda$ are given by
\begin{equation*}
%\label{expression-Q-1-d}
\begin{split}
Q_{11,\lambda}&={\rm i}\xi Q_{11}+R_{1}, \quad
Q_{21,\lambda}={\rm i}\xi Q_{21}+R_{2}, \quad
Q_{31,\lambda}={\rm i}\xi Q_{31}+R_{3}, \\
Q_{12,\lambda}&=-{\rm i}\xi Q_{12}, \quad
Q_{22,\lambda}=-{\rm i}\xi Q_{22}+{\rm e}^{-{\rm i} \lambda \xi}, \quad
Q_{32,\lambda}=-{\rm i}\xi Q_{32}, \\ 
Q_{13,\lambda}&=-{\rm i}\xi Q_{13}, \quad
Q_{23,\lambda}=-{\rm i}\xi Q_{23}, \quad
Q_{33,\lambda}=-{\rm i}\xi Q_{33}+{\rm e}^{-{\rm i} \lambda \xi} \\
\end{split}
\end{equation*}
where
\begin{align*}
    R_{1}&=(2\lambda-\lambda_{1}^{*}-\lambda_{2}^{*})-\frac{2{\rm i}}{M_{non}}\left(
        \frac{(b_{1}-b_{2})^{2}}{b_{1}+b_{2}}\tau_{00}+\sum_{s=1}^{2}b_{3-s}|c_{ss}|\tau_{4-2s,2s-2}
    \right){\rm e}^{{\rm i}\lambda\xi},\\
    R_{i+1}&=\frac{2c_{ii}b_{i}}{M_{non}}\left(
        |c_{3-i,3-i}|^{2}\tau_{i,3-i}+\frac{b_{1}-b_{2}}
        {b_{1}+b_{2}}\tau_{00}
    \right){\rm e}^{-2{\rm i}a\xi}{\rm e}^{{\rm i}\lambda \xi}, \quad i=1,2, 
\end{align*}
and $\tau_{ij}(x)={\rm e}^{[(2i-2)b_{1}+(2j-2)b_{2}]x}$.
Now we calculate the derivative of squared eigenfunctions with respect to $\lambda$. Denote
$\mathbf{e}_{i}$ be the identity column vector with $i$-th component be 1, and define
\begin{equation*}
\begin{split}
\mathbf{G}_{i-1}(\lambda)=\begin{pmatrix}
R_{2}(\lambda)Q_{1i}^{*}(\lambda^{*}) &
R_{3}(\lambda)Q_{1i}^{*}(\lambda^{*}) &
-R_{1}(\lambda)Q_{2i}^{*}(\lambda^{*}) &
-R_{1}(\lambda)Q_{3i}^{*}(\lambda^{*})
\end{pmatrix}, \\
\mathbf{G}_{-i+1}(\lambda)=\begin{pmatrix}
Q_{2i}(\lambda)R_{1}^{*}(\lambda^{*}) &
Q_{3i}(\lambda)R_{1}^{*}(\lambda^{*}) &
-Q_{1i}(\lambda)R_{2}^{*}(\lambda^{*}) &
-Q_{1i}(\lambda)R_{3}^{*}(\lambda^{*})
\end{pmatrix}
\end{split}
\end{equation*}
for $i=2,3$. 
We obtain
\begin{equation*}
%\label{squared-eigenfunctions-derivative}
\begin{split}
&\hat{\mathbf{P}}_{i,\lambda}(\lambda)=
2{\rm i}\xi\hat{\mathbf{P}}_{i}(\lambda)+\mathbf{G}_{i}(\lambda)-Q_{11}(\lambda){\rm e}^{{\rm i}\lambda\xi}\mathbf{e}_{i+2}, \\
&\hat{\mathbf{P}}_{-i,\lambda}(\lambda)=
-2{\rm i}\xi\hat{\mathbf{P}}_{-i}(\lambda)+\mathbf{G}_{-i}(\lambda)+Q_{11}^{*}(\lambda^{*}){\rm e}^{-{\rm i}\lambda\xi}\mathbf{e}_{i},
\end{split}
\end{equation*}
for $i=1,2$. 

The squared eigenfunctions \eqref{eigenfunction-JL1} are used 
to find the spectrum of \(\mathcal{J}\mathcal{L}_{1}\).
The linearized operator for CNLS equations \eqref{CNLS} is
\begin{equation*}
\mathcal{L}_{t} =
\left( - \frac{1}{2} \partial_{x}^{2} - {\rm i} \partial_{t} - |\mathbf{q}|^2 \right) \mathbb{I}_{2}
-\mathbf{q} \otimes \mathbf{q}^{*}-\mathbf{q}\otimes \mathbf{q}\cdot^{*}
\end{equation*}
obtained by the variation of \eqref{CNLS}. 
Let $\mathbf{\Phi}(\lambda;x,t)$ be a solution of the Lax pair \eqref{CNLS-lax} 
and $\mathbf{\Psi}(\lambda;x,t)$ be a solution of the adjoint spectral problem
\begin{equation*}
-\mathbf{\Psi}_{x}=\mathbf{\Psi}\mathbf{U}, \quad 
-\mathbf{\Psi}_{t}=\mathbf{\Psi}\mathbf{V},
\end{equation*}
setting $\mathbf{F} = \mathbf{\Phi}(\lambda;x,t)\mathbf{\Psi}(\lambda;x,t)$, we have
\begin{equation}\label{L-equ}
\begin{split}
\mathbf{F}_{x}=[\mathbf{U},\mathbf{F}],\quad
\mathbf{F}_{t}=[\mathbf{V},\mathbf{F}].
\end{split}
\end{equation}
Denote 
\begin{equation*}
\mathbf{F}=\begin{pmatrix}
f & \mathbf{h}^{T} \\
\mathbf{g} & \mathbf{K}
\end{pmatrix},
\end{equation*}
it leads to
\begin{align*}
\left( - \frac{1}{2} \partial_{x}^{2} - {\rm i} \partial_{t} - |\mathbf{q}|^2 \right) \mathbf{g}
- (\mathbf{q}\otimes \mathbf{q}^{*})\mathbf{g}+(\mathbf{q}\otimes \mathbf{q})\mathbf{h} =& 0, \\
\left( - \frac{1}{2} \partial_{x}^{2} - {\rm i} \partial_{t} - |\mathbf{q}|^2 \right) \mathbf{h}
-(\mathbf{q}\otimes \mathbf{q}^{*})^{*}\mathbf{h}+(\mathbf{q}\otimes \mathbf{q})^{*}\mathbf{g} =& 0 ,
\end{align*}
due to equations \eqref{L-equ}. 
Hence the squared eigenfunctions satisfy
\begin{equation}\label{linear-problem-t}
\mathcal{L}_{t}
\begin{pmatrix}
\mathbf{g} \\ -\mathbf{h}
\end{pmatrix}=0. 
\end{equation}
We identify four squared eigenfunctions, \(\mathbf{P}_{j}\) for \(j = \pm 1, \pm 2\), that satisfy the linearized equations \eqref{linear-problem-t}. The squared eigenfunctions are written in the separable form 
\begin{equation*}
\mathbf{P}_{j} = \hat{\mathbf{P}}_{j} e^{{\rm i} \Omega_j t}, \quad 
j = \pm 1, \pm 2,
\end{equation*}
hence the squared eigenfunctions \eqref{eigenfunction-JL1}
satisfy 
\begin{equation}
\label{spectral-JL-ch-1}
\mathcal{J}\mathcal{L}_{1}\hat{\mathbf{P}}_{j} = \Omega_{j}\hat{\mathbf{P}}_{j}, \qquad j=\pm 1,\pm 2,
\end{equation}
where $\Omega_{j}=2{\rm i}(\lambda-\lambda_{j})(\lambda-\lambda_{j}^{*})$ and $\Omega_{-j}=-\Omega_{j}$ for $j=1,2$. Computations in Appendix \ref{B.asy-exp} show that 
\begin{equation*}
\begin{split}
\hat{\mathbf{P}}_{i} \in \mathcal{S}(\mathbb{R})\subset L^{2}(\mathbb{R}), \quad 
\lambda=\lambda_{1},\lambda_{2},\lambda_{1}^{*},\lambda_{2}^{*}, \\
\hat{\mathbf{P}}_{1},\hat{\mathbf{P}}_{2} \in  {\rm e}^{-2\mathrm{Im}(\lambda)\xi}L^{\infty}(\mathbb{R}), \quad 
\lambda\ne\lambda_{1},\lambda_{2},\lambda_{1}^{*},\lambda_{2}^{*}, \\
\hat{\mathbf{P}}_{-1},\hat{\mathbf{P}}_{-2} \in  {\rm e}^{2\mathrm{Im}(\lambda)\xi}L^{\infty}(\mathbb{R}), \quad 
\lambda\ne\lambda_{1},\lambda_{2},\lambda_{1}^{*},\lambda_{2}^{*}.
\end{split}
\end{equation*} 
Hence the squared eigenfunctions belong to $L^{\infty}(\mathbb{R},\mathbb{C}^4)$ if and only if 
$\lambda\in\{\lambda_{1},\lambda_{2},\lambda_{1}^{*},\lambda_{2}^{*}\}\cup \mathbb{R}$ which is the Lax spectrum of the Lax pair \eqref{CNLS-lax} at the 
non-degenerate vector solitons. Since 
$\Omega_{j}=2{\rm i}((\lambda-\mathrm{Re}\lambda_{j})^{2}+\mathrm{Im}^{2}(\lambda_{j}))$,
the values of $\Omega_{i}$ on Lax spectrum are
\begin{equation*}
\begin{split}
&\Omega_{j} \in
[2{\rm i}b_{i}^{2},{\rm i}\infty ),
\quad -\Omega_{j}\in(-{\rm i}\infty , -2{\rm i}b_{i}^{2}],
\quad \lambda\in\mathbb{R},   \\
&\pm\Omega_{1}(\lambda_{1})=0 ,\quad \pm\Omega_{1}(\lambda_{1}^{*})=0 ,\quad 
\pm\Omega_{1}(\lambda_{2})=\pm 2{\rm i}(b_{1}^{2}-b_{2}^{2}) ,
\quad \pm\Omega_{1}(\lambda_{2}^{*})=\pm 2{\rm i}(b_{1}^{2}-b_{2}^{2}) ,
\\
&\pm\Omega_{2}(\lambda_{1})=\pm 2{\rm i}(b_{2}^{2}-b_{1}^{2}) ,\quad 
\pm\Omega_{2}(\lambda_{1}^{*})=\pm 2{\rm i}(b_{2}^{2}-b_{1}^{2}) ,\quad 
\pm\Omega_{2}(\lambda_{2})=0 ,
\quad \pm\Omega_{2}(\lambda_{2}^{*})=0. 
\end{split}
\end{equation*}
Using the squared eigenfunctions, all solutions for spectral problem \eqref{spectral-JL} are obtained according to the following lemma.

\begin{lem}
	\label{Lem-sol-to-JL1}
The squared eigenfunctions \eqref{eigenfunction-JL1} 
	are fundamental solutions of the spectral problem \eqref{spectral-JL} for any $\Omega\in\mathbb{C}$. If 
	$\Omega\notin\{ \pm 2{\rm i}b_{1}^{2},\pm 2{\rm i}b_{2}^{2}, 0,\pm 2{\rm i}(b_{1}^{2}-b_{2}^{2}) \}$, the FMS are 
	\begin{equation}\label{eigenfunction-JL}
	\begin{split}
	&\hat{\mathbf{P}}_{1}\left(a\pm \sqrt{-\frac{{\rm i}}{2}\Omega-b_{1}^{2}}\right), \quad 
	\hat{\mathbf{P}}_{2}\left(a\pm \sqrt{-\frac{{\rm i}}{2}\Omega-b_{2}^{2}}\right),
	\\ 
	&\hat{\mathbf{P}}_{-1}\left(a\pm \sqrt{\frac{{\rm i}}{2}\Omega-b_{1}^{2}}\right), \quad 
	\hat{\mathbf{P}}_{-2}\left(a\pm \sqrt{\frac{{\rm i}}{2}\Omega-b_{2}^{2}}\right).
	\end{split}
	\end{equation}
	If $\Omega\in\{ \pm 2{\rm i}b_{1}^{2},\pm 2{\rm i}b_{2}^{2}, 0
	,\pm 2{\rm i}(b_{1}^{2}-b_{2}^{2}) \}$, the FMS is given by the limit of a linear combination of squared eigenfunctions. 
\end{lem}

\begin{proof}
	We rewrite the spectral problem \eqref{spectral-JL} for $\hat{\mathbf{P}}$ as the following system of linear first-order ODEs 
	\begin{equation}\label{linear-CNLS-ODE}
	\begin{pmatrix}
	\hat{\mathbf{P}} \\
	\hat{\mathbf{P}}_{x}
	\end{pmatrix}_{x}
	=
	\mathbf{A}(\Omega)
	\begin{pmatrix}
	\hat{\mathbf{P}} \\
	\hat{\mathbf{P}}_{x}
	\end{pmatrix},
	\end{equation}
	where $\mathbf{A}(\Omega)$ is obtained by the spectral problem \eqref{spectral-JL}. 
	
	For general $\lambda$, we can find eight solutions for 
	\eqref{linear-CNLS-ODE}. Hence these solutions are all solutions for the spectral problem \eqref{linear-CNLS-ODE}. 
	More {concretely}, if 
	$\lambda\notin\{ \mathrm{Re}(\lambda_{1}), \lambda_{1}, \lambda_{1}^{*},\mathrm{Re}(\lambda_{2}), \lambda_{2}, \lambda_{2}^{*} \}$, 
	then there are eight independent solution \eqref{eigenfunction-JL}. 
	We denote the FMS be $\hat{\mathbf{P}}_{A}(\lambda)$. 
	
	On the other hand, for 
	$\lambda_{0}\in\{ \mathrm{Re}(\lambda_{1}), \lambda_{1}, \lambda_{1}^{*},\mathrm{Re}(\lambda_{2}), \lambda_{2}, \lambda_{2}^{*} \}$, 
	the FMS is given by 
	\begin{equation*}
	\lim_{\lambda\to\lambda_{0}}\hat{\mathbf{P}}_{A}(x,t;\lambda)[\hat{\mathbf{P}}_{A}(0,0;\lambda)]^{-1},
	\end{equation*}
	which proves the last assertion.
\end{proof}

Based on Lemma \ref{Lem-sol-to-JL1}, we obtain the spectrum of operator $\mathcal{J}\mathcal{L}_{1}$ as follows.
 
\begin{prop}
	\label{spectrum-JL1}
The spectrum of operator $\mathcal{J}\mathcal{L}_{1}$ in  $L^2(\mathbb{R},\mathbb{C}^4)$ is given by 
	\begin{equation*}
	\sigma(\mathcal{J}\mathcal{L}_{1})=
	[-{\rm i}\infty,-2{\rm i}\min_{i=1,2}b_{i}^{2}]
	\cup[2{\rm i}\min_{i=1,2}b_{i}^{2},+{\rm i}\infty]\cup 
	\{0,\pm 2{\rm i}(b_{1}^{2}-b_{2}^{2})\}. 
	\end{equation*}
	The essential spectrum is
	\begin{equation*}
	\sigma_{ess}(\mathcal{J}\mathcal{L}_{1})=
	[-{\rm i}\infty,-2{\rm i}\min_{i=1,2}b_{i}^{2}]
	\cup[2{\rm i}\min_{i=1,2}b_{i}^{2},+{\rm i}\infty]\backslash 
	\{ \pm 2{\rm i}(b_{1}^{2}-b_{2}^{2}) \}
	\end{equation*}
	and the eigenfunctions are 
	\begin{equation*}
	%\label{eigenfunction-continuous}
	\hat{\mathbf{P}}_{j}(\lambda), \quad \lambda\in\mathbb{R}, \qquad j = \pm 1, \pm 2.
	\end{equation*}
	The point spectrum is 
	\begin{equation*}
	\sigma_{point}(\mathcal{J}\mathcal{L}_{1})=\{0,\pm 2{\rm i}(b_{1}^{2}-b_{2}^{2})\}. 
	\end{equation*} 
	The eigenfunctions for $\Omega=0$ are spanned by
	\begin{equation*}
	%\label{eigenfunction-point}
	\hat{\mathbf{P}}_{1}(\lambda_{1}), \quad 
	\hat{\mathbf{P}}_{1}(\lambda_{1}^{*}) ,\quad
	\hat{\mathbf{P}}_{2}(\lambda_{2}), \quad 
	\hat{\mathbf{P}}_{2}(\lambda_{2}^{*}),
	\end{equation*}
	and the generalized eigenfunctions are spanned by
	\begin{equation*}
	%\label{ge-eigenfunction-point}
	\hat{\mathbf{P}}_{1,\lambda}(\lambda_{1}), \quad 
	\hat{\mathbf{P}}_{1,\lambda}(\lambda_{1}^{*}) ,\quad
	\hat{\mathbf{P}}_{2,\lambda}(\lambda_{2}), \quad 
	\hat{\mathbf{P}}_{2,\lambda}(\lambda_{2}^{*}).
	\end{equation*}
	The eigenfunctions for $\Omega=2{\rm i}(b_{1}^{2}-b_{2}^{2})$ are spanned by
	\begin{equation*}
	%\label{eigenfunction-point-1}
	\hat{\mathbf{P}}_{1}(\lambda_{2}), \quad 
	\hat{\mathbf{P}}_{1}(\lambda_{2}^{*}) .
	\end{equation*}
	The eigenfunctions for $\Omega=-2{\rm i}(b_{1}^{2}-b_{2}^{2})$ are spanned by
	\begin{equation*}
	%\label{eigenfunction-point-2}
	\hat{\mathbf{P}}_{2}(\lambda_{1}), \quad 
	\hat{\mathbf{P}}_{2}(\lambda_{1}^{*}).
	\end{equation*}
\end{prop}

\begin{proof}
	Having determined all solutions to the spectral problem \eqref{spectral-JL} by Lemma \ref{Lem-sol-to-JL1}, 
	we now focus on the asymptotic behavior of squared eigenfunctions described in \eqref{eigenfunction-JL}. 
	The kernel of \(\mathcal{J}\mathcal{L}_{1}\) can be obtained by squared eigenfunctions on 
	the point spectrum of Lax spectrum. 
    Since the determinant of the FMS $\mathbf{\Phi}^{[2]}_{r,non}$ is zero at the 
    points $\lambda={\lambda_{1}, \lambda_{2}, \lambda_{1}^{*}, \lambda_{2}^{*}}$, 
    the squared eigenfunctions in the point spectrum are linearly dependent, and it can be shown that
	\begin{equation*}
        \begin{split}
        \hat{\mathbf{P}}_{-1}(\lambda_{1}) &= \frac{1}{c_{11}^{2}}\frac{b_{2}-b_{1}}{b_{1}+b_{2}}\hat{\mathbf{P}}_{1}(\lambda_{1}) , \\
        \hat{\mathbf{P}}_{-1}(\lambda_{1}^{*}) &= (c_{11}^{*})^{2}\frac{b_{1}+b_{2}}{b_{2}-b_{1}}\hat{\mathbf{P}}_{1}(\lambda_{1}^{*}), \\
        \hat{\mathbf{P}}_{-1}(\lambda_{2})&=\frac{(b_{1}-b_{2})}{2c_{11}c_{22}b_{1}}\left( \hat{\mathbf{P}}_{2}(\lambda_{1})+|c_{11}|^{2}\hat{\mathbf{P}}_{2}(\lambda_{1}^{*}) \right), \\
        \hat{\mathbf{P}}_{-1}(\lambda_{2}^{*})&=\frac{c_{22}^{*}(b_{2}-b_{1})}{2c_{11}b_{1}}\left(\hat{\mathbf{P}}_{2}(\lambda_{1})+\frac{(b_{1}+b_{2})^{2}}{(b_{2}-b_{1})^{2}}|c_{11}|^{2}\hat{\mathbf{P}}_{2}(\lambda_{1}^{*})\right), \\
        \hat{\mathbf{P}}_{-2}(\lambda_{1})&=\frac{(b_{2}-b_{1})}{2c_{11}c_{22}b_{2}}\left( \hat{\mathbf{P}}_{1}(\lambda_{2})+|c_{22}|^{2}\hat{\mathbf{P}}_{1}(\lambda_{2}^{*}) \right), \\
        \hat{\mathbf{P}}_{-2}(\lambda_{1}^{*}) &= \frac{c_{11}^{*}(b_{1}-b_{2})}{2c_{22}b_{2}}\left(\hat{\mathbf{P}}_{1}(\lambda_{2})+\frac{(b_{1}+b_{2})^{2}}{(b_{2}-b_{1})^{2}}|c_{22}|^{2}\hat{\mathbf{P}}_{1}(\lambda_{2}^{*})\right), \\
        \hat{\mathbf{P}}_{-2}(\lambda_{2})&=\frac{1}{c_{22}^{2}}\frac{b_{1}-b_{2}}{b_{1}+b_{2}}\hat{\mathbf{P}}_{2}(\lambda_{2}) , \\
        \hat{\mathbf{P}}_{-2}(\lambda_{2}^{*}) &= (c_{22}^{*})^{2}\frac{b_{1}+b_{2}}{b_{1}-b_{2}}\hat{\mathbf{P}}_{2}(\lambda_{2}^{*}). \\
        \end{split}
    \end{equation*}
	To analyze the generalized eigenspace on point spectrum, we take the derivative with respect to \(\lambda\) 
	on both sides of \eqref{spectral-JL-ch-1} 
	\begin{equation*}
	\mathcal{J}\mathcal{L}_{1}\hat{\mathbf{P}}_{j,\lambda}(\lambda) = \Omega_{j,\lambda}(\lambda)\hat{\mathbf{P}}_{j}(\lambda) +
	\Omega_{j}(\lambda)\hat{\mathbf{P}}_{j,\lambda}(\lambda), \qquad j=\pm 1,\pm 2. 
	\end{equation*}
	The proof is completed by using the asymptotic behavior of generalized eigenfunctions in Appendix \ref{B.asy-exp}.
\end{proof}

\begin{rem}
It is well-known for general soliton solutions \cite{yang_nonlinear_2010} that the squared eigenfunctions for the point spectrum are also eigenfunctions of the recursion operator $\mathcal{L}_{r}$ given by \eqref{reu-op}. However, for non-degenerate vector soliton solutions with $c_{12}=c_{21}=0$, the corresponding relationships need to be slightly amended. In this case, we have 
\begin{equation*}
\mathcal{L}_{r}\mathbf{p}(\lambda)=\lambda\mathbf{p}(\lambda).
\end{equation*}
for the eigenvectors given by 
\begin{equation*}
    \mathbf{p}\in\{\hat{\mathbf{P}}_{\pm 1}(\lambda)\hat{\mathbf{P}}_{\pm 2}(\lambda): \lambda=\lambda_{1},\lambda_{1}^{*},\lambda_{2},\lambda_{2}^{*}\}\backslash
\{\hat{\mathbf{P}}_{1}(\lambda_{2}),\hat{\mathbf{P}}_{-1}(\lambda_{2}^{*}),\hat{\mathbf{P}}_{2}(\lambda_{1}),\hat{\mathbf{P}}_{-2}(\lambda_{1}^{*})\}.
\end{equation*}
However, the left eigenvectors satisfy the following equations:
\begin{align*}
    \mathcal{L}_{r}\hat{\mathbf{P}}_{1}(\lambda_{2}) &= \lambda_{1}\hat{\mathbf{P}}_{1}(\lambda_{2}) + |c_{22}|^{2}(\lambda_{2}-\lambda_{1}^{*})\hat{\mathbf{P}}_{1}(\lambda_{2}^{*}), \\
    \mathcal{L}_{r}\hat{\mathbf{P}}_{2}(\lambda_{1}) &= \lambda_{2}\hat{\mathbf{P}}_{2}(\lambda_{1}) + |c_{11}|^{2}(\lambda_{1}-\lambda_{2}^{*})\hat{\mathbf{P}}_{2}(\lambda_{1}^{*}), \\
    \mathcal{L}_{r}\hat{\mathbf{P}}_{-1}(\lambda_{2}^{*}) &= \lambda_{1}^{*}\hat{\mathbf{P}}_{-1}(\lambda_{2}^{*}) + |c_{22}|^{2}(\lambda_{2}^{*}-\lambda_{1})\hat{\mathbf{P}}_{-1}(\lambda_{2}), \\
    \mathcal{L}_{r}\hat{\mathbf{P}}_{-2}(\lambda_{1}^{*}) &= \lambda_{2}^{*}\hat{\mathbf{P}}_{-2}(\lambda_{1}^{*}) + |c_{11}|^{2}(\lambda_{1}^{*}-\lambda_{2})\hat{\mathbf{P}}_{-2}(\lambda_{1}).
\end{align*}
For example, the squared eigenfunction matrices satisfy
\begin{equation*}
    p_{1}(\mathbf{\Phi}^{[2]}_{r,non}(x;\lambda_{2}))\sim\begin{pmatrix}
        0 & 0 & 0 \\
        0 & 0 & 0 \\
        0 & 2{\rm i}c_{22}b_{2}(b_{2}^{2}-b_{1}^{2}) & 0
    \end{pmatrix},\quad x\to-\infty,
\end{equation*}
and hence the integral in $\mathcal{L}_{r}\hat{\mathbf{P}}_{1}(\lambda_{2})$ is computed as follows:
\begin{equation*}
    \mathcal{L}_{r}\hat{\mathbf{P}}_{1}(\lambda_{2}) = \lambda_{2}\hat{\mathbf{P}}_{1}(\lambda_{2})+{\rm i}c_{22}b_{2}(b_{2}^{2}-b_{1}^{2})(0,\hat{q}_{1},-\hat{q}_{2}^{*},0)^{T}
\end{equation*}
where $\hat{\mathbf{q}}=(\hat{q}_{1},\hat{q}_{2})^{T}$ is given in \eqref{q-non2-spearate}. By straightforward calculation, we obtain 
\begin{equation*}
    (0,\hat{q}_{1},-\hat{q}_{2}^{*},0)^{T}=-\frac{1}{c_{22}b_{2}(b_{1}+b_{2})}\hat{\mathbf{P}}_{1}(\lambda_{2})+\frac{c_{22}^{*}}{b_{2}(b_{2}-b_{1})}\hat{\mathbf{P}}_{1}(\lambda_{2}^{*}),
\end{equation*}
hence $\mathcal{L}_{r}\hat{\mathbf{P}}_{1}(\lambda_{2})$ is equal to $\lambda_{1}\hat{\mathbf{P}}_{1}(\lambda_{2}) + |c_{22}|^{2}(\lambda_{2}-\lambda_{1}^{*})\hat{\mathbf{P}}_{1}(\lambda_{2}^{*})$. 
For the scalar NLS equation, the above relations do not arise since each spectral parameter corresponds to one scattering parameter, while in CNLS equations, one spectral parameter $\lambda_{i}$ corresponds to two scattering parameters $\mathbf{c}_{i}$.
\end{rem}

\begin{figure}[htbp]
	\centering
	\includegraphics[scale=0.24]{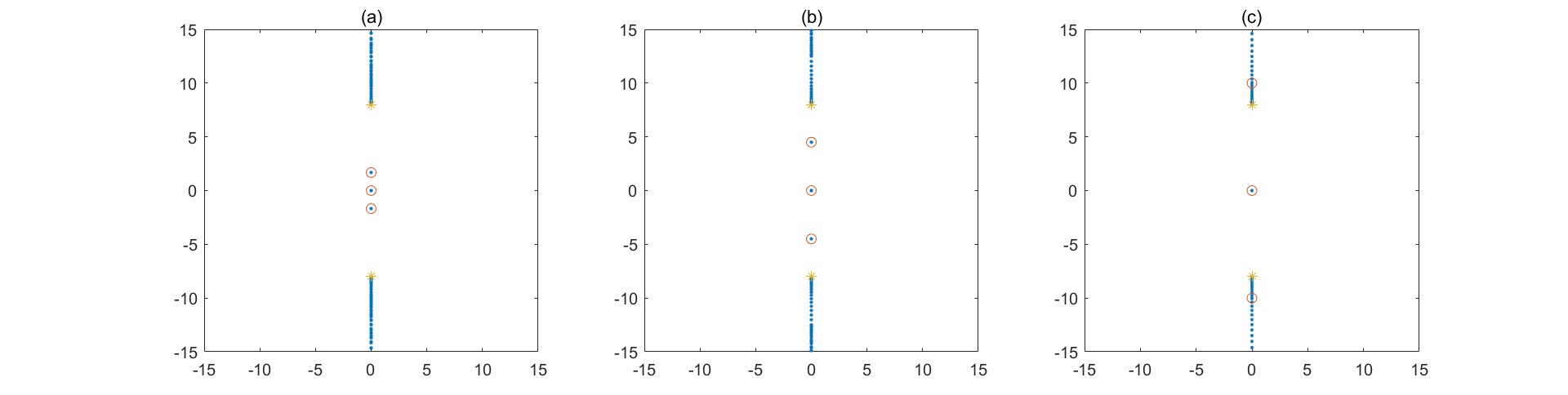}
	\caption{Computation of the spectrum of operator  $\mathcal{J}\mathcal{L}_{1}$ through Fourier 
		collocation method. The blue points are the numerical results. The red circles are the eigenvalues in point spectrum 
		and the yellow stars are the end points of essential spectrum $\pm 2{\rm i} \min\{ b_{1}^{2}, b_2^2 \}$. The parameters are $a=0$, $b_{1}=2$, $c_{11} = c_{22} = 1$, and (a) $b_{2}=2.2$, (b) $b_{2}=2.5$, (c) $b_{2}=3$. }
	\label{computation-spectrum}
\end{figure}

As an illustration, we apply the Fourier collocation method to numerically compute the spectrum of the operator \(\mathcal{J}\mathcal{L}_{1}\). Figure \ref{computation-spectrum} presents three cases. 
For a fixed \(b_{1}\), as \(b_{2}\) increases, 
isolated eigenvalues of \(\mathcal{J}\mathcal{L}_{1}\) on panels (a,b) 
become embedded eigenvalues on panel (c). 
The numerical results are in agreement with the analytical 
expressions in Proposition \ref{spectrum-JL1}.

Since the spectrum of the operator \(\mathcal{J}\mathcal{L}_{1}\) is contained on the purely imaginary axis, the proof of Theorem \ref{spectral-stability} on the spectral stability of non-degenerate vector solitons is complete.

\section{Spectra of $\mathcal{J}\mathcal{L}_{2}$ and $\mathcal{L}_{2}$}
\label{JL2-L2}

We consider the spectra of \(\mathcal{J}\mathcal{L}_{2}\) 
and \(\mathcal{L}_{2}\) for the breather solutions of Definition \ref{def-breather}. The spectrum of \(\mathcal{J}\mathcal{L}_{2}\) 
can be obtained by using squared eigenfunctions, and the key tool connecting the spectra of these two operators is the closure relation for squared eigenfunctions. The main difficulty lies in the inner product \((\mathcal{L}_{2}\cdot,\cdot)\) within the space 
\(\mathrm{gKer}(\mathcal{J}\mathcal{L}_{2}) \backslash \mathrm{Ker}(\mathcal{J}\mathcal{L}_{2})\). 
We propose a new method to calculate it.

As the main outcome of our analysis, we compute the number of negative eigenvalues and the dimension of the kernel for the second variation of the Lyapunov functional
\eqref{Ly-non}:
\begin{align*}
    \mathcal{L}_{2}=&
    16(a^{2}+b_{1}^{2})(a^{2}+b_{2}^{2})\frac{\delta^{2} H_{0}}{\delta q^{2}}
    -16a(2a^{2}+b_{1}^{2}+b_{2}^{2})\frac{\delta^{2} H_{1}}{\delta q^{2}}\\&+
    4(6a^{2}+b_{1}^{2}+b_{2}^{2})\frac{\delta^{2} H_{2}}{\delta q^{2}}
    -8a\frac{\delta^{2} H_{3}}{\delta q^{2}}+\frac{\delta^{2} H_{4}}{\delta q^{2}}.
\end{align*}
This linear operator appears in the expansion of the Lyapunov functional around the breather solution $\mathbf{q}^{[2]}$:
\begin{equation*}
%\label{exp-I2}
    \mathcal{I}_{2}(\mathbf{q}^{[2]}+\epsilon \mathbf{u})=\mathcal{I}_{2}(\mathbf{q}^{[2]})+
    \frac{1}{2}\epsilon^{2}
    (\mathcal{L}_{2}(\mathbf{u}),\mathbf{u})+\mathcal{O}(\epsilon^{3}).
\end{equation*}
It is necessary to study the spectrum of \(\mathcal{L}_{2}\) for the proof of  the nonlinear stability of breathers in Section \ref{sec-5}. 

Recall that the non-degenerate vector solitons are included in the breather 
solutions (\ref{breather}) for $c_{11},c_{22}\ne 0$ and $c_{12}=c_{21}=0$.
For general non-degenerate breather solutions, three of the scattering parameters $c_{11},c_{12},c_{21},c_{22}$ are required to be nonzero.
If $c_{11}\neq 0$ and $c_{12}=c_{22}=0$, then the breather solutions become the degenerate vector soliton with the single-humped profiles, stability of which is well-studied, see \cite{cipolatti2000orbitally,li_structural_1998,li_mechanism_2000,nguyen2011orbital,nguyen_existence_2015,ohta_stability_1996}. Hence we assume that $c_{11}\neq 0$ with either $c_{12}\neq 0$ or $c_{22} \neq 0$ and without loss of generality, we consider $c_{11},c_{12} \ne 0$ and arbitrary $c_{21},c_{22}\in\mathbb{C}$.

\subsection{The squared eigenfunctions}

The squared eigenfunctions for the non-degenerate vector soliton solutions have been constructed in Section \ref{sec-3}. Here we consider the general case of breathers. To eliminate the singularity associated 
with $\lambda$ in the point spectrum, we introduce the regular Darboux matrix and the regular FMS
\begin{equation*}
    \mathbf{D}_{r}^{[2]}=(\lambda-\lambda_{1}^{*})(\lambda-\lambda_{2}^{*})\mathbf{D}^{[2]}|_{t=0},\quad 
    \mathbf{\Phi}^{[2]}_{r}=(\lambda-\lambda_{1}^{*})(\lambda-\lambda_{2}^{*})\mathbf{\Phi}^{[2]}|_{t=0}. 
\end{equation*}
The squared eigenfunctions for breather solutions are
\begin{equation}\label{squared-eigenfunction-JL2}
    \mathbf{P}_{\pm 1}(\lambda)=s_{\pm 1}(\mathbf{\Phi}^{[2]}_{r}), \qquad  \mathbf{P}_{\pm 2}(\lambda)=s_{\pm 2}(\mathbf{\Phi}^{[2]}_{r}).
\end{equation}
To obtain the asymptotic behavior of squared eigenfunction, we define $\mathbf{D}_{r}^{+}=(D_{r,ij}^{+})_{1\leq i,j\leq 3}$ with 
    \begin{equation*}
        \begin{split}
            &D_{r,11}^{+}=
                (\lambda-\lambda_{1}^{*})(\lambda-\lambda_{2}^{*})-\frac{1}{M^{+}}\sum_{s=1}^{2}
                    (\lambda-\lambda_{s})(\lambda-\lambda_{3-s}^{*})d_{3-s,3-s}{\rm e}^{-4b_{s}x}
                , \\
                &D_{r,i+1,1}^{+}=\frac{2}{M^{+}}\sum_{s=1}^{2}
                b_{s}D_{i,3-s}(\lambda-\lambda_{3-s}^{*}){\rm e}^{-2 b_{s}x}{\rm e}^{-2{\rm i}ax}, \\
                &D_{r,1,i+1}^{+}=\frac{-2}{M^{+}}\sum_{s=1}^{2}
                \left(b_{s}d_{3-s,3-s}c_{is}^{*}(\lambda-\lambda_{3-s}^{*})-\frac{2b_{1}b_{2}d_{s,3-s}c_{i,3-s}^{*}}{b_{1}+b_{2}}(\lambda-\lambda_{s}^{*})\right){\rm e}^{-2 b_{s}x}
                {\rm e}^{2{\rm i}ax}, \\
                &D_{r,i+1,j+1}^{+}=
                (\lambda-\lambda_{1}^{*})(\lambda-\lambda_{2}^{*})\delta_{ij}-\frac{2{\rm i}}{M^{+}}\sum_{s=1}^{2}
                    b_{s}D_{i,3-s}c_{js}^{*}(\lambda-\lambda_{3-s}^{*})
        \end{split}
    \end{equation*}
    for $i,j=1,2$ and $\mathbf{D}_{r}^{-}=(D_{r,ij}^{-})_{1\leq i,j\leq 3}$ with 
    \begin{equation*}
        \begin{split}
            D_{r,11}^{-}=&
                (\lambda-\lambda_{1})(\lambda-\lambda_{2})-\frac{1}{M^{-}}\sum_{s=1}^{2}
                    (\lambda-\lambda_{s}^{*})(\lambda-\lambda_{3-s})d_{ss}{\rm e}^{4b_{s}x}
                , \\
            D_{r,i+1,1}^{-}=&\frac{2}{M^{-}}\sum_{s=1}^{2}
                b_{s}\frac{b_{s}-b_{3-s}}{b_{1}+b_{2}}c_{is}(\lambda-\lambda_{3-s}){\rm e}^{2b_{s}x}
                {\rm e}^{-2{\rm i}ax}, \\
            D_{r,1,i+1}^{-}=&\frac{-2}{M^{-}}\sum_{s=1}^{2}
                b_{s}\frac{b_{s}-b_{3-s}}{b_{1}+b_{2}}c_{is}^{*}(\lambda-\lambda_{3-s}^{*}){\rm e}^{2b_{s}x}
                {\rm e}^{2{\rm i}ax}, \\
            D_{r,i+1,i+1}^{-}=&
                (\lambda-\lambda_{1}^{*})(\lambda-\lambda_{2}^{*})-\frac{1}{M^{-}}\sum_{s=1}^{2}
                    (2{\rm i}b_{s}|c_{is}|^{2}-(\lambda-\lambda_{s}^{*})d_{ss})(\lambda-\lambda_{3-s}^{*}){\rm e}^{4b_{s}x}
                \\
            D_{r,i+1,4-i}^{-}=&
                    -\frac{2{\rm i}}{M^{-}}\sum_{s=1}^{2}
                        b_{s}c_{is}c_{3-i,s}^{*}(\lambda-\lambda_{3-s}^{*}){\rm e}^{4b_{s}x}
        \end{split}
    \end{equation*}
    for $i=1,2$. The elements of regular Darboux matrix $\mathbf{D}_{r}^{[2]}=(D_{r,ij}^{[2]})_{1\leq i,j\leq 3}$ satisfy
\begin{equation*}
    D_{r,ij}^{[2]}\sim D_{r,ij}^{\pm},\quad x\to\pm \infty,
\end{equation*}
from which the asymptotic behavior of the squared eigenfunctions can be derived.
The following proposition gives the construction of the squared eigenfunctions. 

\begin{prop}\label{asy-squ-eig-2}
    Assuming $c_{11},c_{21},c_{12}, c_{22} \neq 0$, the squared eigenfunctions \eqref{squared-eigenfunction-JL2} are nonzero and belong to the class of Schwartz functions on $\mathbb{R}$ for $\lambda=\lambda_{1},\lambda_{1}^{*},\lambda_{2},\lambda_{2}^{*}$. Moreover, they satisfy the following relations:
    \begin{equation*}
        \begin{split}
            \mathbf{P}_{2}(\lambda_{1}) &= \frac{c_{21}}{c_{11}} \mathbf{P}_{1}(\lambda_{1}), \quad
            \mathbf{P}_{2}(\lambda_{1}^{*}) = -\frac{c_{11}^{*}}{c_{21}^{*}} \mathbf{P}_{1}(\lambda_{1}^{*}) + \frac{1}{c_{11}^{*}c_{21}^{*}} \mathbf{P}_{-1}(\lambda_{1}^{*}), \\
            \mathbf{P}_{2}(\lambda_{2}) &= \frac{c_{22}}{c_{12}} \mathbf{P}_{1}(\lambda_{2}), \quad
            \mathbf{P}_{2}(\lambda_{2}^{*}) = -\frac{c_{12}^{*}}{c_{22}^{*}} \mathbf{P}_{1}(\lambda_{2}^{*}) + \frac{1}{c_{12}^{*}c_{22}^{*}} \mathbf{P}_{-1}(\lambda_{2}^{*}),
        \end{split}
    \end{equation*}
with linearly independent $\mathbf{P}_{\pm 1}(\lambda_{1})$, $\mathbf{P}_{\pm 1}(\lambda_{2})$,  $\mathbf{P}_{\pm 1}(\lambda_{1}^*)$, and $\mathbf{P}_{\pm 1}(\lambda_{2}^*)$. The functions
$$
\mathbf{P}_{1,\lambda}(\lambda_{1}), \mathbf{P}_{1,\lambda}(\lambda_{2}), \mathbf{P}_{-1,\lambda}(\lambda_{1}^{*}), \mathbf{P}_{-1,\lambda}(\lambda_{2}^{*})
$$ 
also belong to the class of Schwartz functions on $\mathbb{R}$, whereas 
$$
\mathbf{P}_{1,\lambda}(\lambda_{1}^{*}), \mathbf{P}_{1,\lambda}(\lambda_{2}^{*}), 
            \mathbf{P}_{-1,\lambda}(\lambda_{1}), \mathbf{P}_{-1,\lambda}(\lambda_{2})
$$ 
are unbounded as $|x| \to \infty$. 
\end{prop}

\begin{proof}
The relations between squared eigenfunctions $\mathbf{P}_{2}(\lambda)$ at $\lambda=\lambda_{1},\lambda_{1}^{*},\lambda_{2},\lambda_{2}^{*}$ are obtained by algebraic 
    calculations. The expressions for $\mathbf{P}_{-2}(\lambda)$ at  $\lambda=\lambda_{1},\lambda_{1}^{*},\lambda_{2},\lambda_{2}^{*}$ can be obtained by the relation \eqref{symmetric-squared-eigenfunction-matrix}. In view of the symmetry \eqref{symmetric-squared-eigenfunctions}, it suffices to consider 
$p_1(\mathbf{\Phi}_{r}^{[2]})(\lambda)$. We provide the calculations
for $p_{1}(\mathbf{\Phi}_{r}^{[2]})(\lambda_{1})$, and the calculations for the other entries are similar. The squared eigenfunction matrix satisfies 
\begin{equation*}
    p_1(\mathbf{\Phi}_{r}^{[2]})(\lambda) = {\rm e}^{2{\rm i}\lambda x} p_1(\mathbf{D}_{r}^{[2]})(\lambda) \sim
    {\rm e}^{2{\rm i}\lambda x} p_1(\mathbf{D}_{r}^{\pm})(\lambda), \quad x \to \pm\infty,
\end{equation*}
hence we have
\begin{equation*}
    p_1(\mathbf{\Phi}_{r}^{[2]})(\lambda_{1}) \sim 
    {\rm e}^{2{\rm i}a x}{\rm e}^{-2b_{1} x} p_1(\mathbf{D}_{r}^{\pm})(\lambda_{1}), \quad 
    x \to \pm \infty. 
\end{equation*}
It can be observed that \( s_1(\mathbf{\Phi}_{r}^{[2]})(\lambda_{1}) \) exhibits exponential decay as \( x \to +\infty \) due to the term \( {\rm e}^{-2b_{1} x} \). Elements of the first column of \( \mathbf{D}_{r}^{-} \) at \( \lambda = \lambda_{1} \) are represented as:
\begin{equation*}
    \begin{split}
        D_{r,11}^{-}(\lambda_{1}) &= (\lambda_{1}-\lambda_{2})(\lambda_{1}-\lambda_{1}^{*})d_{11}{\rm e}^{4b_{1}x}, \\
        D_{r,i+1,1}^{-}(\lambda_{1}) &= \frac{2b_{1}}{M^{-}}
            \frac{b_{1}-b_{2}}{b_{1}+b_{2}}c_{i1}(\lambda_{1}-\lambda_{2}){\rm e}^{2b_{1}x}
            {\rm e}^{-2{\rm i}ax},
    \end{split}
\end{equation*}
and elements of the second and third columns of \( \mathbf{D}_{r}^{-} \) at \( \lambda = \lambda_{1}^{*} \) are represented as
\begin{equation*}
    \begin{split}
        D_{r,1,i+1}^{-}(\lambda_{1}^{*}) &= -\frac{2b_{1}}{M^{-}}
            \frac{b_{1}-b_{2}}{b_{1}+b_{2}}c_{i1}^{*}(\lambda_{1}^{*}-\lambda_{2}^{*}){\rm e}^{2b_{1}x}{\rm e}^{2{\rm i}ax}, \\
        D_{r,i+1,i+1}^{-}(\lambda_{1}^{*}) &= -\frac{2{\rm i}b_{1}}{M^{-}}
            |c_{i1}|^{2}(\lambda_{1}^{*}-\lambda_{2}^{*}){\rm e}^{4b_{1}x}, \\
        D_{r,i+1,4-i}^{-}(\lambda_{1}^{*}) &= -\frac{2{\rm i}b_{1}}{M^{-}}
            c_{i1}c_{3-i,1}^{*}(\lambda_{1}^{*}-\lambda_{2}^{*}){\rm e}^{4b_{1}x},
    \end{split}
\end{equation*}
for $i=1,2$. Hence, \( s_1(\mathbf{\Phi}_{r}^{[2]})(\lambda_{1}) \) exhibits exponential decay as $x \to -\infty$ as well.   
For the derivative of \( s_1(\mathbf{\Phi}_{r}^{[2]}) \) at $\lambda=\lambda_{1}$, we have 
\begin{equation*}
    p_1(\mathbf{\Phi}_{r}^{[2]})_{\lambda}(\lambda_{1}) \sim 
    {\rm e}^{2{\rm i}a x}{\rm e}^{-2b_{1} x} p_1(\mathbf{D}_{r}^{\pm})_{\lambda}(\lambda_{1})+
    2{\rm i}x{\rm e}^{2{\rm i}a x}{\rm e}^{-2b_{1} x} p_1(\mathbf{D}_{r}^{\pm})(\lambda_{1}), \quad 
    x \to \pm \infty, 
\end{equation*}
with
\begin{equation*}
    p_{1}(\mathbf{D}_{r}^{\pm})_{\lambda}(\lambda_{1})=
    (\mathbf{\mathbf{D}}_{r,\lambda}^{\pm}(\lambda_{1}))_{1}
    ((\mathbf{\mathbf{D}}_{r}^{\pm})^{\dagger}(\lambda_{1}^{*}))^{2}+
    (\mathbf{\mathbf{D}}_{r}^{\pm}(\lambda_{1}))_{1}
    ((\mathbf{\mathbf{D}}_{r,\lambda}^{\pm})^{\dagger}(\lambda_{1}^{*}))^{2}. 
\end{equation*}
In view of the exponential decay \( {\rm e}^{-2b_{1} x} \) 
for \( p_{1}(\mathbf{\Phi}_{r}^{[2]})_{\lambda}(\lambda_{1}) \) as \( x \to +\infty \), 
we only need to consider the case \( x \to -\infty \). The functions 
\begin{equation*}
    \begin{split}
        (D_{r,11}^{-})_{\lambda}(\lambda_{1}) &= (\lambda_{1}-\lambda_{2})-\frac{1}{M^{-}}\sum_{s=1}^{2}
                (2\lambda_{1}-\lambda_{3-s}-\lambda_{s}^{*})d_{11}{\rm e}^{4b_{s}x}, \\
        (D_{r,i+1,1}^{-})_{\lambda}(\lambda_{1}) &= \frac{2}{M^{-}}\sum_{s=1}^{2}
            b_{s}\frac{b_{s}-b_{3-s}}{b_{1}+b_{2}}c_{is}{\rm e}^{2b_{s}x}{\rm e}^{-2{\rm i}ax}
    \end{split}
\end{equation*}
and 
\begin{equation*}
    \begin{split}
        (D_{r,1,i+1}^{-})_{\lambda}(\lambda_{1}^{*}) =& -\frac{2}{M^{-}}\sum_{s=1}^{2}
            b_{s}\frac{b_{s}-b_{3-s}}{b_{1}+b_{2}}c_{is}^{*}{\rm e}^{2b_{s}x}{\rm e}^{2{\rm i}ax}, \\
        (D_{r,i+1,i+1}^{-})_{\lambda}(\lambda_{1}^{*}) =& (\lambda_{1}^{*}-\lambda_{2}^{*})-\frac{1}{M^{-}}
        \sum_{s=1}^{2}
                (2{\rm i}b_{s}|c_{is}|^{2}-d_{ss}(\lambda_{1}^{*}-\lambda_{2}^{*})){\rm e}^{4b_{s}x},\\
        (D_{r,i+1,4-i}^{-})_{\lambda}(\lambda_{1}^{*}) =& -\frac{2{\rm i}}{M^{-}}\sum_{s=1}^{2}
                b_{s}c_{is}c_{3-i,s}^{*}{\rm e}^{4b_{s}x},
    \end{split}
\end{equation*}
for $i = 1,2$. Hence, \( s_{1}(\mathbf{\Phi}_{r}^{[2]})_{\lambda}(\lambda_{1}) \) exhibits the exponential decay as \( x = - \infty \) as well.
\end{proof}

\subsection{The spectrum for $\mathcal{J}\mathcal{L}_{2}$}

Since \(\mathbf{q}^{[2]}\) is a solution of $\frac{\delta \mathcal{I}_{2}}{\delta \mathbf{q}}=0$ for every $t \in \mathbb{R}$, it suffices to consider \(\mathcal{J}\mathcal{L}_{2}(\mathbf{q}^{[2]})|_{t=0}\). 
The eigenfunctions of \(\mathcal{J}\mathcal{L}_{2}\) are given by the squared eigenfunctions \eqref{squared-eigenfunction-JL2}.
Define 
\begin{equation*}
    \mathbf{V}_{n} = {\rm i} \sum_{i=0}^{n} \mathbf{L}_{i} \lambda^{n-i}, \quad n \in \mathbb{N},
\end{equation*}
where \(\mathbf{L}_{i}\) is defined in Appendix \ref{app.A2}. We consider the Lax pair associated with a linear combination of the first flows in the CNLS hierarchy:
\begin{equation}
\label{CNLS-higer-1}
    \begin{split}
        \mathbf{\Phi}_{2,x}(\lambda; x, t_{2}) &= \mathbf{U}(\lambda, \mathbf{u}) \mathbf{\Phi}_{2}(\lambda; x, t_{2}) \\
        \mathbf{\Phi}_{2,t_{2}}(\lambda; x, t_{2}) &= \sum_{i=0}^{4} 2^{i} \mu_{i} \mathbf{V}_{i}(\lambda, \mathbf{u}) \mathbf{\Phi}_{2}(\lambda; x, t_{2}). 
    \end{split}
\end{equation}
The squared eigenfunctions of the Lax pair \eqref{CNLS-higer-1} satisfy
\begin{equation*}
    \mathbf{B}_{2,x} = [\mathbf{U}, \mathbf{B}_{2}], \quad  
    \mathbf{B}_{2,t_{2}} =  \sum_{i=0}^{4} 2^{i} \mu_{i} [\mathbf{V}_{i}, \mathbf{B}_{2}]. 
\end{equation*}
If we denote
\begin{equation*}
    \mathbf{B}_{2}=\begin{pmatrix}
        f_{2} & \mathbf{h}_{2}^{T} \\
        \mathbf{g}_{2} & -\mathbf{f}_{1,2}
    \end{pmatrix}, 
\end{equation*}
then the straightforward calculation shows that 
\begin{equation*}
    \begin{pmatrix}
        \mathbf{g}_{2} \\
        -\mathbf{h}_{2}
    \end{pmatrix}_{t_{2}}=
    \mathcal{J}\mathcal{L}_{2}(\mathbf{u})\begin{pmatrix}
        \mathbf{g}_{2} \\
        -\mathbf{h}_{2}
    \end{pmatrix}. 
\end{equation*}

Applying the $2$-fold Darboux transformation to the solution
\begin{equation*}
    \mathbf{\Phi}_{2}^{[0]}={\rm e}^{{\rm i}(\lambda x+16\mathcal{P}_{2}(\lambda) t_{2})\sigma_{3}}
\end{equation*}
of the Lax pair \eqref{CNLS-higer-1} associated with zero potential $\mathbf{Q}_{2}^{[0]}=\mathbf{0}$ and the vector
\begin{equation*}
    |\mathbf{y}_i\rangle=\mathbf{\Phi}_{2}^{[0]}(\lambda_{i};x,t_{2})c^{[i]}
    ={\rm e}^{{\rm i}(\lambda_{i} x+16\mathcal{P}_{2}(\lambda_{i}) t_{2})\sigma_{3}}\begin{pmatrix}
        1 \\ \ii c_{1i} \\ \ii c_{2i}
    \end{pmatrix}, \quad i=1,2,
\end{equation*}
we obtain the new FMS 
\begin{equation*}
    \mathbf{\Phi}_{2}^{[2]}(\lambda;x,t_{2})
    =\mathbf{\Phi}^{[2]}(\lambda;x,0){\rm e}^{16{\rm i}\mathcal{P}_{2}(\lambda) t_{2}\sigma_{3}} 
\end{equation*}
of the same Lax pair \eqref{CNLS-higer-1} but associated with the potential
$\mathbf{Q}_{2}^{[2]}(x,t_{2})=\mathbf{Q}^{[2]}(x,0)$. 
Hence we have obtained the eigenfunctions of the operator $\mathcal{J}\mathcal{L}_{2}(\mathbf{q}^{[2]})|_{t=0}$. From here, we prove the following theorem, where we simplify the 
notation and write $\mathcal{J}\mathcal{L}_{2}$ instead of  $\mathcal{J}\mathcal{L}_{2}(\mathbf{q}^{[2]})|_{t=0}$.

\begin{thm}
	\label{thm-J-L2}
    The squared eigenfunctions \eqref{squared-eigenfunction-JL2}
    satisfy the spectral problem for the operator \(\mathcal{J}\mathcal{L}_{2}\) given by
    \begin{equation*}
        \mathcal{J}\mathcal{L}_{2}\mathbf{P}_{\pm 1}(\lambda) =\pm 32{\rm i}\mathcal{P}_{2}(\lambda)
        \mathbf{P}_{\pm 1}(\lambda), \qquad         \mathcal{J}\mathcal{L}_{2}\mathbf{P}_{\pm 2}(\lambda) =\pm 32{\rm i}\mathcal{P}_{2}(\lambda)
        \mathbf{P}_{\pm 2}(\lambda).
        \end{equation*}
    The spectrum of \(\mathcal{J}\mathcal{L}_{2}\) is given by 
    \begin{equation*}
        \sigma(\mathcal{J}\mathcal{L}_{2}) = (-{\rm i}\infty, -32{\rm i}b_{1}^{2}b_{2}^{2}] \cup [32{\rm i}b_{1}^{2}b_{2}^{2}, {\rm i}\infty) \cup \{0\}
    \end{equation*}
and includes the essential spectrum 
    \begin{equation*}
        \sigma_{ess}(\mathcal{J}\mathcal{L}_{2}) = (-{\rm i}\infty, -32{\rm i}b_{1}^{2}b_{2}^{2}] \cup [32{\rm i}b_{1}^{2}b_{2}^{2}, {\rm i}\infty)
    \end{equation*}
and the point spectrum
    \begin{equation*}
        \sigma_{point}(\mathcal{J}\mathcal{L}_{2}) = \{0\}.
    \end{equation*}
If $c_{11},c_{12}\neq 0$, then the eigenfunctions of ${\rm Ker}(\mathcal{J}\mathcal{L}_{2})$ are spanned by eight squared eigenfunctions 
$$
\mathbf{P}_{\pm 1}(\lambda_1), \quad \mathbf{P}_{\pm 1}(\lambda_2), \quad 
\mathbf{P}_{\pm 1}(\lambda_1^*), \quad \mathbf{P}_{\pm 1}(\lambda_2^*),
$$
whereas the generalized eigenfunctions of ${\rm Ker}(\mathcal{J}\mathcal{L}_{2})$ are spanned by four squared eigenfunctions
    \begin{equation*}
        \mathbf{P}_{1,\lambda}(\lambda_{1}), \quad 
        \mathbf{P}_{1,\lambda}(\lambda_{2}), \quad
        \mathbf{P}_{-1,\lambda}(\lambda_{1}^{*}), \quad 
        \mathbf{P}_{-1,\lambda}(\lambda_{2}^{*}).
    \end{equation*}
    Specifically, we have
    \begin{equation*}
        \begin{split}
        \mathcal{J}\mathcal{L}_{2}\mathbf{P}_{1,\lambda}(\lambda_{1})&=
        64b_{1}(b_{1}^{2}-b_{2}^{2})\mathbf{P}_{1}(\lambda_{1}), \quad
        \mathcal{J}\mathcal{L}_{2}\mathbf{P}_{-1,\lambda}(\lambda_{1}^{*})=
        64b_{1}(b_{1}^{2}-b_{2}^{2})\mathbf{P}_{-1}(\lambda_{1}^{*}), \\
        \mathcal{J}\mathcal{L}_{2}\mathbf{P}_{1,\lambda}(\lambda_{2})&=
        -64b_{2}(b_{1}^{2}-b_{2}^{2})\mathbf{P}_{1}(\lambda_{2}), \quad
        \mathcal{J}\mathcal{L}_{2}\mathbf{P}_{-1,\lambda}(\lambda_{2}^{*})=
        -64b_{2}(b_{1}^{2}-b_{2}^{2})\mathbf{P}_{-1}(\lambda_{2}^{*}). \\
        \end{split}
    \end{equation*}
\end{thm}

\begin{proof}
    The point spectrum consists of squared eigenfunctions corresponding to 
    $\lambda \in \{\lambda_{1}, \lambda_{2},$ $ \lambda_{1}^{*}, \lambda_{2}^{*}\}$. 
    For these values of \(\lambda\), the FMS becomes singular, leading to only eight independent 
    eigenfunctions in the kernel by Proposition \ref{asy-squ-eig-2} (see also Remark \ref{rem-asy-squ-eig-2}). The essential spectrum is a direct consequence of Weyl's essential spectrum theorem. The polynomial satisfies
    \begin{equation*}
        \mathcal{P}_{2}(\lambda) = ((\lambda - a)^{2} + b_{1}^{2})((\lambda - a)^{2} + b_{2}^{2}) \geq b_{1}^{2}b_{2}^{2},
    \end{equation*}
    for \(\lambda \in \mathbb{R}\), and the lower bound is attained if and only if \(\lambda = a\).
\end{proof}

\begin{rem}
	Compared to the spectrum of the operator $\mathcal{J}\mathcal{L}_{1}$ in Theorem \ref{spectral-stability}, the spectrum of the operator $\mathcal{J}\mathcal{L}_{2}$ in Theorem \ref{thm-J-L2} does not include a pair of purely imaginary eigenvalues (isolated or embedded) of negative Krein signature. This is related to the fact that compared to the four negative eigenvalues of the operator $\mathcal{L}_{1}$, the operator 
	$\mathcal{L}_{2}$ has only two negative eigenvalues, see Theorem \ref{spectrum-L2-1} and Remark \ref{rem-spectrum-L1}.
\end{rem}

\begin{rem}\label{rem-asy-squ-eig-2} 
In Proposition \ref{asy-squ-eig-2}, we require the assumption $c_{11},c_{21},c_{12}, c_{22} \neq 0$.  
If $c_{11},c_{12}\ne 0$ with zero $c_{21}$ or $c_{22}$ or both, some of squared eigenfunctions will be vanishing.  We can perform a similar calculation as in Proposition \ref{asy-squ-eig-2} with a proper adjustment. 
For example, if $c_{21}=0$, the functions $\mathbf{P}_{1}(\lambda_{1}^*)$ 
and $\mathbf{P}_{-1}(\lambda_{1}^*)$ are linearly dependent but $\mathbf{P}_{2,\lambda}(\lambda_{1})$ is linearly independent from 
other functions in $\mathrm{Ker}(\mathcal{J}\mathcal{L}_{2})$. By straightforward calculations, the following squared eigenfunctions in $\mathrm{Ker}(\mathcal{J}\mathcal{L}_{2})$ are linearly independent if $c_{21}=0$ and $c_{11},c_{12},c_{22}\ne 0$:
    \begin{equation*}
        \mathbf{P}_{1}(\lambda_{1}), \quad 
        \mathbf{P}_{1}(\lambda_{2}), \quad
        \mathbf{P}_{-1}(\lambda_{1}^{*}), \quad
        \mathbf{P}_{-1}(\lambda_{2}^{*}), \quad 
        \mathbf{P}_{2,\lambda}(\lambda_{1}), \quad 
        \mathbf{P}_{-1}(\lambda_{2}), \quad
        \mathbf{P}_{-2}(\lambda_{1}), \quad 
        \mathbf{P}_{-2,\lambda}(\lambda_{1}^{*}).
    \end{equation*}
Similarly, the following squared eigenfunctions in $\mathrm{Ker}(\mathcal{J}\mathcal{L}_{2})$ are linearly independent if $c_{22}=0$ and $c_{11},c_{12},c_{21}\ne 0$:
    \begin{equation*}
        \mathbf{P}_{1}(\lambda_{1}), \quad 
        \mathbf{P}_{1}(\lambda_{2}), \quad
        \mathbf{P}_{-1}(\lambda_{1}^{*}), \quad
        \mathbf{P}_{-1}(\lambda_{2}^{*}), \quad 
        \mathbf{P}_{2,\lambda}(\lambda_{2}), \quad 
        \mathbf{P}_{-1}(\lambda_{1}), \quad
        \mathbf{P}_{-2}(\lambda_{2}), \quad 
        \mathbf{P}_{-2,\lambda}(\lambda_{2}^{*})
    \end{equation*}
and, finally, if $c_{21}=c_{22}=0$ and $c_{11},c_{12}\ne 0$:
    \begin{equation*}
        \mathbf{P}_{1}(\lambda_{1}), \quad 
        \mathbf{P}_{1}(\lambda_{2}), \quad
        \mathbf{P}_{-1}(\lambda_{1}^{*}), \quad
        \mathbf{P}_{-1}(\lambda_{2}^{*}), \quad 
        \mathbf{P}_{2,\lambda}(\lambda_{1}), \quad 
        \mathbf{P}_{2,\lambda}(\lambda_{2}), \quad
        \mathbf{P}_{-2,\lambda}(\lambda_{1}^{*}), \quad 
        \mathbf{P}_{-2,\lambda}(\lambda_{2}^{*}).
    \end{equation*}
In all these cases, the generalized eigenfunctions remain the same as in Theorem \ref{thm-J-L2}.
\end{rem}

\subsection{The closure relation and orthogonal condition}

The squared eigenfunctions satisfy the closure relation, which implies that the set of squared eigenfunctions forms a complete basis in the \(L^{2}\) space. 
Using the closure relation, any perturbation in \(L^{2}\) can be expressed in terms of the squared eigenfunctions.

The closure relation and orthogonal conditions relate the spectrum of \(\mathcal{J}\mathcal{L}_{2}\) to the spectrum of \(\mathcal{L}_{2}\). To compute the quadratic form \((\mathcal{L}_{2} \cdot, \cdot)\) 
for squared eigenfunctions, it is sufficient to calculate 
\((\mathcal{J}^{-1} \cdot, \cdot) = -(\mathcal{J} \cdot, \cdot)\) for squared eigenfunctions.
Since
\begin{equation*}
    (\mathcal{J}\mathcal{L}_{2})^{*} = -\mathcal{L}_{2}\mathcal{J},
\end{equation*}
the adjoint squared eigenfunctions \(\mathcal{J}\mathbf{P}_{i}\) satisfy the spectral problem
\begin{equation*}
    -\mathcal{L}_{2}\mathcal{J}(\mathcal{J}\mathbf{P}_{\pm i}) = \mp 32{\rm i}\mathcal{P}_{2}(\lambda)
    (\mathcal{J}\mathbf{P}_{\pm i}),\quad i=1,2.
\end{equation*}
Thus, computing the quadratic form for squared eigenfunctions relies on the analysis of the inner products between squared eigenfunctions and 
adjoint squared eigenfunctions. This approach allows us to determine the number of negative eigenvalues and the kernel of \(\mathcal{L}_{2}\).
The closure relation (see Appendix \ref{C.colsure}) leads to the decomposition
\begin{equation*}
    L^{2}(\mathbb{R},\mathbb{C}^4) = \mathbb{E} + \mathrm{gKer}(\mathcal{J}\mathcal{L}_{2}),
\end{equation*}
where
\begin{equation*}
    \mathbb{E} = \mathrm{span} \left\{
        \int_{\mathbb{R}} w_{i}(\lambda)\mathbf{P}_{i}(\lambda,x) \, \mathrm{d}\lambda, \; \;  
        w_{i} \in L^{2}(\mathbb{R},\mathbb{C}), \;\; i = \pm 1, \pm 2
    \right\}.
\end{equation*}
The sum is not a direct sum with respect to the inner product. 
However, orthogonality conditions apply to squared eigenfunctions 
both in the continuous and discrete spectra. The following proposition specifies the orthogonality conditions between the squared eigenfunctions in the continuous spectrum, where we use the Kronecker symbol:
$\delta_{ij}=1$ if $i=j$ and $\delta_{ij}=0$ if $i \ne j$. 

\begin{prop}
    The following orthogonality relations hold:
    \begin{equation}\label{ort-cod}
        \int_{\mathbb{R}}\mathbf{P}_{i}^{\dagger}(\lambda,x) \mathcal{J}\mathbf{P}_{j}(\lambda',x) \, \mathrm{d}x = {\rm i}\pi \, |(\lambda - \lambda_{1})(\lambda - \lambda_{2})|^{4} \delta(\lambda - \lambda')\delta_{ij},
    \end{equation}
    for \( i,j \in \{ \pm 1, \pm 2\} \) and $\lambda,\lambda'\in\mathbb{R}$. 
    Moreover, the orthogonality conditions
    \begin{equation}\label{ort-cod-continuous-discrete}
        (\mathcal{L}_{2}\mathbf{u}, \mathbf{v}) = 0, \quad \mathbf{u} \in \mathbb{E}, \;\; \mathbf{v} \in \mathrm{gKer}(\mathcal{J}\mathcal{L}_{2})
    \end{equation}
    also hold.
\end{prop}

\begin{proof}
Appendix \ref{C.colsure} and the expression for \(\mathbf{S}(\lambda)\) give the relations
    \begin{equation*}
        \mathbf{O}_{j} = \frac{s_{11}(\lambda)}{\mathcal{P}_{2}(\lambda)}\mathbf{P}_{j}, \quad 
        \mathbf{O}_{j+2} = \frac{s_{11}^{-1}(\lambda)}{\mathcal{P}_{2}(\lambda)}\mathbf{P}_{-j}, \quad j = 1, 2.
    \end{equation*}
The orthogonality relation \eqref{ort-cod} is obtained from \eqref{orth-cod-general} and \eqref{S-matrix-non-2}. Since the bases for \(\mathbb{E}\) and \(\mathrm{gKer}(\mathcal{J}\mathcal{L}_{2})\) satisfy the spectral problem for the operator \(\mathcal{J}\mathcal{L}_{2}\), the inner product \((\mathcal{L}_{2}\mathbf{u}, \mathbf{v})\) can be transformed into \((\mathcal{J}\mathbf{u}, \tilde{\mathbf{v}})\) for some \(\tilde{\mathbf{v}} \in \mathrm{gKer}(\mathcal{J}\mathcal{L}_{2})\). The conditions \eqref{ort-cod-continuous-discrete} follow from the closure relations when the continuous and discrete spectra are not empty.
\end{proof}

It remains to obtain the orthogonal condition between squared eigenfunctions 
and adjoint squared eigenfunctions in the point spectrum. 
We use here the squared eigenfunction matrix \eqref{squared-eigenfunction-matrix} defined by $\mathbf{\Phi}$. 
The singularities of $\mathbf{\Phi}$ are eliminated after a multiplication of $\mathbf{\Phi}$ by a constant factor. In view of \eqref{diff-L}, the matrix
\begin{equation*}
    (\mathbf{G}(\eta)\mathbf{F}(\lambda))_{x} = {\rm i}(\lambda - \eta) \mathbf{G}(\eta) \sigma_{3} \mathbf{F}(\lambda)
    - {\rm i} \lambda \mathbf{G}(\eta) \mathbf{F}(\lambda) \sigma_{3} + {\rm i} \eta \sigma_{3} \mathbf{G}(\eta) \mathbf{F}(\lambda)
    + {\rm i} [\mathbf{Q}, \mathbf{G}(\eta) \mathbf{F}(\lambda)]
\end{equation*}
leads to
\begin{equation}\label{trace-1}
    \mathrm{Tr}((\mathbf{G}(\eta)\mathbf{F}(\lambda))_{x}) =
    {\rm i}(\lambda - \eta) \left(\mathrm{Tr}(\mathbf{G}(\eta) \sigma_{3} \mathbf{F}(\lambda)) - \mathrm{Tr}(\sigma_{3} \mathbf{G}(\eta) \mathbf{F}(\lambda))\right).
\end{equation}
Let $\lambda_0, \eta_0 \in\mathbb{C}$ be arbitrary. Substituting
\begin{equation*}
    \mathbf{F}(\lambda)=\sum_{i=0}^{4}\mathbf{F}_{i}(\lambda-\lambda_{0})^{i},\quad 
    \mathbf{G}(\lambda)=\sum_{i=0}^{4}\mathbf{G}_{i}(\eta-\eta_{0})^{i},
\end{equation*}
into \eqref{trace-1} and grouping the terms with respect to $(\eta-\eta_{0})^{j}(\lambda-\lambda_{0})^{i}$,
we obtain
\begin{align*}
%\label{trace-2}
        \mathrm{Tr}(\mathbf{G}_{0}\mathbf{F}_{0})_{x}=&{\rm i}(\lambda_{0}-\eta_{0})(\mathrm{Tr}(\mathbf{G}_{0}\sigma_{3}\mathbf{F}_{0})-\mathrm{Tr}(\sigma_{3}\mathbf{G}_{0}\mathbf{F}_{0})),\\
        \mathrm{Tr}(\mathbf{G}_{j}\mathbf{F}_{0})_{x}=&
    {\rm i}(\lambda_{0}-\eta_{0})(\mathrm{Tr}(\mathbf{G}_{j}\sigma_{3}\mathbf{F}_{0})-\mathrm{Tr}(\sigma_{3}\mathbf{G}_{j}\mathbf{F}_{0}))-
    {\rm i}(\mathrm{Tr}(\mathbf{G}_{j-1}\sigma_{3}\mathbf{F}_{0})-\mathrm{Tr}(\sigma_{3}\mathbf{G}_{j-1}\mathbf{F}_{0})),\\
    \mathrm{Tr}(\mathbf{G}_{0}\mathbf{F}_{i})_{x}=&
    {\rm i}(\lambda_{0}-\eta_{0})(\mathrm{Tr}(\mathbf{G}_{0}\sigma_{3}\mathbf{F}_{i})-\mathrm{Tr}(\sigma_{3}\mathbf{G}_{0}\mathbf{F}_{i}))+
    {\rm i}(\mathrm{Tr}(\mathbf{G}_{0}\sigma_{3}\mathbf{F}_{i-1})-\mathrm{Tr}(\sigma_{3}\mathbf{G}_{0}\mathbf{F}_{i-1})),\\
        \mathrm{Tr}(\mathbf{G}_{j}\mathbf{F}_{i})_{x}=&
    {\rm i}(\lambda_{0}-\eta_{0})(\mathrm{Tr}(\mathbf{G}_{j}\sigma_{3}\mathbf{F}_{i})-\mathrm{Tr}(\sigma_{3}\mathbf{G}_{j}\mathbf{F}_{i}))\\&+
    {\rm i}(\mathrm{Tr}(\mathbf{G}_{j}\sigma_{3}\mathbf{F}_{i-1})-\mathrm{Tr}(\sigma_{3}\mathbf{G}_{j}\mathbf{F}_{i-1}))-
    {\rm i}(\mathrm{Tr}(\mathbf{G}_{j-1}\sigma_{3}\mathbf{F}_{i})-\mathrm{Tr}(\sigma_{3}\mathbf{G}_{j-1}\mathbf{F}_{i}))
\end{align*}
for $1\leq i,j\leq 4$. This yields for $\lambda_{0}\ne\eta_{0}$, 
\begin{align*}
        \frac{{\rm i}}{2}(\mathrm{Tr}(\mathbf{G}_{0}\sigma_{3}\mathbf{F}_{0})-\mathrm{Tr}(\sigma_{3}\mathbf{G}_{0}\mathbf{F}_{0}))=&
    \frac{1}{2(\lambda_{0}-\eta_{0})}\mathrm{Tr}(\mathbf{G}_{0}\mathbf{F}_{0})_{x}, \\
    \frac{{\rm i}}{2}(\mathrm{Tr}(\mathbf{G}_{1}\sigma_{3}\mathbf{F}_{0})-\mathrm{Tr}(\sigma_{3}\mathbf{G}_{1}\mathbf{F}_{0}))=&
    \frac{1}{2(\lambda_{0}-\eta_{0})}\mathrm{Tr}(\mathbf{G}_{1}\mathbf{F}_{0}+\frac{1}{\lambda_{0}-\eta_{0}}\mathbf{G}_{0}\mathbf{F}_{0})_{x}. 
    \end{align*}
and for $\lambda_{0}=\eta_{0}$, 
\begin{align*}
        \frac{{\rm i}}{2}(\mathrm{Tr}(\mathbf{G}_{0}\sigma_{3}\mathbf{F}_{0})-\mathrm{Tr}(\sigma_{3}\mathbf{G}_{0}\mathbf{F}_{0}))=&
        -\frac{1}{2}\mathrm{Tr}(\mathbf{G}_{1}\mathbf{F}_{0})_{x}, \\
    \frac{{\rm i}}{2}(\mathrm{Tr}(\mathbf{G}_{1}\sigma_{3}\mathbf{F}_{0})-\mathrm{Tr}(\sigma_{3}\mathbf{G}_{1}\mathbf{F}_{0}))=&
    -\frac{1}{2}\mathrm{Tr}(\mathbf{G}_{2}\mathbf{F}_{0})_{x}. 
    \end{align*}
Let us denote
\begin{equation*}
    \mathbf{F}=\begin{pmatrix}
        f_{1} & \mathbf{h}_{1}^{T} \\
        \mathbf{g}_{1} & \mathbf{K}_{1}
    \end{pmatrix},\quad \mathbf{G}=\begin{pmatrix}
        f_{2}^{*} & \mathbf{g}^{\dagger}_{2} \\
        \mathbf{h}_{2}^{*} & \mathbf{K}_{2}^{\dagger}
    \end{pmatrix}.
\end{equation*}
Then, we obtain 
\begin{align*}
        \mathrm{Tr}(\mathbf{G}(\eta)\sigma_{3}\mathbf{F}(\lambda))-\mathrm{Tr}(\sigma_{3}\mathbf{G}(\eta)\mathbf{F}(\lambda)) &=
    2(-\mathbf{g}^{\dagger}_{2}\mathbf{g}_{1}+\mathbf{h}^{\dagger}_{2}\mathbf{h}_{1}), \\
    \mathrm{Tr}(\mathbf{G}(\eta)\mathbf{F}(\lambda))&=f_{1}f_{2}^{*}+
    \mathbf{g}^{\dagger}_{2}\mathbf{g}_{1}+\mathbf{h}^{\dagger}_{2}\mathbf{h}_{1}+
    \mathrm{Tr}(\mathbf{K}_{2}^{\dagger}\mathbf{K}_{1}),
    \end{align*}
which yields the following lemma.

\begin{lem}\label{thm-orthogonal-condition-discrete}
    Consider the squared eigenfunction matrix $p_{i}(\mathbf{\Phi})$ defined in \eqref{squared-eigenfunction-matrix} and the squared eigenfunctions $s_{i}(\mathbf{\Phi})$ defined in \eqref{squared-eigenfunctions}, associated with the spectral problem \eqref{spectral-problem} with the symmetric potential $\mathbf{Q} = \mathbf{Q}^{\dagger}$. The identity
    \begin{equation*}
    %\label{discrete-1}
        2(\lambda - \eta) s_{j}(\mathbf{\Phi})(\eta^{*})^{\dagger} \mathcal{J} s_{i}(\mathbf{\Phi})(\lambda) = \mathrm{Tr}\left(p_{j}^{\dagger}(\mathbf{\Phi})(\eta^{*}) p_{i}(\mathbf{\Phi})(\lambda)\right)_{x}
    \end{equation*}
    holds for any spectral parameters $\lambda$ and $\eta$. If $\lambda = \lambda_{0}$ and $\eta = \eta_{0}$, where $\lambda_{0}$ and $\eta_{0}$ are eigenvalues of the spectral problem \eqref{spectral-problem} and $\lambda_{0} \ne \eta_{0}$, the integrals between the squared eigenfunctions and the adjoint squared eigenfunctions in the point spectrum is given by
    \begin{align*}
            \int_{\mathbb{R}} s_{j}(\mathbf{\Phi})^{\dagger}(\eta_{0}^{*}) \mathcal{J} s_{i}(\mathbf{\Phi})(\lambda_{0})\mathrm{d}x =& \frac{1}{2(\lambda_{0} - \eta_{0})} \mathrm{Tr}\left(p_{j}^{\dagger}(\mathbf{\Phi})(\eta_{0}^{*}) p_{i}(\mathbf{\Phi})(\lambda_{0})\right) \bigg|_{-\infty}^{+\infty}, \\
            \int_{\mathbb{R}} s_{j,\eta}(\mathbf{\Phi})^{\dagger}(\eta_{0}^{*})\mathcal{J} s_{i}(\mathbf{\Phi})(\lambda_{0})\mathrm{d}x =& \frac{1}{2(\lambda_{0} - \eta_{0})} \mathrm{Tr}\left(p_{j,\eta}^{\dagger}(\mathbf{\Phi})(\eta_{0}^{*}) p_{i}(\mathbf{\Phi})(\lambda_{0}) 
            \right.\\&\left.+ \frac{1}{\lambda_{0} - \eta_{0}} p_{j}^{\dagger}(\mathbf{\Phi})(\eta_{0}^{*}) p_{i}(\mathbf{\Phi})(\lambda_{0})\right) \bigg|_{-\infty}^{+\infty}.
        \end{align*}
    For $\lambda_{0} = \eta_{0}$, the integrals are given by
    \begin{align*}
        \int_{\mathbb{R}} s_{j}(\mathbf{\Phi})^{\dagger}(\eta_{0}^{*}) \mathcal{J} s_{i}(\mathbf{\Phi})(\lambda_{0})\mathrm{d}x &= -\frac{1}{2} \mathrm{Tr}\left(p_{j,\eta}^{\dagger}(\mathbf{\Phi})(\eta_{0}^{*}) p_{i}(\mathbf{\Phi})(\lambda_{0})\right) \bigg|_{-\infty}^{+\infty}, \\
        \int_{\mathbb{R}} s_{j,\eta}(\mathbf{\Phi})^{\dagger}(\eta_{0}^{*})\mathcal{J} s_{i}(\mathbf{\Phi})(\lambda_{0})\mathrm{d}x &= -\frac{1}{4} \mathrm{Tr}\left(p_{j,\eta\eta}^{\dagger}(\mathbf{\Phi})(\eta_{0}^{*}) p_{i}(\mathbf{\Phi})(\lambda_{0})\right) \bigg|_{-\infty}^{+\infty}.
        \end{align*}
\end{lem}

We use Lemma \ref{thm-orthogonal-condition-discrete} to calculate  $(\mathcal{L}_{2}\cdot,\cdot)$ in the generalized kernel of $\mathcal{J}\mathcal{L}_{2}$. 

\subsection{The inner product for the generalized kernel of $\mathcal{J}\mathcal{L}_{2}$}

It remains to calculate the quadratic form $(\mathcal{L}_{2} \cdot, \cdot)$ in the subspace spanned by the generalized eigenfunctions for the zero eigenvalue 
of $\mathcal{J}\mathcal{L}_{2}$:
\begin{equation*}
    \mathrm{span}\left\{
    \mathbf{P}_{1,\lambda}(\lambda_{1}), 
    \mathbf{P}_{1,\lambda}(\lambda_{2}),
    \mathbf{P}_{-1,\lambda}(\lambda_{1}^{*}),
    \mathbf{P}_{-1,\lambda}(\lambda_{2}^{*})
    \right\}.
\end{equation*}
Equivalently, we need to analyze the quadratic form $(\mathcal{J} \cdot, \cdot)$ 
using the relation between $\mathcal{J}\mathcal{L}_{2}$ and $\mathcal{L}_{2}$. 
Recall the differential form \eqref{diff-form} which is skew-symmetric, 
i.e., $\omega^{\dagger}(\mathbf{h}, \mathbf{f}) = -\omega(\mathbf{f}, \mathbf{h})$. 
The quadratic form becomes
\begin{equation*}
    (\mathbf{h}, \mathcal{L}_{2} \mathbf{f}) = -\mathrm{Re} \int_{\mathbb{R}} \omega(\mathbf{h}, \mathcal{J} \mathcal{L}_{2} \mathbf{f}).
\end{equation*}
Hence we need to calculate the integral of
\begin{equation*}
%\label{simplify-inner product- L2}
    \omega(\mathbf{h},\mathbf{g}), \quad \mathbf{g}\in\mathrm{span}\{
        \mathbf{P}_{1}(\lambda_{1}), 
            \mathbf{P}_{1}(\lambda_{2}) ,
            \mathbf{P}_{-1}(\lambda_{1}^{*}), 
            \mathbf{P}_{-1}(\lambda_{2}^{*})\},\,\, 
    \mathbf{h}\in \mathrm{gKer}(\mathcal{J}\mathcal{L}_{2})\backslash 
    \mathrm{Ker}(\mathcal{J}\mathcal{L}_{2})
\end{equation*}
on the real line. 
It is noted that if $\mathbf{h}\in\mathrm{Ker}(\mathcal{J}\mathcal{L}_{2})
=\mathrm{Ker}(\mathcal{L}_{2})$, then the self-adjoint operator 
$\mathcal{L}_{2}$ induces
\begin{equation*}
    (\mathbf{h},\mathcal{L}_{2}\mathbf{f})=(\mathcal{L}_{2}\mathbf{h},\mathbf{f})=0. 
\end{equation*}
We use Lemma \ref{thm-orthogonal-condition-discrete} to obtain 
the following proposition.

\begin{prop}\label{negative-1}
If $c_{11},c_{12} \neq 0$ for the breather solutions of Definition \ref{def-breather}, then the matrix is given by
    \begin{equation*}
        \left(\int_{\mathbb{R}}\omega(\mathbf{h},\mathbf{g})\right)_{4\times 4}=2(b_{1}^{2}-b_{2}^{2})^{2}\begin{pmatrix}
            0 & b_{1}^{2} &0 &0\\
            b_{1}^{2}& 0 &0 &0\\
            0 & 0 &0 &b_{2}^{2}\\
            0& 0 &b_{2}^{2}& 0\\
        \end{pmatrix},
    \end{equation*}
    where 
    \begin{align*}
            &\mathbf{h}\in\mathrm{span}\left\{\frac{\mathbf{P}_{1,\lambda}(\lambda_{1})}{\ii c_{11}},
            \frac{\mathbf{P}_{-1,\lambda}(\lambda_{1}^{*})}{-\ii c_{11}^{*}},
            \frac{\mathbf{P}_{1,\lambda}(\lambda_{2})}{\ii c_{12}},
            \frac{\mathbf{P}_{-1,\lambda}(\lambda_{2}^{*})}{-\ii c_{12}^{*}}\right\},
            \\
            &\mathbf{g}\in\mathrm{span}\left\{\frac{\mathbf{P}_{1}(\lambda_{1})}{\ii c_{11}},
            \frac{\mathbf{P}_{-1}(\lambda_{1}^{*})}{-\ii c_{11}^{*}},
            \frac{\mathbf{P}_{1}(\lambda_{2})}{\ii c_{12}},
            \frac{\mathbf{P}_{-1}(\lambda_{2}^{*})}{-\ii c_{12}^{*}}\right\}.
        \end{align*}
\end{prop}

\begin{proof}
    We provide the proof for $\int_{\mathbb{R}}\mathbf{P}_{-1,\lambda}^{\dagger}(\lambda_{1}^{*})\mathcal{J} \mathbf{P}_{1}(\lambda_{1})\mathrm{d}x$. The proof for all other entries is analogous. Lemma \ref{thm-orthogonal-condition-discrete} and the symmetry condition \eqref{symmetric-squared-eigenfunctions} imply the identity
    \begin{align*}
        \int_{\mathbb{R}}\mathbf{P}_{-1,\lambda}^{\dagger}(\lambda_{1}^{*})\mathcal{J} \mathbf{P}_{1}(\lambda_{1})\mathrm{d}x &= -\frac{1}{4} \mathrm{Tr} \left(p_{-1,\lambda \lambda}^{\dagger} (\mathbf{\Phi}_{r}^{[2]})(\lambda_{1}^{*}) p_{1} (\mathbf{\Phi}_{r}^{[2]})(\lambda_{1}) \right) \bigg|_{-\infty}^{+\infty} \\
            &= -\frac{1}{4} \mathrm{Tr} \left(p_{1,\lambda \lambda} (\mathbf{\Phi}_{r}^{[2]})(\lambda_{1}) p_{1} (\mathbf{\Phi}_{r}^{[2]})(\lambda_{1}) \right) \bigg|_{-\infty}^{+\infty}.
    \end{align*}
    The second derivative of squared eigenfunction matrix $p_{1} (\mathbf{\Phi}_{r}^{[2]})$ at $\lambda=\lambda_{1}$ is given by
    \begin{equation*}
        p_{1} (\mathbf{\Phi}_{r}^{[2]})_{\lambda \lambda} (\lambda_{1}) \sim 
        {\rm e}^{2{\rm i}a x} {\rm e}^{-2b_{1} x} \left(p_{1} (\mathbf{D}_{r}^{\pm})_{\lambda \lambda} (\lambda_{1}) +
        4{\rm i}x p_{1} (\mathbf{D}_{r}^{\pm})_{\lambda} (\lambda_{1}) -
        4x^{2} p_{1} (\mathbf{D}_{r}^{\pm}) (\lambda_{1}) \right)
    \end{equation*}
    as $x \to \pm \infty$. 
    As $x \to +\infty$, we have 
    \begin{equation*}
        \mathrm{Tr} \left(p_{1,\lambda \lambda} (\mathbf{\Phi}_{r}^{[2]})(\lambda_{1}) p_{1} (\mathbf{\Phi}_{r}^{[2]})(\lambda_{1}) \right) \to 0
    \end{equation*}
    due to the term \({\rm e}^{-2b_{1} x}\) and the fact that \(D_{r,ij}^{+}\) are bounded. Now we consider the 
    case $x\to -\infty$. It follows from the expression of the second derivatives of the matrix $\mathbf{D}_{r}^{-}$ that
    \begin{equation*}
        (\mathbf{D}_{r}^{-})_{\lambda \lambda} = 2 .
    \end{equation*}
Since
    \begin{align*}
        p_{1} (\mathbf{D}_{r}^{\pm})_{\lambda \lambda} (\lambda_{1}) =& (\mathbf{\mathbf{D}}_{r,\lambda \lambda}^{\pm} (\lambda_{1}))_{1} \left((\mathbf{\mathbf{D}}_{r}^{\pm})^{\dagger} (\lambda_{1}^{*})\right)^{2} +
        2 (\mathbf{\mathbf{D}}_{r,\lambda}^{\pm} (\lambda_{1}))_{1} \left((\mathbf{\mathbf{D}}_{r,\lambda}^{\pm})^{\dagger} (\lambda_{1}^{*})\right)^{2} \\& \quad 
        + (\mathbf{\mathbf{D}}_{r}^{\pm} (\lambda_{1}))_{1} \left((\mathbf{\mathbf{D}}_{r,\lambda \lambda}^{\pm})^{\dagger} (\lambda_{1}^{*})\right)^{2},
        \end{align*}
we collect the constant terms and obtain
    \begin{align*}
            \mathrm{Tr} \left(p_{1,\lambda \lambda} (\mathbf{\Phi}_{r}^{[2]})(\lambda_{1}) p_{1} (\mathbf{\Phi}_{r}^{[2]})(\lambda_{1})\right)
            &\to 2(b_{1} - b_{2})^{2} \left(
                \frac{1}{M^{-}} 2{\rm i} b_{1}c_{11}  \frac{b_{1} - b_{2}}{b_{1} + b_{2}} (b_{1} - b_{2})
            \right)^{2} \\
            &= -8 (b_{1}^{2} - b_{2}^{2})^{2} c_{11}^{2} b_{1}^{2}, \quad \mbox{\rm as} \;\; x \to -\infty.
    \end{align*}
This concludes the proof for $(\mathbf{P}_{-1,\lambda}(\lambda_{1}^{*}), \mathcal{J} \mathbf{P}_{1}(\lambda_{1}))$. 
\end{proof}

\begin{rem}
	If $c_{12},c_{21} = 0$ and $c_{11},c_{22} \neq 0$ for the non-degenerate vector soliton solutions, 
    the generalized kernel for $\mathcal{J}\mathcal{L}_{2}$ is 
    \begin{equation*}
        \mathrm{span}\{
	\hat{\mathbf{P}}_{1,\lambda}(\lambda_{1}), 
	\hat{\mathbf{P}}_{1,\lambda}(\lambda_{1}^{*}) ,
	\hat{\mathbf{P}}_{2,\lambda}(\lambda_{2}), 
	\hat{\mathbf{P}}_{2,\lambda}(\lambda_{2}^{*})\},
    \end{equation*}
and the matrix is given by 
	\begin{equation*}
	\left(\int_{\mathbb{R}}\omega(\mathbf{h},\mathbf{g})\right)_{4\times 4}=2(b_{1}^{2}-b_{2}^{2})\begin{pmatrix}
	0 & -b_{1}^{2} &0 &0\\
	b_{1}^{2}& 0 &0 &0\\
	0 & 0 &0 &b_{2}^{2}\\
	0& 0 &-b_{2}^{2}& 0\\
	\end{pmatrix},
	\end{equation*}
	where 
	\begin{align*}
	&\mathbf{h}\in\mathrm{span}\{
	\hat{\mathbf{P}}_{1,\lambda}(\lambda_{1}), 
	\hat{\mathbf{P}}_{1,\lambda}(\lambda_{1}^{*}) ,
	\hat{\mathbf{P}}_{2,\lambda}(\lambda_{2}), 
	\hat{\mathbf{P}}_{2,\lambda}(\lambda_{2}^{*})\},
	\\
	&\mathbf{g}\in\mathrm{span}\{
	\hat{\mathbf{P}}_{1}(\lambda_{1}), 
	\hat{\mathbf{P}}_{1}(\lambda_{1}^{*}) ,
	\hat{\mathbf{P}}_{2}(\lambda_{2}), 
	\hat{\mathbf{P}}_{2}(\lambda_{2}^{*})\},
	\end{align*}
	due to different definitions of the FMS for breathers and non-degenerate vector solitons
		\begin{align*}
		\mathbf{\Phi}^{[2]}_{r,non} = \mathbf{\Phi}^{[2]}|_{t=0} \left( 
		\begin{matrix}
		(\lambda-\lambda_{1}^{*})(\lambda-\lambda_{2}^{*}) & 0 & 0 \\
		0 & \lambda-\lambda_{1}^{*} & 0 \\ 0 & 0 & \lambda-\lambda_{2}^{*}
\end{matrix}		\right)
		\end{align*}
		versus
		\begin{align*} \mathbf{\Phi}^{[2]}_{r}&=(\lambda-\lambda_{1}^{*})(\lambda-\lambda_{2}^{*})\mathbf{\Phi}^{[2]}|_{t=0}.
		\end{align*}
\end{rem}

\subsection{The spectrum for \(\mathcal{L}_{2}\)}

The number of negative eigenvalues of $\mathcal{L}_{2}$ is given by the following theorem. 

\begin{thm}\label{spectrum-L2-1}
    For breather solutions \eqref{breather} with $c_{11},c_{12} \neq 0$, the operator
    $\mathcal{L}_{2}(\mathbf{q}^{[2]})$ defined by \eqref{L2} in 
    $L^{2}(\mathbb{R},\mathbb{C}^4)$ has 
    two negative eigenvalues (counting multiplicities) and the zero 
    eigenvalue of multiplicity eight. 
\end{thm}

\begin{proof}
    We consider the breather solutions \eqref{breather} with $c_{11}, c_{12} \ne 0$. Define the cone 
\begin{equation*}
    \mathcal{N}_2 = \left\{ \mathbf{u} \in H^{2}(\mathbb{R},\mathbb{C}^4) : \quad (\mathcal{L}_{2} \mathbf{u}, \mathbf{u}) < 0 \right\}
\end{equation*}
and denote the dimension of the maximal linear subspace contained in $\mathcal{N}_2$ by $\mathrm{dim}(\mathcal{N}_2)$. The dimension of $\mathcal{N}_2$ is equal to the 
number of negative eigenvalues of $\mathcal{L}_{2}$.
To calculate $\mathrm{dim}(\mathcal{N}_2)$, 
we use the closure relation. Any function $\mathbf{v} \in L^{2}(\mathbb{R},\mathbb{C}^4)$ 
can be decomposed as
\begin{equation*}
    \mathbf{v}(x) = \sum_{j=1}^{4} \int_{\mathbb{R}} w_{j}(\lambda) \mathbf{P}_{j}(x; \lambda) \, \mathrm{d}\lambda + \alpha \mathbf{R}(x),
\end{equation*}
where $w_{j}(\lambda)$ are functions of $\lambda$, $\alpha$ is a constant, 
and $\mathbf{R} \in \mathrm{gKer}(\mathcal{J} \mathcal{L}_{2})$.
Since the essential spectrum and the point spectrum are orthogonal with respect to the quadratic form $(\mathcal{L}_{2} \cdot, \cdot)$, we will consider the continuous and discrete parts separately.

For the continuous part, we have by \eqref{ort-cod},
\begin{align*}
    &(\mathcal{L}_{2} \int_{\mathbb{R}} w_{i}(\lambda) \mathbf{P}_{i}(\lambda;x) \, \mathrm{d}\lambda, 
    \int_{\mathbb{R}} w_{j}(\lambda') \mathbf{P}_{j}(\lambda';x) \, \mathrm{d}\lambda') \\
    &= (-32 \, \mathrm{i} \mathcal{J} \int_{\mathbb{R}} w_{i}(\lambda) \mathcal{P}(\lambda) \mathbf{P}_{i}(\lambda;x) \, \mathrm{d}\lambda, 
    \int_{\mathbb{R}} w_{j}(\lambda') \mathbf{P}_{j}(\lambda';x) \, \mathrm{d}\lambda') \\
    &= -32 \, \mathrm{i} \int_{\mathbb{R}^{3}} w_{i}^{*}(\lambda) w_{j}(\lambda') |\mathcal{P}(\lambda)| 
    \mathbf{P}_{i}^{\dagger}(\lambda;x) \mathcal{J} \mathbf{P}_{j}(\lambda';x) \, \mathrm{d}\lambda \, \mathrm{d}\lambda' \, \mathrm{d}x \\
    &= 32 \pi \delta_{ij} \int_{\mathbb{R}} |(\lambda - \lambda_{1})(\lambda - \lambda_{2})|^{4} |w_{j}(\lambda)|^{2} |\mathcal{P}(\lambda)| \, \mathrm{d}\lambda \geq 0.
\end{align*}
This calculation indicates that there is no contribution to $\mathrm{dim}(\mathcal{N}_2)$ from the continuous part. Thus, $\mathrm{dim}(\mathcal{N}_2)$ coincides with the number of negative eigenvalues of the Hermitian matrix 
$\mathbf{H}(\mathrm{gKer}(\mathcal{J}\mathcal{L}_{2}))$ associated with the quadratic form 
$(\mathcal{L}_{2}|_{\mathrm{gKer}(\mathcal{J}\mathcal{L}_{2})} \cdot, \cdot)$. Since $\mathcal{L}_{2}(\mathbf{u})=0$ 
for $\mathbf{u}\in\mathrm{Ker}(\mathcal{L}_{2})$, it suffices to consider the Hermitian matrix 
$\mathbf{H}(\mathrm{gKer}(\mathcal{J}\mathcal{L}_{2}) \setminus \mathrm{Ker}(\mathcal{J}\mathcal{L}_{2}))$. 
By Proposition \ref{negative-1}, the Hermitian matrix is 
\begin{equation*}
    2^{7} (b_{1}^{2} - b_{2}^{2})^{3} \begin{pmatrix}
        0 & b_{1}^{3} & 0 & 0 \\
        b_{1}^{3} & 0 & 0 & 0 \\
        0 & 0 & 0 & b_{2}^{3} \\
        0 & 0 & b_{2}^{3} & 0
    \end{pmatrix}.
\end{equation*}
Thus, $\mathrm{dim}(\mathcal{N}_2) = 2$. In addition, the kernel of \(\mathcal{L}_{2}\) is eight-dimensional since $\mathrm{Ker}(\mathcal{L}_{2})=\mathrm{Ker}(\mathcal{JL}_{2}) = 8$ by Theorem \ref{thm-J-L2}.
\end{proof}

    \begin{rem}
    	\label{rem-spectrum-L1}
    	For non-degenerate vector solitons \eqref{non-sol-1}, the operator  $\mathcal{L}_1$ defined by \eqref{def-L1} in $L^2(\mathbb{R},\mathbb{C}^4)$ has four negative eigenvalues (counting multiplicities) and the zero eigenvalue of multiplicity four.  Since the number of negative eigenvalues of \(\mathcal{L}_{1}\) is equal to the dimension of \(\mathcal{N}_1\), we calculate $\dim(\mathcal{N}_1)$. 
        Following a similar argument as in the proof of Theorem \ref{spectrum-L2-1}, the dimension of \(\mathcal{N}_1\) corresponds to the number of negative eigenvalues of the matrix
        \begin{equation*}
            (\mathcal{L}_{1}\mathbf{f},\mathbf{g}),\quad \mathbf{f},\mathbf{g} \in 
            \mathrm{span}\{\mathrm{gKer}(\mathcal{J}\mathcal{L}_{1})\cup \mathrm{Ker}(\pm 2{\rm i}(b_{1}^{2}-b_{2}^{2})-\mathcal{J}\mathcal{L}_{1})\}.
        \end{equation*}
        Since the spaces \(\mathrm{gKer}(\mathcal{J}\mathcal{L}_{1})\) and 
        \(\mathrm{Ker}(\pm 2{\rm i}(b_{1}^{2}-b_{2}^{2})-\mathcal{J}\mathcal{L}_{1})\) are orthogonal with respect to the product \((\mathcal{L}_{1}\cdot,\cdot)\) by \cite{haragus_spectra_2008}, 
        we analyze these two spaces separately and switch the product to \((\cdot,\mathcal{J}\cdot)\). 
        The product \((\cdot,\mathcal{J}\cdot)\) in space \(\mathrm{gKer}(\mathcal{J}\mathcal{L}_{1})\)
was analyzed in Theorem \ref{spectrum-L2-1}. For \(\mathbf{f},\mathbf{g} \in \{\hat{\mathbf{P}}_{2}(\lambda_{1}), \hat{\mathbf{P}}_{2}(\lambda_{1}^{*}), \hat{\mathbf{P}}_{1}(\lambda_{2}), \hat{\mathbf{P}}_{1}(\lambda_{2}^{*})\}\), 
        we derive the matrix
        \begin{equation*}
            (\int_{\mathbb{R}}\omega(\mathbf{f},\mathbf{g})) = {\rm i}(b_{1}^{2}-b_{2}^{2})^{2} \begin{pmatrix}
                |c_{11}|^{2}b_{1} & -b_{1} & 0 & 0 \\
                -b_{1} & \frac{b_{1}(b_{1}-b_{2})^{2}}{|c_{11}|^{2}(b_{1}+b_{2})^{2}} & 0 & 0 \\
                0 & 0 & |c_{22}|^{2}b_{2} & -b_{2} \\
                0 & 0 & -b_{2} & \frac{b_{2}(b_{1}-b_{2})^{2}}{|c_{22}|^{2}(b_{1}+b_{2})^{2}}
            \end{pmatrix}
        \end{equation*}
        by Lemma \ref{thm-orthogonal-condition-discrete}, which yields
        \begin{equation*}
            ((\mathcal{L}_{1}\mathbf{f},\mathbf{g})) = (b_{1}^{2}-b_{2}^{2})^{3} \begin{pmatrix}
                -|c_{11}|^{2}b_{1} & b_{1} & 0 & 0 \\
                b_{1} & -\frac{b_{1}(b_{1}-b_{2})^{2}}{|c_{11}|^{2}(b_{1}+b_{2})^{2}} & 0 & 0 \\
                0 & 0 & |c_{22}|^{2}b_{2} & -b_{2} \\
                0 & 0 & -b_{2} & \frac{b_{2}(b_{1}-b_{2})^{2}}{|c_{22}|^{2}(b_{1}+b_{2})^{2}}
            \end{pmatrix},
        \end{equation*}
with two negative eigenvalues. Using the same argument as in Theorem \ref{spectrum-L2-1}, we conclude
        \begin{equation*}
            n((\mathcal{L}_{1}\mathbf{f},\mathbf{g})) = 2,\quad \mathbf{f},\mathbf{g} \in 
            \mathrm{gKer}(\mathcal{J}\mathcal{L}_{1}),
        \end{equation*}
        and
        \begin{equation*}
            n((\mathcal{L}_{1}\mathbf{f},\mathbf{g})) = 2,\quad \mathbf{f},\mathbf{g} \in 
            \mathrm{Ker}(\pm 2{\rm i}(b_{1}^{2}-b_{2}^{2})-\mathcal{J}\mathcal{L}_{1}).
        \end{equation*}
        Thus, $\dim(\mathcal{N}_1) = 4$. The kernel of $\mathcal{L}_1$ is four-dimensional since $\mathrm{Ker}(\mathcal{L}_1)=\mathrm{Ker}(\mathcal{JL}_1) = 4$ by Proposition \ref{spectrum-JL1}.       
\end{rem}

Next we consider the spectrum of $\mathcal{L}_{2}$ in the real space $\mathrm{X}$. To do so, we need to transfer the eigenfunctions from $L^{2}(\mathbb{R},\mathbb{C}^4)$ into $\mathrm{X}$.
We consider the transformation 
    \begin{equation*}
\mathcal{C}: L^{2}(\mathbb{R},\mathbb{C}^4) \to \mathrm{X} : \qquad    \mathcal{C} \mathbf{P} = \mathbf{P}+(\mathbf{\Sigma}\mathbf{P})^{*}.
    \end{equation*}
    Since 
    \begin{equation*}
        \mathbf{\Sigma}\mathcal{L}_{2}^{*}\mathbf{\Sigma}=\mathcal{L}_{2}, 
    \end{equation*}
    the operator $\mathcal{L}_{2}$ and $\mathcal{C}$ commute
    \begin{equation*}
        \mathcal{L}_{2}\mathcal{C}=
        \mathcal{C}\mathcal{L}_{2}.
    \end{equation*}
As a result, we obtain the following lemma.

    \begin{lem}\label{spectrum-L2-2}
The self-adjoint operator $\mathcal{L}_{2}(\mathbf{q}^{[2]})$ in the real Hilbert space $\mathrm{X}$ has two negative eigenvalues (counting multiplicities) and the zero eigenvalue of multiplicity eight.
    \end{lem}

    \begin{proof}
        For any eigenfunction $\mathbf{P} \in \mathrm{Ker}(\mathcal{L}_{2})$, we have 
        \begin{equation*}
            \mathcal{L}_{2} \mathcal{C}(\mathbf{P}) = \mathcal{C} \mathcal{L}_{2}(\mathbf{P}) = 0.
        \end{equation*}
        Hence, the kernel of $\mathcal{L}_{2}$ in $\mathrm{X}$ is 
        \begin{equation*}
            \mathrm{Ker}_{\mathrm{X}}(\mathcal{L}_{2}) = \mathrm{span}\{ \mathcal{C}\mathbf{P}, \mathcal{C}\mathrm{i}\mathbf{P} : \quad \mathbf{P} \in \mathrm{Ker}(\mathcal{L}_{2}) \}.
        \end{equation*}
        There are eight independent eigenfunctions in $\mathrm{Ker}_{\mathrm{X}}(\mathcal{L}_{2})$ in view of the symmetry \eqref{symmetric-squared-eigenfunctions}, hence  $\mathrm{dim}\,\mathrm{Ker}_{\mathrm{X}}(\mathcal{L}_{2}) = 8$.
        
        For the negative eigenvalues, the same argument in Theorem \ref{spectrum-L2-1} can be applied for 
        \begin{equation*}
            \mathrm{gKer}_{\mathrm{X}}(\mathcal{L}_{2}) = \mathrm{span}\{ \mathcal{C}\mathbf{P}, \mathcal{C}\mathrm{i}\mathbf{P} : \quad \mathbf{P} \in \mathrm{gKer}(\mathcal{L}_{2}) \}.
        \end{equation*}
        Hence, the number of negative eigenvalues in $\mathrm{X}$ is two.
    \end{proof}

\begin{figure}[htbp!]
    \centering
    \includegraphics[scale=0.2]{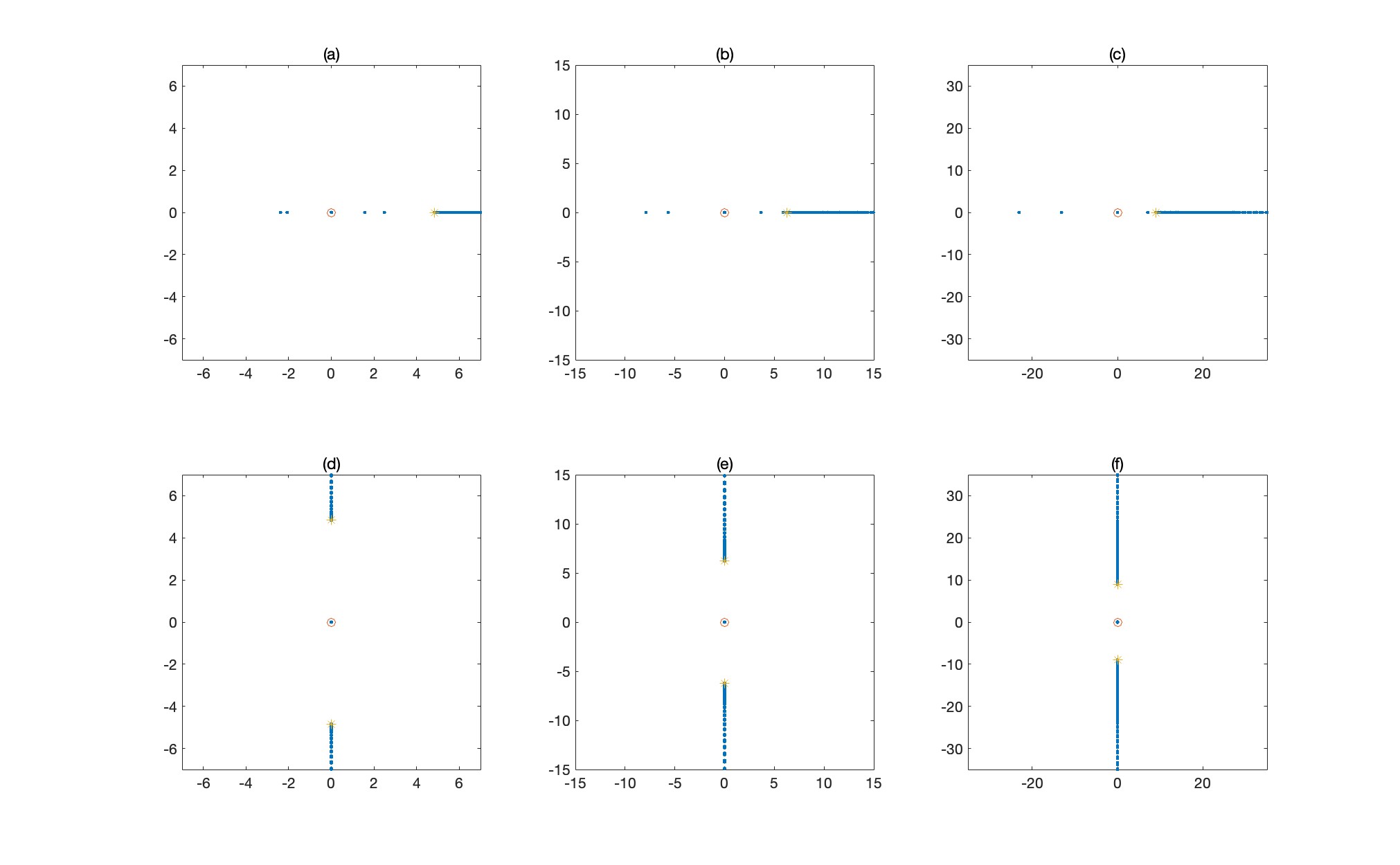}
    \caption{Approximations of the spectrum of $\mathcal{L}_{2}$ (the first row) and $\mathcal{J} \mathcal{L}_{2}$ (the second row) using the Fourier collocation method. Eigenvalues are divided by 64 for convenience. 
    The blue points are the numerical results, and the red circles represent 
the analytical results. The yellow stars are the end points of essential spectrum. The parameters are: $a=1$, $b_{1}=2$, $c_{11}=c_{22}=1$, $c_{12}=c_{21}=0$ with (a,d) $b_{2}=2.2$,  
    (b,e) $b_{2}=2.5$, and (c,f) $b_{2}=3$.}
    \label{computation-nonlinear-spectrum}
\end{figure}

In order to illustrate the result of Theorem \ref{spectrum-L2-1}, we have computed the spectrum of the operators $\mathcal{L}_{2}$ and $\mathcal{J} \mathcal{L}_2$ numerically. We use the Fourier collocation method. Eigenvalues of the operator $\mathcal{L}_{2}$ and $\mathcal{J} \mathcal{L}_{2}$ are shown in Figure \ref{computation-nonlinear-spectrum} in agreement with the theoretical analysis. 

Based on Lemma \ref{spectrum-L2-2}, we derive a coercivity result for the operator $\mathcal{L}_2$ in $\mathrm{X}$. This result is needed for the nonlinear stability analysis of breathers in Section \ref{sec-5}.

Denote the negative eigenvalues and normalized eigenfunctions of \(\mathcal{L}_{2}\) in \(\mathrm{X}\) as \(-\lambda_{-1,2}^{2}\) and \(\eta_{-1,2}\), respectively, i.e.,
\begin{align*}
        &\mathcal{L}_{2}\eta_{-1}=-\lambda_{-1}^{2}\eta_{-1},\\
        &\mathcal{L}_{2}\eta_{-2}=-\lambda_{-2}^{2}\eta_{-2},
    \end{align*}
with $\|\eta_{-1}\|_{L^{2}}=\|\eta_{-2}\|_{L^{2}}=1$. It can be proven by the elliptic estimates that the zero eigenvalue is isolated from the rest of the spectrum of $\mathcal{L}_2$. However, this relies on the exact form of \(\mathcal{L}_{2}\). Instead of elliptic estimates, we use the fact that the zero eigenvalue of $\mathcal{J}\mathcal{L}_{2}$ is isolated from the rest of the spectrum of $\mathcal{J}\mathcal{L}_{2}$ in order to prove  that the zero eigenvalue of $\mathcal{L}_{2}$ is also isolated.

\begin{lem}
	\label{lem-coercivity}
If the perturbation $z \in H^{2}(\mathbb{R},\mathrm{X})$ satisfies
    \begin{equation*}
        z \in \mathrm{Ker}_{\mathrm{X}}(\mathcal{L}_{2})^{\bot} \cap
        \mathrm{span}\{
            \eta_{-1}, 
            \eta_{-2}
        \}^{\bot}, 
    \end{equation*}
    then 
    \begin{equation*}
        (\mathcal{L}_{2} z,z) \geq C_{0}\|z\|_{L^{2}}^{2}
    \end{equation*}
    for some positive constant \( C_{0} \).
\end{lem}

\begin{proof}
 It suffices to show that zero is not the limit of the point spectrum of $\mathcal{L}_{2}$ since the essential 
    spectrum of $\mathcal{L}_{2}$ is bounded away from zero. We argue by contradiction. 
    Suppose there exists a sequence 
    $\{\lambda_{n} : \lambda_{n} > 0\}_{n=1}^{\infty}$ such that $\lambda_{n} \to 0$ as $n \to \infty$. 
    The normalized eigenfunctions $\phi_{n}$ satisfy
    \begin{equation*}
        \mathcal{L}_{2}\phi_{n} = \lambda_{n}\phi_{n},
    \end{equation*}
    with $\|\phi_{n}\|_{L^{2}} = 1$. Since zero is an isolated eigenvalue of $\mathcal{J}\mathcal{L}_{2}$, for sufficiently small $\lambda_{n}$, we have $(\mathcal{J}\mathcal{L}_{2} - \lambda_{n})^{-1}$
    uniformly bounded on $\mathrm{Ker}(\mathcal{L}_{2})^{\bot} = \mathrm{Ker}(\mathcal{J}\mathcal{L}_{2})^{\bot}$. Hence, 
    \begin{equation*}
        \begin{split}
            \|\phi_{n}\| &\leq C \|(\mathcal{J}\mathcal{L}_{2} - \lambda_{n})\phi_{n}\| \\
            &\leq C \lambda_{n} \|\mathcal{J}\phi_{n} - \phi_{n}\| \\ 
            &\leq 4C\lambda_{n} \to 0,
        \end{split}
    \end{equation*}
    which is a contradiction to $\|\phi_{n}\| = 1$. 
\end{proof}

\begin{rem}
    The proof relies on the bounded inverse of $\mathcal{J}$ (which implies that 
    $\mathrm{Ker}(\mathcal{L}_{2}) = \mathrm{Ker}(\mathcal{J}\mathcal{L}_{2})$) and the fact that zero is an isolated eigenvalue for 
    $\mathcal{J}\mathcal{L}_{2}$. For other integrable systems, if $\mathcal{J}$ does not have an inverse with a finite-dimensional kernel, then we also need to consider $\mathrm{Ker}(\mathcal{J})$. 
\end{rem}

\section{Nonlinear stability of breathers}
\label{sec-5}

We give the proof of Theorem \ref{nonlinear-stability} that states the nonlinear stability of breather solutions. The nonlinear manifold of breather solutions is characterized by four parameters of the breather:
\begin{equation*}
    \left\{ \mathbf{q}^{[2]}(x,t;a,b_1,b_2;c_{11},c_{12},c_{21},c_{22}) : \quad  c_{11},c_{12},c_{21},c_{22} \in\mathbb{C}\right\}. 
\end{equation*}
The main approach to proving the nonlinear stability of breathers involves analyzing  the nonlinear manifold and utilizing the conservation laws to constrain perturbations within a space where the operator $\mathcal{L}_{2}$ is coercive. Subsequently, the time-invariant property and continuity of the Lyapunov functional lead to the conclusion of nonlinear stability. 
To simplify the notation, we will use $\mathbf{q}$ instead of $\mathbf{q}^{[2]}$ to denote the breather solutions. 

\subsection{The reduced Hamiltonian}

Let \(\mathrm{n}(A)\), \(\mathrm{z}(A)\), and \(\mathrm{p}(A)\) denote respectively the number of negative, zero, and positive eigenvalues of a linear, 
self-adjoint operator $A$ in a Hilbert space. We proceed to define the conservation laws
\begin{align*}
        \mathcal{Q}_{a}:=\sum_{n=0}^{4}\partial_{a}\mu_{n} H_{n}&=0,  \\
        \mathcal{Q}_{b_{1}}:=\sum_{n=0}^{4}\partial_{b_{1}}\mu_{n} H_{n}&=-\frac{64}{3}b_{1}(b_{1}-2b_{2})(b_{1}+b_{2})^{2},  \\
        \mathcal{Q}_{b_{2}}:=\sum_{n=0}^{4}\partial_{b_{2}}\mu_{n} H_{n}&=\frac{64}{3}b_{2}(2b_{1}-b_{2})(b_{1}+b_{2})^{2},
\end{align*}
where we have used the Lyapunov functional \eqref{Ly-non} and the expressions \eqref{con-law-spectral} for conservation laws. 
The Hessian matrix is given by
\begin{equation}
\label{Hes-mat}
\begin{split}
& \left(\partial_{\sigma\tau}\mathcal{I}_{2} -
    \sum_{n=0}^{4} \partial_{\sigma\tau}\mu_{n} H_{n}\right)_{\sigma,\tau\in \{a,b_{1},b_{2}\}} \\
& \quad     = 64(b_{1} + b_{2}) \begin{pmatrix}
        (b_{1} - b_{2})^{2} & 0 & 0 \\
        0 & b_{1}(b_{2} - b_{1}) & 0 \\
        0 & 0 & b_{2}(b_{1} - b_{2}) \\
    \end{pmatrix}.
    \end{split}
\end{equation}
It is related to the reduced Hamiltonian \(\mathcal{L}_{2}\mathcal{P}\), where \(\mathcal{P}\) is the projection of \(\mathrm{X}\) onto
\begin{equation*}
    \mathrm{X}_{1} = \mathrm{span}\left\{
    \frac{\delta \mathcal{Q}_{a}}{\delta \mathbf{q}}, \; 
\frac{\delta \mathcal{Q}_{b_1}}{\delta \mathbf{q}}, \; \frac{\delta \mathcal{Q}_{b_2}}{\delta \mathbf{q}} \right\}^{\bot}.
\end{equation*}
By utilizing the Hessian matrix \eqref{Hes-mat}, we prove the following lemma.

\begin{lem}\label{spectrum-reduce}
It is true that
\begin{align*}
        &\mathrm{n}(\mathcal{L}_{2}\mathcal{P}) = \mathrm{n}(\mathcal{L}_{2}) - \mathrm{p}\left(\partial_{\sigma\tau}\mathcal{I}_{2} -
    \sum_{n=0}^{4} \partial_{\sigma\tau}\mu_{n} H_{n}\right), \\
    &\mathrm{z}(\mathcal{L}_{2}\mathcal{P}) = \mathrm{z}(\mathcal{L}_{2}).
    \end{align*}
\end{lem}

\begin{proof}
    Differentiating $\frac{\delta\mathcal{I}_{2}}{\delta\mathbf{q}}=\sum_{n=0}^{4}\mu_{n} \frac{\delta H_{n}}{\delta\mathbf{q}}=0$ with respect to $\sigma$ for $\sigma \in \{a,b_1,b_2\}$, we obtain
\begin{equation}\label{con-L-X}
    \mathcal{L}_{2}\partial_{\sigma} \mathbf{q}
    +\sum_{n=0}^{4}\partial_{\sigma}\mu_{n} \frac{\delta H_{n}}{\delta\mathbf{q}}
    =\mathcal{L}_{2}\partial_{\sigma} \mathbf{q}
    +\frac{\delta \mathcal{Q}_{\sigma}}{\delta \mathbf{q}}=0. 
\end{equation}
Differentiating $\mathcal{I}_{2}=\sum_{n=0}^{4}\mu_{n} H_{n}$ with respect to $\sigma$ for $\sigma \in \{a,b_1,b_2\}$, we obtain
\begin{equation*}
    \partial_{\sigma} \mathcal{I}_{2}=\sum_{n=0}^{4}\partial_{\sigma}\mu_{n} H_{n}+
    \left(
        \sum_{n=0}^{4}\mu_{n}\frac{\delta H_{n}}{\delta\mathbf{q}},\partial_{\sigma}\mathbf{q}
    \right)
    =\sum_{n=0}^{4}\partial_{\sigma}\mu_{n} H_{n}
\end{equation*}
due to the same equation $\frac{\delta\mathcal{I}_{2}}{\delta\mathbf{q}}=\sum_{n=0}^{4}\mu_{n} \frac{\delta H_{n}}{\delta\mathbf{q}}=0$.
Another differentiation with respect to $\tau$ for $\tau \in \{a,b_1,b_2\}$ yields
\begin{align*}
        \partial_{\tau}\partial_{\sigma} \mathcal{I}_{2}&=\sum_{n=0}^{4}\partial_{\tau}\partial_{\sigma}\mu_{n} H_{n}+
    \left(\sum_{n=0}^{4}\partial_{\sigma}\mu_{n} 
    \frac{\delta H_{n}}{\delta\mathbf{q}},\partial_{\tau}\mathbf{q}\right)\\
        &=\sum_{n=0}^{4}\partial_{\tau}\partial_{\sigma}\mu_{n} H_{n}-
        \left(\mathcal{L}_{2}\partial_{\sigma} \mathbf{q},\partial_{\tau}\mathbf{q}\right).
    \end{align*}
Comparison with \eqref{Hes-mat} shows that the Hessian matrix 
is related to the negative eigenvalues of \(\mathcal{L}_{2}\) by 
\begin{equation*}
\partial_{\sigma\tau}\mathcal{I}_{2} -
\sum_{n=0}^{4} \partial_{\sigma\tau}\mu_{n} H_{n} = -\left(\mathcal{L}_{2}\partial_{\sigma} \mathbf{q},\partial_{\tau}\mathbf{q}\right).
\end{equation*}
The remaining proof is standard but we give it for the sake of completeness.
Based on the identity \eqref{con-L-X}, we define
\begin{equation*}
    \mathrm{Y} = \left\{
        \partial_{a} \mathbf{q}, \; \partial_{b_1} \mathbf{q}, \; \partial_{b_2} \mathbf{q}
    \right\}.
\end{equation*}
Since the Hessian matrix given by \eqref{Hes-mat} has no zero eigenvalues, we have
\begin{equation}\label{ne-ze-1}
    \begin{split}
        &\mathrm{n}\left(\left.\mathcal{L}_{2}\right|_{\mathrm{Y}}\right) = \mathrm{p}\left(\partial_{\sigma\tau}\mathcal{I}_{2} -
        \sum_{n=0}^{4} \partial_{\sigma\tau}\mu_{n} H_{n}\right), \\
        &\mathrm{z}\left(\left.\mathcal{L}_{2}\right|_{\mathrm{Y}}\right) = 0,
    \end{split}
\end{equation}
and \(\mathrm{Y} \cap \mathrm{X}_{1} = \{0\}\). Hence, we have the direct sum decomposition 
\(\mathrm{X}_{1} \oplus \mathrm{Y}\). 
Since \(\mathcal{L}_{2}\) is a one-to-one map from \(\mathrm{Y}\) to \(\mathrm{X}_{1}^{\bot}\), 
according to \eqref{con-L-X}, \(\mathrm{Y}\) is isomorphic to \(\mathrm{X}_{1}^{\bot}\) so that $\mathrm{X}_{1}\oplus\mathrm{Y}=\mathrm{X}$.
In addition, by \eqref{con-L-X} and the definition of $\mathrm{X}_{1}$, 
for any $u_{1}\in\mathrm{X}_{1}$, $v_{1}\in\mathrm{Y}$,  we obtain 
\begin{equation*}
    (\mathcal{L}_{2} v_{1}, u_{1}) = - \left(\sum_{n=0}^{4}\partial_{\sigma}\mu_{n} \frac{\delta H_{n}}{\delta\mathbf{q}}, u_{1} \right) = 0
\end{equation*}
for some $\sigma$. Hence the sum is also direct under the product $(\mathcal{L}_{2}\cdot,\cdot)$. 
In view of \eqref{ne-ze-1}, the number of negative and zero eigenvalues 
for $\mathcal{L}_{2}$ in space $\mathrm{X}_{1}\oplus\mathrm{Y}$ is given by
\begin{equation*}
        \mathrm{n}(\mathcal{L}_{2})=
    \mathrm{n}(\left.\mathcal{L}_{2}\right|_{\mathrm{X}_{1}\oplus\mathrm{Y}})=
    \mathrm{n}(\mathcal{L}_{2}\mathcal{P})
    +\mathrm{p}\left(\partial_{\sigma\tau}\mathcal{I}_{2} -
    \sum_{n=0}^{4} \partial_{\sigma\tau}\mu_{n} H_{n}\right), 
\end{equation*}
and
\begin{equation*}
    \mathrm{z}(\mathcal{L}_{2})=
    \mathrm{z}(\left.\mathcal{L}_{2}\right|_{\mathrm{X}_{1}\oplus\mathrm{Y}})
    =
    \mathrm{z}(\mathcal{L}_{2}\mathcal{P}),
\end{equation*}
which completes the proof.
\end{proof}

Since $n(\mathcal{L}_2) = 2$ by Lemma \ref{spectrum-L2-2}, 
it follows from (\ref{Hes-mat}) and Lemma \ref{spectrum-reduce} that 
$n(\mathcal{L}_2 \mathcal{P}) = 0$. Hence \(\mathcal{L}_{2}\) 
is coercive in the space
\begin{equation*}
    \mathcal{R}'(\mathbf{q}) := \mathrm{Ker}_{\mathrm{X}}(\mathcal{L}_{2})^{\bot} \cap
    \mathrm{span}\left\{
        \frac{\delta \mathcal{Q}_{a}}{\delta \mathbf{q}}, \; 
        \frac{\delta \mathcal{Q}_{b_1}}{\delta \mathbf{q}}, \; 
        \frac{\delta \mathcal{Q}_{b_2}}{\delta \mathbf{q}}
    \right\}^{\bot}
\end{equation*}
according to the following lemma. Compared to Lemma \ref{lem-coercivity}, we obtain the coercivity for \(\mathcal{L}_{2}\) in terms of the spectral parameters of the breathers.

\begin{lem}
	\label{lem-coercivity-parameters}
If the perturbation \(\mathbf{z} \in H^{2}(\mathbb{R},\mathrm{X})\) satisfies 
$\mathbf{z} \in \mathcal{R}'(\mathbf{q})$, then 
    \begin{equation*}
        (\mathcal{L}_{2} \mathbf{z}, \mathbf{z}) \geq C_{1} \|\mathbf{z}\|_{H^{2}}^{2},
    \end{equation*}
    where \(C_{1}\) is a positive constant.
\end{lem}

\begin{proof}
    We have previously established the \(L^{2}\) coercivity of \(\mathcal{L}_{2}\) in Lemma \ref{spectrum-reduce}. To prove the \(H^{2}\) coercivity, we consider the ODE \(\frac{\delta \mathcal{I}_{2}}{\delta \mathbf{q}} = 0\) satisfied by the breather solutions.
    
    We proceed by contradiction. Suppose there exists a bounded sequence \(\{z_{n}\} \subset H^{2}(\mathbb{R})\) such that \(\|z_{n}\|_{H^{2}} = 1\) and \((\mathcal{L}_{2} z_{n}, z_{n}) \to 0\) as \(n \to \infty\). 
    The \(L^{2}\)-coercivity induce that 
    \begin{equation*}
        \lim_{n\to \infty}\|z_{n}\|^{2}_{L^{2}}\to 0.
    \end{equation*}
    We also have \(\lim_{n\to \infty}\|\partial_{x}z_{n}\|_{L^{2}}=0\) since 
    \begin{equation*}
        \|\partial_{x}z_{n}\|_{L^{2}}=(\partial_{x}z_{n},\partial_{x}z_{n})=
        -(\partial_{x}^{2}z_{n},z_{n})\leq \|\partial_{x}^{2}z_{n}\|_{L^{2}}^{1/2}\|z_{n}\|_{L^{2}}^{1/2}\to 0,\quad n\to\infty.  
    \end{equation*}
    Hence 
    \begin{equation*}
        \lim_{n\to \infty}\|\partial_{x}^{2}z_{n}\|^{2}_{L^{2}}\to 1. 
    \end{equation*}
    By the definition of \(\mathcal{L}_{2}\), we obtain
    \begin{align*}
        \|\partial_{x}^{2} z_{n}\|_{L^{2}} &+ \sum_{i=0}^{2}(\partial_{x}^{i}z_{n}, 
        f_{i}(\mathbf{q}, \mathbf{q}^{*},\mathbf{q}_{x}, \mathbf{q}_{x}^{*},\mathbf{q}_{xx}, \mathbf{q}_{xx}^{*}) \partial_{x} z_{n}) \\
        &+(z_{n}, g_{0}(\mathbf{q}, \mathbf{q}^{*},\mathbf{q}_{x}, \mathbf{q}_{x}^{*},\mathbf{q}_{xx}, \mathbf{q}_{xx}^{*}) z_{n})\to 0,
    \end{align*}
    for some polynomial functions \(f_{0},f_{1},f_{2},g_{0}\). Since
    \begin{equation*}
        \left|(\partial_{x}^{i}z_{n}, f \partial_{x}^{j} z_{n})\right|\leq 
        \|f\|_{L^{\infty}}\|\partial_{x}^{i}z_{n}\|_{L^{2}}^{1/2}
        \|\partial_{x}^{j} z_{n}\|_{L^{2}}^{1/2}
        \to 0,
    \end{equation*}
    for $i=0,1,2$, $j=0,1$ and polynomial $f=f(\mathbf{q}, \mathbf{q}^{*},\mathbf{q}_{x}, \mathbf{q}_{x}^{*},\mathbf{q}_{xx}, \mathbf{q}_{xx}^{*})$,  we obtain 
    \begin{equation*}
        \sum_{i=0}^{2}(\partial_{x}^{i}z_{n}, f_{i}(\mathbf{q}, \mathbf{q}^{*},\mathbf{q}_{x}, \mathbf{q}_{x}^{*},\mathbf{q}_{xx}, \mathbf{q}_{xx}^{*}) \partial_{x} z_{n})
        +(z_{n}, g_{0}(\mathbf{q}, \mathbf{q}^{*},\mathbf{q}_{x}, \mathbf{q}_{x}^{*},\mathbf{q}_{xx}, \mathbf{q}_{xx}^{*}) z_{n})\to 0. 
    \end{equation*}
    This leads to a contradiction
    \begin{equation*}
        \|\partial_{x}^{2} z_{n}\| \to 0.  
    \end{equation*}
    Therefore, the \(H^{2}\)-coercivity holds.
\end{proof}

Based on Lemma \ref{lem-coercivity-parameters},  we define the space 
\begin{equation*}
    \mathcal{R}(\mathbf{q}) := \left\{
        \mathbf{v}\in H^{2}(\mathbb{R},\mathrm{X}) : \quad \mathbf{v}\in \mathrm{Ker}_{\mathrm{X}}(\mathcal{L}_{2})^{\bot}, \;\;
        \mathcal{Q}_{\sigma}(\mathbf{q})=\mathcal{Q}_{\sigma}(\mathbf{q}+\mathbf{v}), \;\; \sigma \in \{a, b_1, b_2\}
    \right\}
\end{equation*}
and obtain the following lemma.

\begin{lem}
    There exists a constant \( C_{2} > 0 \) such that for sufficiently small \( \mathbf{z} \in \mathcal{R}(\mathbf{q}) \) in the \( H^{2} \) norm, it is true that 
    \begin{equation*}
        (\mathcal{L}_{2} \mathbf{z}, \mathbf{z}) \geq C_{1}
        \|\mathbf{z}\|_{H^{2}}^{2} - C_{2} \|\mathbf{z}\|_{H^{2}}^{3}.
    \end{equation*}
\end{lem}

\begin{proof}
    For any $\mathbf{z} \in \mathcal{R}(\mathbf{q})$, we decompose it as
    \begin{equation*}
        \mathbf{z} = \mathbf{z}_{1} + 
        \sum_{\sigma \in \{a, b_{1}, b_{2}\}} \alpha_{\sigma} \frac{\delta \mathcal{Q}_{\sigma}}{\delta \mathbf{q}} + 
        \sum_{i, j} \beta_{ij} \partial_{c_{ij}} \mathbf{q}, 
    \end{equation*}
    where $\mathbf{z}_{1} \in \mathcal{R}'(\mathbf{q})$. 
    Since $\mathcal{Q}_{\sigma}(\mathbf{q} + \mathbf{z}) = \mathcal{Q}_{\sigma}(\mathbf{q})$, 
    expanding $\mathcal{Q}_{\sigma}$ around $\mathbf{q}$ with perturbation $\mathbf{z}$ yields
    \begin{equation}\label{linear-equation-1}
        \sum_{\sigma \in \{a, b_{1}, b_{2}\}} \alpha_{\sigma} 
        \left( \frac{\delta \mathcal{Q}_{\sigma}}{\delta \mathbf{q}}, \frac{\delta \mathcal{Q}_{\tau}}{\delta \mathbf{q}} \right) + 
        \sum_{i, j} \beta_{ij} \left( \partial_{c_{ij}} \mathbf{q}, \frac{\delta \mathcal{Q}_{\tau}}{\delta \mathbf{q}} \right) 
        = \mathcal{O}(\|\mathbf{z}\|_{H^{2}(\mathbb{R})}^{2}), \quad \tau \in \{a, b_{1}, b_{2}\}.
    \end{equation}
    Moreover, since $\mathbf{z}, \mathbf{z}_{1} \in \mathrm{Ker}(\mathcal{L}_{2})^{\bot}$, the identity
    \begin{equation}\label{linear-equation-2}
        \sum_{\sigma \in \{a, b_{1}, b_{2}\}} \alpha_{\sigma} 
        \left( \frac{\delta \mathcal{Q}_{\sigma}}{\delta \mathbf{q}}, \partial_{c_{kl}} \mathbf{q} \right) + 
        \sum_{i, j} \beta_{ij} \left( \partial_{c_{ij}} \mathbf{q}, \partial_{c_{kl}} \mathbf{q} \right) 
        = 0 , \quad 1 \leq k, l \leq 2
    \end{equation}
    holds. Solving equations \eqref{linear-equation-1} and \eqref{linear-equation-2}, we observe that 
    the coefficient matrix is the Gram matrix in 
    $\mathrm{Ker}(\mathcal{L}_{2}) \cup \mathrm{span} \left\{ \frac{\delta \mathcal{Q}_{\sigma}}{\delta \mathbf{q}} \right\}$. 
    Thus, the order of coefficients is given by
    \begin{equation}\label{ord-coeff}
        \alpha_{\sigma} = \mathcal{O}(\|\mathbf{z}\|_{H^{2}}^{2}), \quad 
        \beta_{ij} = \mathcal{O}(\|\mathbf{z}\|_{H^{2}}^{2}),
    \end{equation}
so that we obtain     for some constants $C_1, C_{2}, C_{3} > 0$ 
    \begin{align*}
            (\mathcal{L}_{2} \mathbf{z}, \mathbf{z}) & = (\mathcal{L}_{2} \mathbf{z}_{1}, \mathbf{z}_{1}) + 
            2 \sum_{\sigma} \alpha_{\sigma} \left( \mathbf{z}_{1}, \mathcal{L}_{2} \frac{\delta \mathcal{Q}_{\sigma}}{\delta \mathbf{q}} \right) + 
            \sum_{\sigma, \tau} \alpha_{\sigma} \alpha_{\tau} 
            \left( \mathcal{L}_{2} \frac{\delta \mathcal{Q}_{\sigma}}{\delta \mathbf{q}}, 
            \frac{\delta \mathcal{Q}_{\tau}}{\delta \mathbf{q}} \right) \\
            & \geq C_{1} \|\mathbf{z}_{1}\|_{H^{2}}^{2} - C_{3} \|\mathbf{z}\|_{H^{2}}^{2} \|\mathbf{z}_{1}\|_{H^{2}} \\
            & \geq C_{1}
            \left\|\mathbf{z} - \sum_{\sigma \in \{a, b_{1}, b_{2}\}} \alpha_{\sigma} \frac{\delta \mathcal{Q}_{\sigma}}{\delta \mathbf{q}} + 
            \sum_{i, j} \beta_{ij} \partial_{c_{ij}} \mathbf{q} \right\|_{H^{2}}^{2} - C_{3} \|\mathbf{z}\|_{H^{2}}^{3} \\
            & \geq C_{1} \|\mathbf{z}\|_{H^{2}}^{2} - C_{2} \|\mathbf{z}\|_{H^{2}}^{3}
        \end{align*}
by the Cauchy-Schwartz inequality and the estimates \eqref{ord-coeff}. 
\end{proof}

\subsection{Lyapunov functional}

We identify the perturbation to the nonlinear manifold of breather solutions \eqref{breather} by using the orthogonality conditions to the kernel of $\mathcal{L}_{2}$. The four complex scattering parameters of breather solutions \eqref{breather} generate a manifold of real dimension eight. This manifold corresponds to the kernel $\mathrm{Ker}_{\mathrm{X}}(\mathcal{L}_{2})$, according to the following lemma. Compared to Theorem \ref{spectrum-L2-1}, where the eight-dimensional kernel of $\mathcal{L}_{2}$ in $L^2(\mathbb{R},\mathbb{C}^4)$ was obtained by using squared eigenfunctions, 
the kernel of $\mathcal{L}_2$ in $\mathrm{X}$ is characterized in terms of the scattering parameters.
 
\begin{lem}
	\label{lem-kernel}
    The kernel of $\mathcal{L}_{2}$ in the real Hilbert space $\mathrm{X}$ is spanned by the derivatives of $\mathbf{q}$ with respect to the four scattering parameters 
    $c_{ij}=|c_{ij}|{\rm e}^{{\rm i}\theta_{ij}}$: 
    \begin{equation*}
        \mathrm{Ker}_{\mathrm{X}}(\mathcal{L}_{2}) = \mathrm{span} \left\{ 
            \partial_{|c_{ij}|} \mathbf{q}^{[2]}, \partial_{\theta_{ij}} \mathbf{q}^{[2]}, \quad i,j = 1,2 
        \right\}.
    \end{equation*}
\end{lem}

\begin{proof}
    We have proven that $\mathrm{dim}\ \mathrm{Ker}_{\mathrm{X}}(\mathcal{L}_{2}) = 8$ by Lemma \ref{spectrum-L2-2}. It remains to show that 
    $\partial_{|c_{ij}|} \mathbf{q}^{[2]}, \partial_{\theta_{ij}} \mathbf{q}^{[2]} \in \mathrm{Ker}_{\mathrm{X}}(\mathcal{L}_{2})$ and that the eight eigenfunctions are linearly independent. Since the operator $\mathcal{L}_{2}$ is 
    independent of $c_{ij}$, we have the expansion
    \begin{equation*}
        \nabla \mathcal{I}_{2} (\mathbf{q}^{[2]}(c_{ij} + \epsilon)) = \nabla \mathcal{I}_{2} (\mathbf{q}^{[2]}(c_{ij})) +
        \epsilon \mathcal{L}_{2} (\partial_{c_{ij}} \mathbf{q}^{[2]}) + \mathcal{O}(\epsilon^{2}). 
    \end{equation*}
    Since the perturbation on $c_{ij}$ does not affect the Hamiltonian, the ODE
    \begin{equation*}
        \nabla \mathcal{I}_{2} (\mathbf{q}^{[2]}(c_{ij} + \epsilon)) = \nabla \mathcal{I}_{2} (\mathbf{q}^{[2]}(c_{ij})) = 0
    \end{equation*}
    holds. The $\mathcal{O}(\epsilon)$ term represents that
    \begin{equation*}
        \mathcal{L}_{2} (\partial_{c_{ij}} \mathbf{q}^{[2]}) = 0.
    \end{equation*}
    It can be verified that $\partial_{|c_{ij}|} \mathbf{q}^{[2]}, \partial_{\theta_{ij}} \mathbf{q}^{[2]}$ 
    are linearly independent in view of the exponential term. 
    The detailed verification is omitted here.
\end{proof}

The following lemma follows from Lemma \ref{lem-kernel} by standard estimates.

\begin{lem}\label{mod}
    There exist $\delta_{0}, \epsilon_{0} > 0$ such that 
    for all $0 < \delta < \delta_{0}$, there exist functions 
    $$
    (c_{11}(t), c_{21}(t), c_{12}(t), c_{22}(t)) \in \mathbb{C}^{4}
    $$ 
    such that
    \begin{equation*}
        \mathbf{w}(\cdot,t) = \mathbf{u}(\cdot,t) - \mathbf{q}^{[2]}(\cdot,t;a,b_1,b_2,c_{11}(t), c_{21}(t), c_{12}(t), c_{22}(t)) \in \mathrm{Ker}_{\mathrm{X}}(\mathcal{L}_{2})^{\bot}
    \end{equation*}
    for $\|\mathbf{w}(\cdot,t)\|_{L^{2}} \leq \delta$.
    Moreover, if $\mathbf{u} \in C(\mathbb{R},H^2(\mathbb{R},\mathbb{C}^2))$ is the solution of the CNLS equations \eqref{CNLS}, then
    \begin{equation}\label{estimate-derivative}
        \sum_{1 \leq i,j \leq 2} |\partial_{t} c_{ij}(t)| \leq C \|\mathbf{w}(\cdot,t)\|_{H^{2}}
    \end{equation}
    for some constant $C$. 
\end{lem}

\begin{proof}
    The proof relies on the Implicit Function Theorem. For any $t \in \mathbb{R}$, 
    consider the equations
    \begin{align*}
            g_{11} &= (\mathbf{w}, \partial_{c_{11}} \mathbf{q}^{[2]}), \\ g_{12} &= (\mathbf{w}, \partial_{c_{12}} \mathbf{q}^{[2]}), \\
            g_{21} &= (\mathbf{w}, \partial_{c_{21}} \mathbf{q}^{[2]}), \\ g_{22} &= (\mathbf{w}, \partial_{c_{22}} \mathbf{q}^{[2]}).  
        \end{align*}
    Then 
    \begin{equation*}
        g_{ij}|_{\mathbf{u}(t)=\mathbf{q}^{[2]}(t)} = 0,
    \end{equation*}
    and 
    \begin{equation*}
        \left.\partial_{c_{kl}} g_{ij}\right|_{\mathbf{u}(t) = \mathbf{q}^{[2]}(t)} =
        -(\partial_{c_{kl}} \mathbf{q}^{[2]}, \partial_{c_{ij}} \mathbf{q}^{[2]}), 
    \end{equation*}
    which is the Gram matrix in $\mathrm{Ker}_{\mathrm{X}}(\mathcal{L}_{2})$. This matrix is
    non-degenerate (its determinant is non-zero) since $\partial_{c_{ij}} \mathbf{q}^{[2]}$
    are linearly independent. 

    For the equation \eqref{estimate-derivative}, differentiating 
    $(\mathbf{w}, \partial_{c_{kl}} \mathbf{q}^{[2]}) = 0$ with respect to $t$, invoking CNLS equations
    \eqref{CNLS}, leads to
    \begin{equation*}
        (\partial_{t} \mathbf{u} - \partial_{t} \mathbf{q}^{[2]} - \sum_{i,j} \partial_{c_{ij}} \mathbf{q}^{[2]}
        \partial_{t} c_{ij}, \partial_{c_{kl}} \mathbf{q}^{[2]}) + 
        (\mathbf{w}, \partial_{t} \partial_{c_{kl}} \mathbf{q}^{[2]}) = 0. 
    \end{equation*}
    Since 
    \begin{align*}
            \left|(\partial_{t} \mathbf{u} - \partial_{t} \mathbf{q}^{[2]},
            \partial_{c_{kl}} \mathbf{q}^{[2]})\right| & = \left| \left( -{\rm i}
                \frac{1}{2} \mathbf{w}_{xx} + |\mathbf{u}|^{2} \mathbf{u}
                - |\mathbf{q}^{[2]}|^{2} \mathbf{q}^{[2]}, \partial_{c_{kl}} \mathbf{q}^{[2]}
            \right) \right| \\
            & \leq C' \|\partial_{c_{kl}} \mathbf{q}^{[2]}\|_{L^{2}}
                \|\mathbf{w}\|_{H^{2}},
        \end{align*}
    the derivatives of the coefficients $c_{ij}$ satisfy
    \begin{equation}\label{der-cij}
        \sum_{i,j} \partial_{t} c_{ij} (\partial_{c_{ij}} \mathbf{q}^{[2]},
        \partial_{c_{kl}} \mathbf{q}^{[2]}) =
        (\mathbf{w}, \partial_{t} \partial_{c_{kl}} \mathbf{q}^{[2]}) +
        (\partial_{t} \mathbf{u} - \partial_{t} \mathbf{q}^{[2]}, \partial_{c_{kl}} \mathbf{q}^{[2]}) =
        \mathcal{O}(\|\mathbf{w}\|_{H^{2}}). 
    \end{equation}
    Since the matrix $(\partial_{c_{ij}} \mathbf{q}^{[2]}, \partial_{c_{kl}} \mathbf{q}^{[2]})$ is
    non-degenerate and $(c_{11}(t), c_{12}(t), c_{21}(t), c_{22}(t))$ is close to the parameters $(c_{11}, c_{12}, c_{21}, c_{22})$ of the breather solution $\mathbf{q}^{[2]}$ in
    the Implicit Function Theorem, the coefficient matrix in \eqref{der-cij} remains non-degenerate by the continuity
    of the inner product and the determinant. 
\end{proof}

Before proving Theorem \ref{nonlinear-stability}, it is necessary to establish the  continuity of conservation quantities, according to the following lemma.

\begin{lem}
	\label{lem-conserved}
    The conservation quantities $H_{0}(\mathbf{u})$, $H_1(\mathbf{u})$, $H_2(\mathbf{u})$, $H_3(\mathbf{u})$, $H_4(\mathbf{u})$ are continuous in the space $H^{2}(\mathbb{R},\mathbb{C}^2)$. 
\end{lem}

\begin{proof}
    The proof relies on continuous embedding. Let $\mathbf{u}, \mathbf{v} \in H^{2}$ with $\|\mathbf{u}\|_{H^{2}}, 
    \|\mathbf{v}\|_{H^{2}} \leq M$ for some constant $M$. For $H_0(\mathbf{u})$, we have
    \begin{align*}
            |H_{0}(\mathbf{u})-H_{0}(\mathbf{v})| &= 
            \left| \|\mathbf{u}\|_{L^{2}}^{2} - \|\mathbf{v}\|_{L^{2}}^{2} \right| \\
            &= \left|(\|\mathbf{u}\|_{L^{2}} - \|\mathbf{v}\|_{L^{2}})(\|\mathbf{u}\|_{L^{2}} + \|\mathbf{v}\|_{L^{2}})\right| \\
            &\leq 2M \|\mathbf{u}-\mathbf{v}\|_{L^{2}}.
    \end{align*}
    For $H_{1}(\mathbf{u})$, we have
    \begin{align*}
            |H_{1}(\mathbf{u})-H_{1}(\mathbf{v})| &= 
            \left| \int_{\mathbb{R}}
            \mathrm{i}(\mathbf{u}_{x}^{T}\mathbf{u}^{*} - \mathbf{v}_{x}^{T}\mathbf{v}^{*}) \, \mathrm{d}x \right| \\
            &\leq \left| \int_{\mathbb{R}}  
            (\mathbf{u}_{x}^{T}(\mathbf{u}^{*} - \mathbf{v}^{*}) + (\mathbf{u}_{x} - \mathbf{v}_{x})^{T}\mathbf{v}^{*}) \, \mathrm{d}x \right| \\
            &\leq \|\mathbf{u}_{x}\|_{L^{2}} \|\mathbf{u}-\mathbf{v}\|_{L^{2}} + \|\mathbf{v}\|_{L^{2}} \|\mathbf{u}_{x} - \mathbf{v}_{x}\|_{L^{2}} \\
            &\leq 2M \|\mathbf{u}-\mathbf{v}\|_{H^{1}}.
        \end{align*}
    For $H_{2}(\mathbf{u})$, we need to use the embedding $H^{1} \hookrightarrow L^{\infty}$, then
    \begin{align*}
            |H_{2}(\mathbf{u})-H_{2}(\mathbf{v})| &\leq 
            \left| \|\mathbf{u}_{x}\|_{L^{2}}^{2} - \|\mathbf{v}_{x}\|_{L^{2}}^{2} \right| +
            \left| \|\mathbf{u}\|_{L^{4}}^{4} - \|\mathbf{v}\|_{L^{4}}^{4} \right| \\
            &\leq 2M \|\mathbf{u}_{x}-\mathbf{v}_{x}\|_{L^{2}} + 4M^{3} \|\mathbf{u}-\mathbf{v}\|_{L^{4}} \\
            &\leq 2M \|\mathbf{u}_{x}-\mathbf{v}_{x}\|_{L^{2}} + 4M^{3} \|\mathbf{u}-\mathbf{v}\|_{L^{\infty}}^{1/2} \|\mathbf{u}-\mathbf{v}\|_{L^{2}}^{1/2} \\
            &\leq 2M \|\mathbf{u}-\mathbf{v}\|_{H^{1}} + C \|\mathbf{u}-\mathbf{v}\|_{L^{2}}^{1/2}. 
        \end{align*}
    For $H_{3}(\mathbf{u})$, we need to use the $H^{2}$ norm
    \begin{align*}
            |H_{3}(\mathbf{u})-H_{3}(\mathbf{v})| &\leq
            |H_{1}(\mathbf{u}_{x}) - H_{1}(\mathbf{v}_{x})| +
            \left| \int_{\mathbb{R}}
            (|\mathbf{u}|^{2}\mathbf{u}^{T}\mathbf{u}_{x}^{*} - |\mathbf{v}|^{2}\mathbf{v}^{T}\mathbf{v}_{x}^{*}) \, \mathrm{d}x \right| \\
            &\leq 2M \|\mathbf{u}_{x} - \mathbf{v}_{x}\|_{H^{1}} + \left| \int_{\mathbb{R}}
            (|\mathbf{u}|^{2} - |\mathbf{v}|^{2}) \mathbf{u}^{T}\mathbf{u}_{x}^{*} + |\mathbf{v}|^{2} (\mathbf{u}^{T}\mathbf{u}_{x}^{*} - \mathbf{v}^{T}\mathbf{v}_{x}^{*}) \, \mathrm{d}x \right| \\
            &\leq 2M \|\mathbf{u} - \mathbf{v}\|_{H^{2}} + 
            \||\mathbf{u}|^{2} - |\mathbf{v}|^{2}\|_{L^{2}} \|\mathbf{u}^{T}\mathbf{u}_{x}^{*}\|_{L^{2}} +
            \|\mathbf{v}\|_{L^{4}}^{2} \|\mathbf{u}^{T}\mathbf{u}_{x}^{*} - \mathbf{v}^{T}\mathbf{v}_{x}^{*}\|_{L^{2}} \\
            &\leq 2M \|\mathbf{u} - \mathbf{v}\|_{H^{2}} + C \|\mathbf{u} - \mathbf{v}\|_{H^{1}}^{1/2}
        \end{align*}
    since $\|\mathbf{u}\|_{L^{4}} \leq \|\mathbf{u}\|_{L^{\infty}}^{1/2} \|\mathbf{u}\|_{L^{2}}^{1/2}$. The continuity of 
    $H_{4}(\mathbf{u})$ is similar to the above estimates. Specifically, we have
    \begin{align*}
        \left|\int_{\mathbb{R}} |\mathbf{u}|^{6} - |\mathbf{v}|^{6} \, \mathrm{d}x\right| &=
            \left| \|\mathbf{u}\|_{L^{6}}^{6} - \|\mathbf{v}\|_{L^{6}}^{6} \right| \\
            &\leq 6M^{5} \|\mathbf{u} - \mathbf{v}\|_{L^{6}} \\
            &\leq 6M^{5} \|\mathbf{u} - \mathbf{v}\|_{L^{\infty}}^{2/3} \|\mathbf{u} - \mathbf{v}\|_{L^{2}}^{1/3} \\
            &\leq C \|\mathbf{u} - \mathbf{v}\|_{L^{2}}^{1/3}. 
        \end{align*}
       Hence, all conserved quantities are continuous in $H^2$.
\end{proof}

Now we can prove the nonlinear stability for breather solutions. 

\begin{proof}[Proof of Theorem \ref{nonlinear-stability}]
    The well-posedness of the CNLS equations \eqref{CNLS} has been established, 
    as discussed in \cite{cazenave_introduction_1989}. Hence, for initial data \(\mathbf{u}(\cdot,0){\in H^2(\mathbb{R},\mathbb{C}^2)}\), 
    there exist global solutions \(\mathbf{u}(\cdot,t){\in H^2(\mathbb{R},\mathbb{C}^2)}\) for any time $t$. 
    
    We argue by contradiction. Assume that there exists $\epsilon_{0} > 0$ such that there are sequences $\mathbf{u}_{n}$ and $t_{n}$ for which
    \begin{equation*}
        \|\mathbf{u}_{n}(\cdot,0) - \mathbf{q}^{[2]}(\cdot,0;a,b_1,b_2;c_{11}(0),c_{12}(0),c_{21}(0),c_{22}(0))\|_{H^{2}} \to 0
        \end{equation*}
        and
        \begin{equation*}
        \|\mathbf{u}_{n}(\cdot,t_{n}) - \mathbf{q}^{[2]}(\cdot,t_{n};a,b_1,b_2;c_{11}(t_{n}),c_{12}(t_{n}),c_{21}(t_{n}),c_{22}(t_{n}))\|_{H^{2}} = \epsilon_{0} 
    \end{equation*}
    for any $C^{1}$ functions $c_{11}(t)$, $c_{12}(t)$, $c_{21}(t)$, $c_{22}(t)$.
    Since the conservation quantities are continuous by Lemma \ref{lem-conserved}, it follows that
    \begin{equation*}
        \mathcal{I}_{2}(\mathbf{u}_{n}(x,t_{n})) = \mathcal{I}_{2}(\mathbf{u}_{n}(x,0)) \to \mathcal{I}_{2}(\mathbf{q}^{[2]}). 
    \end{equation*}
    The continuity of the conservation laws $\mathcal{Q}_{\sigma}$ implies the existence of a sequence $\mathbf{v}_{n}$ such that 
    \begin{equation*}
        \mathcal{Q}_{\sigma}(\mathbf{v}_{n}) = \mathcal{Q}_{\sigma}(\mathbf{q}^{[2]})
    \end{equation*}
    and 
    \begin{equation*}
        \|\mathbf{v}_{n} - \mathbf{u}_{n}(\cdot,t_{n})\|_{H^{2}} \to 0 \quad \text{as} \quad n \to \infty. 
    \end{equation*}
    Considering the functions $c_{11}(t)$, $c_{12}(t)$, $c_{21}(t)$, $c_{22}(t)$ defined in Lemma \ref{mod}, we get
    \begin{equation*}
        \mathbf{v}_{n}(x) - \mathbf{q}^{[2]}(x,t_{n};a,b_1,b_2;c_{11}(t_{n}),c_{12}(t_{n}),c_{21}(t_{n}),c_{22}(t_{n})) \in \mathrm{Ker}_{\mathrm{X}}(\mathcal{L}_{2})^{\bot}. 
    \end{equation*}
Since the conservation laws are independent of the complex scattering parameters, we find that
    \begin{equation*}
        \mathbf{z}_{n} = \mathbf{v}_{n}(x) - \mathbf{q}^{[2]}(x,t_{n};a,b_1,b_2;c_{11}(t_{n}),c_{12}(t_{n}),c_{21}(t_{n}),c_{22}(t_{n})) \in \mathcal{R}(\mathbf{q}). 
    \end{equation*}
    We then have
    \begin{equation*}
        \mathcal{I}_{2}(\mathbf{v}_{n}) \to \mathcal{I}_{2}(\mathbf{q}^{[2]}), 
    \end{equation*}
    which contradicts
    \begin{align*}
            \mathcal{I}_{2}(\mathbf{v}_{n}) - \mathcal{I}_{2}(\mathbf{q}^{[2]}) &= (\mathcal{L}_{2}\mathbf{z}_{n},\mathbf{z}_{n}) + \mathcal{O}(\|\mathbf{z}_{n}\|_{H^{2}}^{3}) \\
            &\geq C_{1} \|\mathbf{z}_{n}\|_{H^{2}}^{2} + \mathcal{O}(\|\mathbf{z}_{n}\|_{H^{2}}^{3}) \\
            &\geq C_{1} \|\mathbf{u}_{n}(\cdot,t_{n}) - \mathbf{q}^{[2]}(\cdot,t_{n};a,b_1,b_2;c_{11}(t_{n}),c_{12}(t_{n}),c_{21}(t_{n}),c_{22}(t_{n}))\|_{H^{2}}^{2} \\
            &\quad + \|\mathbf{v}_{n} - \mathbf{u}_{n}(\cdot,t_{n})\|_{H^{2}} + \mathcal{O}(\|\mathbf{z}_{n}\|_{H^{2}}^{3}) \\
            &= C_{1}\epsilon_{0} + \|\mathbf{v}_{n} - \mathbf{u}_{n}(\cdot,t_{n})\|_{H^{2}} + \mathcal{O}(\|\mathbf{z}_{n}\|_{H^{2}}^{3}) \\
            &\to C_{1}\epsilon_{0}
    \end{align*}
    for sufficiently small $\|\mathbf{z}_{n}\|_{H^{2}}$ and large $n$. The estimate \eqref{derivative-cij} follows from \eqref{estimate-derivative}.
\end{proof}

\appendix

\titleformat{\section}[display]
{\centering\LARGE\bfseries}{ }{11pt}{\LARGE}

\renewcommand{\appendixname}{Appendix \ \Alph{section}}

\section{\appendixname. Inverse scattering transform}
\label{A.ISM}

\setcounter{equation}{0}
\setcounter{definition}{0}
\setcounter{prop}{0}
\renewcommand\theequation{\Alph{section}.\arabic{equation}}

\renewcommand\thedefinition{\Alph{section}.\arabic{definition}}
\renewcommand\theprop{\Alph{section}.\arabic{prop}}

As the generalization of Fourier transform, the inverse scattering method can be used to 
solve the integrable system and construct their infinite conservation laws \cite{faddeev1987hamiltonian,yang_nonlinear_2010}.

The scattering problem associated with CNLS equations \eqref{CNLS} is defined by the first equation of the Lax pair \eqref{CNLS-lax}. 
Assume that $u$ is of Schwartz class with respect to $x$, 
the first-order differential equation
\begin{equation}\label{spectral-problem-vector}
    \phi_{x}=\mathbf{U}\phi
\end{equation}
has two fundamental solution matrices $\mathbf{\Phi}^{\pm}$ which can refer to the cases $x\to\pm\infty$, i.e., solving the 
ODE \eqref{spectral-problem-vector} with boundary $\phi\to {\rm e}^{{\rm i}\lambda\sigma_{3}x}$ for $x\to\pm\infty$. The matrices 
$\mathbf{\Phi}^{\pm}$ have the asymptotic expression respectively
\begin{equation}\label{asy-Phi}
    \mathbf{\Phi}^{\pm}\sim {\rm e}^{{\rm i}\lambda\sigma_{3}x}, \,\,\,\, x\to \pm\infty.
\end{equation}
Denote $\mathbf{\Phi}^{\pm} =\begin{pmatrix}
    \phi_{1}^{\pm}& \phi_{2}^{\pm} & \phi_{3}^{\pm}
\end{pmatrix}$, 
the vector solutions $\phi_{1}^{+}$, $\phi_{2}^{-}, \phi_{3}^{-}$ are holomorphic 
on $\Omega_{+}=\{\lambda\in\mathbb{C}:\mathrm{Im}{(\lambda)}> 0\}$ and 
the vector solutions $\phi_{1}^{-}$, $\phi_{2}^{+}, \phi_{3}^{+}$ are holomorphic 
on $\Omega_{-}=\{\lambda\in\mathbb{C}:\mathrm{Im}{(\lambda)}< 0\}$. 
Additionally, $\phi_{1}^{\pm}$, $\phi_{2}^{\pm}$ and $\phi_{3}^{\pm}$ are smooth up to the boundary. 
For $\mathrm{Im}{(\lambda)}=0$ i.e. $\lambda\in \mathbb{R}$, the matrices $\mathbf{\Phi}^{\pm}$ are smooth and 
there is a transfer matrix $\mathbf{S}(\lambda;t)=\begin{pmatrix} s_{ij}(\lambda;t) \end{pmatrix}_{1\leq i,j \leq 3}$ satisfying
\begin{equation}\label{sca-mat}
    \mathbf{\Phi}^{-}(\lambda;x,t)=\mathbf{\Phi}^{+}(\lambda;x,t)\mathbf{S}(\lambda;t).
\end{equation}

Let us discuss the symmetry of the Lax pair \eqref{CNLS-lax} for CNLS equations \eqref{CNLS}. Using the symmetry $\mathbf{Q}^{\dagger}=\mathbf{Q}$, it can be verified that
\begin{equation*}
    \mathbf{U}^{\dagger}(\lambda;x,t)=-\mathbf{U}(\lambda^{*};x,t). 
\end{equation*}
The symmetries for $\mathbf{U},\mathbf{V}$ lead to
\begin{equation*}
    \begin{split}
        -\mathbf{\Phi}^{\dagger}(\lambda^{*};x,t)_{x}&=\mathbf{\Phi}^{\dagger}(\lambda^{*};x,t)\mathbf{U}(\lambda;x,t), \\
        -\mathbf{\Phi}^{\dagger}(\lambda^{*};x,t)_{t}&=\mathbf{\Phi}^{\dagger}(\lambda^{*};x,t)\mathbf{V}(\lambda;x,t).
    \end{split}
\end{equation*}
Hence the symmetries for $\mathbf{\Phi}$ are
\begin{equation*}
%\label{sym-Phi}
    \mathbf{\Phi}^{-1}(\lambda;x,t)=\mathbf{\Phi}^{\dagger}(\lambda^{*};x,t), 
\end{equation*}
and the symmetries for $\mathbf{S}$ are
\begin{equation*}
%\label{sym-S1}
    \mathbf{S}^{-1}(\lambda;x,t)=\mathbf{S}^{\dagger}(\lambda^{*};x,t). 
\end{equation*}

Setting $\phi_{i}^{\pm}=\begin{pmatrix} \phi_{1i}^{\pm}&\phi_{2i}^{\pm}&\phi_{3i}^{\pm} \end{pmatrix}^{T}$, 
by \eqref{sca-mat}, we obtain 
\begin{align*}
        \mathbf{S}(\lambda;t)&=(\mathbf{\Phi}^{+}(\lambda;x,t))^{-1}\mathbf{\Phi}^{-}(\lambda;x,t)\\
        &=\det(\mathbf{\Phi}^{+})^{-1}\left(\begin{matrix}
            A_{11}^{+}&A_{21}^{+} &A_{31}^{+} \\
            A_{12}^{+}&A_{22}^{+} &A_{32}^{+} \\
            A_{13}^{+}&A_{23}^{+} &A_{33}^{+} \\
        \end{matrix}\right)\left(
            \begin{matrix}
                \phi_{11}^{-}&\phi_{12}^{-} &\phi_{13}^{-} \\
                \phi_{21}^{-}&\phi_{22}^{-}&\phi_{23}^{-} \\
                \phi_{31}^{-}&\phi_{32}^{-}&\phi_{33}^{-}
            \end{matrix}\right).
    \end{align*}
This yields
\begin{align*}
        s_{11}(\lambda;t)=\det(\mathbf{\Phi}^{+})^{-1}\det{(\phi_{1}^{-},\phi_{2}^{+},\phi_{3}^{+})}, \\
        s_{22}(\lambda;t)=\det(\mathbf{\Phi}^{+})^{-1}\det{(\phi_{1}^{+},\phi_{2}^{-},\phi_{3}^{+})}, \\
        s_{33}(\lambda;t)=\det(\mathbf{\Phi}^{+})^{-1}\det{(\phi_{1}^{+},\phi_{2}^{+},\phi_{3}^{-})},
    \end{align*}
with $s_{11}(\lambda;t)$ being a holomorphic function on $\Omega_{-}$.

\subsection{Conservation laws and trace formulas}
\label{app-A-1}

For $v=\begin{pmatrix} v_{1}&v_{2}&v_{3} \end{pmatrix}^{T}=
\begin{pmatrix} v_{1}&\tilde{v}^{T}\end{pmatrix}^{T}$, we consider the corresponding differential equation associated with 
\eqref{CNLS-lax}
\begin{equation*}
    \begin{split}
        v_{x}=\mathbf{U}v,\\
        v_{t}=\mathbf{V}v.
    \end{split}
\end{equation*}
The term $\omega=\tilde{v}/v_{1}$ satisfies the Riccati equation
\begin{equation}\label{Ric-equ}
    \omega_{x}={\rm i} \mathbf{q}-2{\rm i}\lambda\omega-{\rm i}\omega \mathbf{r}^{T} \omega.
\end{equation}
Since $\omega$ is holomorphic in $\mathbb{C}$ about $\lambda$, assuming 
\begin{equation}\label{Ric-exp}
    \omega(\lambda;x,t)=\sum_{n=1}^{+\infty}\frac{\omega_{n}(x,t)}{(2{\rm i} \lambda)^{n}},
\end{equation}
and substituting \eqref{Ric-exp} into \eqref{Ric-equ}, it leads to
\begin{align*}
        \omega_{1}&={\rm i}\mathbf{q}, \\
        \omega_{2}&=-{\rm i}\mathbf{q}_{x}, \\
        \omega_{3}&={\rm i}\mathbf{q}_{xx}+{\rm i}\mathbf{q}\mathbf{r}^{T}\mathbf{q},\\
        \omega_{4}&=-{\rm i}\mathbf{q}_{xxx}-{\rm i}\mathbf{q}\mathbf{r}_{x}^{T}\mathbf{q}-{\rm i}\mathbf{q}_{x}\mathbf{r}^{T}\mathbf{q}
        -{\rm i}\mathbf{q}\mathbf{r}^{T}\mathbf{q}_{x}\\
        \omega_{5}&={\rm i}\mathbf{q}_{xxxx}+{\rm i}(3\mathbf{q}_{x}\mathbf{r}_{x}^{T}\mathbf{q}+3\mathbf{q}\mathbf{r}_{x}^{T}\mathbf{q}_{x}
        +\mathbf{q}\mathbf{r}_{xx}^{T}\mathbf{q}+3\mathbf{q}_{xx}\mathbf{r}^{T}\mathbf{q}\\
        &\quad +3\mathbf{q}\mathbf{r}^{T}\mathbf{q}_{xx}
        +5\mathbf{q}_{x}\mathbf{r}^{T}\mathbf{q}_{x}
        +2\mathbf{q}\mathbf{r}^{T}\mathbf{q}\mathbf{r}^{T}\mathbf{q})\\
        &\cdots \\
        \omega_{n}&=-\omega_{n-1,x}-{\rm i}\sum_{i=1}^{n-2}\omega_{i}\mathbf{r}^{T}\omega_{n-1-i}, \quad n\geq 1.
    \end{align*}
Since
\begin{equation}\label{ln-v1}
    \begin{split}
        (\ln v_{1})_{x}&={\rm i}\lambda+{\rm i} \mathbf{r}^{T} \omega, \\
        (\ln v_{1})_{t}&={\rm i}\lambda^2-\frac{1}{2}{\rm i}\mathbf{r}^{T}\mathbf{q}+
        ({\rm i}\lambda \mathbf{r}^{T}+\frac{1}{2}\mathbf{r}^{T}_{x})\omega,
    \end{split}
\end{equation}
using the compatibility condition, the relation $(\ln v_{1})_{xt}=(\ln v_{1})_{tx}$ reads
\begin{equation}\label{conse}
    {\rm i} (\mathbf{r}^{T} \omega)_{t}=
    (-\frac{1}{2}{\rm i}\mathbf{r}^{T}\mathbf{q}+
    ({\rm i}\lambda \mathbf{r}^{T}+\frac{1}{2}\mathbf{r}^{T}_{x})\omega)_{x}.
\end{equation}
With the symmetry $\mathbf{q}=\mathbf{r}^{*}$ and the boundary condition ${\mathbf{q}}$ is of Schwartz class about $x$, 
integrating both sides of equation \eqref{conse} with respect to $x$ on the real line, it gives rise to
\begin{equation*}
    {\rm i} \frac{d}{dt} \int_{\mathbb{R}}\mathbf{r}^{T}\omega\mathrm{d}x =0.
\end{equation*}
Grouping the terms with respect to $\lambda$, it concludes that CNLS equations \eqref{CNLS} admit 
an infinite number of conservation laws $\int_{\mathbb{R}}q\omega_{n}\mathrm{d}x, n\geq 1$, including the conservation of 
mass, momentum and energy: 
\begin{equation}\label{hami}
    \begin{split}
        H_{0}(\mathbf{q})&:=-\frac{1}{2}\int_{\mathbb{R}}{\rm i}\mathbf{r}^{T}\omega_{1}\mathrm{d}x
        =\frac{1}{2}\int_{\mathbb{R}}\left| \mathbf{q} \right|^{2}\mathrm{d}x, \\
        H_{1}(\mathbf{q})&:=\frac{{\rm i}}{2}\int_{\mathbb{R}}{\rm i}\mathbf{r}^{T}\omega_{2}\mathrm{d}x
        =\frac{1}{2}\int_{\mathbb{R}}{\rm i} \mathbf{q}_{x}^{T}\mathbf{q}^{*} \mathrm{d}x, \\
        H_{2}(\mathbf{q})&:=\frac{1}{2}\int_{\mathbb{R}}{\rm i}\mathbf{r}^{T}\omega_{3}\mathrm{d}x
        =-\frac{1}{2}\int_{\mathbb{R}} \left( \mathbf{q}_{xx}^{T}\mathbf{q}^{*}
        +\left|\mathbf{q}\right|^{4} \right) \mathrm{d}x=\frac{1}{2}\int_{\mathbb{R}} \left( |\mathbf{q}_{x}|^{2}
        -\left|\mathbf{q}\right|^{4} \right) \mathrm{d}x, \\
        \cdots \\
        H_{n-1}(\mathbf{q})&:=-\frac{(-{\rm i})^{n-1}}{2}\int_{\mathbb{R}}{\rm i}\mathbf{r}^{T}\omega_{n}\mathrm{d}x. 
    \end{split}
\end{equation}

Now we induce the trace formulas. 
From the conservation laws, we can obtain the expansion of $a(\lambda;t):=\hat{s}_{11}(\lambda;t)$. 
In view of the asymptotic expansion \eqref{asy-Phi} and the determinant representation, we obtain 
\begin{equation*}
    a(\lambda;t)
=\lim_{x\to-\infty}\det(\mathbf{\Phi}^{-}(x))^{-1}\phi_{11}^{+}{\rm e}^{-2{\rm i}\lambda x}
=\lim_{x\to-\infty}\phi_{11}^{+}{\rm e}^{-{\rm i}\lambda x}. 
\end{equation*}
There exists a 
function $g(\lambda;x,t)$ such that $\phi_{11}^{+}(\lambda;x,t)={\rm e}^{{\rm i}\lambda x+g(\lambda;x,t)}$. 
In view of \eqref{ln-v1} and \eqref{asy-Phi}, 
we obtain $g_{x}={\rm i} \mathbf{r}^{T} \omega$ and $\lim_{x\to+\infty}g(\lambda;x,t)=0$. Hence
\begin{equation}\label{exp-a-1}
    \ln a(\lambda;t)=\lim_{x\to-\infty}g(\lambda;x,t)=-\int_{-\infty}^{+\infty}{\rm i} \mathbf{r}^{T} \omega \mathrm{d}x
    =-\sum_{n=1}^{+\infty}\int_{-\infty}^{+\infty}\frac{{\rm i} \mathbf{r}^{T}\omega_{n}(x,t)}{(2{\rm i} \lambda)^{n}}\mathrm{d}x,
\end{equation}
by Lebesgue dominated convergence theorem. Since $a(\lambda)$ is holomorphic and the simple zeros of $a(\lambda;t)$ correspond to the point spectrum of a 
nonself-adjoint operator, we assume that $\{\lambda_{1},\lambda_{2},\cdots,\lambda_{N}\}$ 
are all simple zeros of $a(\lambda;t)$, then 
we define
\begin{equation*}
    \tilde{a}(\lambda;t)=a(\lambda;t)
    \prod_{i=1}^{N}\frac{\lambda-\lambda_{i}^*}{\lambda-\lambda_{i}}
\end{equation*}
which is holomorphic, has no zero and tend to 1 as $\lambda\to\infty$. Moreover, $\eta(\lambda;t)=\ln\tilde{a}(\lambda;t)$ is holomorphic in $\Omega_{+}$ and continuous to the 
real line. For the value on real line, by Sochocki-Plemelj formula, since $\eta$ vanish at infinity, for $\lambda$ in the real line, 
we obtain
\begin{align*}
        \eta(\lambda;t)&=\frac{1}{2\pi{\rm i}}\left( \lim_{\epsilon\to 0^{+}} \int_{-\infty}^{+\infty}\frac{\eta(\mu;t)}{\mu-\lambda-{\rm i}\epsilon}\mathrm{d}\mu \right)\\
        &=\frac{1}{2\pi{\rm i}}\left( \mathrm{p.v.}\int_{-\infty}^{+\infty}\frac{\eta(\mu;t)}{\mu-\lambda}\mathrm{d}\mu 
        + {\rm i}\pi \eta(\lambda;t)\right), 
    \end{align*}
hence
\begin{equation*}
        \eta(\lambda;t) =\frac{1}{\pi{\rm i}}\mathrm{p.v.}\int_{-\infty}^{+\infty}\frac{\eta(\mu;t)}{\mu-\lambda}\mathrm{d}\mu .
\end{equation*}
Then the imaginary part is 
\begin{equation*}
    \mathrm{Im} \eta(\lambda;t)= -\frac{1}{\pi}\mathrm{p.v.}\int_{-\infty}^{+\infty}\frac{\mathrm{Re}\eta(\mu;t)}{\mu-\lambda}\mathrm{d}\mu .
\end{equation*}
By Sochocki-Plemelj formula again, we obtain
\begin{equation*}
    \mathrm{Im} \eta(\lambda;t)=-\frac{1}{\pi}\int_{-\infty}^{+\infty}\frac{\mathrm{Re}\eta(\mu;t)}{\mu-\lambda-{\rm i}0}\mathrm{d}\mu+{\rm i}
    \mathrm{Re} \eta(\lambda;t).
\end{equation*}
For $\mathrm{Im}(\lambda)>0$, 
\begin{align*}
%\label{Im-eta}
        \frac{1}{2\pi{\rm i}}\int_{-\infty}^{+\infty}\frac{\mathrm{Im} \eta(\mu;t)}{\mu-\lambda}\mathrm{d}\mu
        =&\frac{1}{2\pi{\rm i}}\int_{-\infty}^{+\infty}\frac{1}{\mu-\lambda}
        \left(-\frac{1}{\pi}\int_{-\infty}^{+\infty}\frac{\mathrm{Re}\eta(\mu';t)}{\mu'-\mu-{\rm i}0}\mathrm{d}\mu'+{\rm i}
        \mathrm{Re} \eta(\mu;t)\right)
        \mathrm{d}\mu\\
        =&\frac{1}{2\pi}\int_{-\infty}^{+\infty}\frac{\mathrm{Re} \eta(\mu;t)}{\mu-\lambda}\mathrm{d}\mu
        \\
        &-
        \frac{1}{\pi}\int_{-\infty}^{+\infty}\mathrm{Re}\eta(\mu';t)\frac{1}{2\pi{\rm i}}\int_{-\infty}^{+\infty}
        \frac{1}{\mu-\lambda}\frac{1}{\mu'-\mu-{\rm i}0}\mathrm{d}\mu\mathrm{d}\mu'\\
        =&\frac{1}{2\pi}\int_{-\infty}^{+\infty}\frac{\mathrm{Re} \eta(\mu;t)}{\mu-\lambda}\mathrm{d}\mu-\frac{1}{\pi}\int_{-\infty}^{+\infty}
        \frac{\mathrm{Re}\eta(\mu';t)}{\mu'-\lambda-{\rm i}0}\mathrm{d}\mu'\\
        =&-\frac{1}{2\pi}\int_{-\infty}^{+\infty}\frac{\mathrm{Re} \eta(\mu;t)}{\mu-\lambda}\mathrm{d}\mu.
\end{align*}
Then we obtain
\begin{equation*}
    \eta(\lambda;t)=\frac{1}{2\pi{\rm i}}\int_{-\infty}^{+\infty}\frac{\eta(\mu;t)}{\mu-\lambda}\mathrm{d}\mu=\frac{1}{\pi{\rm i}}\int_{-\infty}^{+\infty}\frac{\mathrm{Re}\eta(\mu;t)}{\mu-\lambda}\mathrm{d}\mu,
\end{equation*}
for $\mathrm{Im}(\lambda)>0$. Since $\mathrm{Re}\eta(\mu;t)=\ln |\tilde{a}(\mu)|$, it leads to
\begin{align*}
        a(\lambda;t)&={\rm e}^{\eta(\lambda;t)}
        \prod_{i=1}^{N}\frac{\lambda-\lambda_{i}}{\lambda-\lambda_{i}^*}\\
        &=
        {\rm e}^{\frac{1}{2\pi{\rm i}}\int_{-\infty}^{+\infty}\frac{\ln |\tilde{a}(\mu)|^{2}}{\mu-\lambda}\mathrm{d}\mu}
        \prod_{i=1}^{N}\frac{\lambda-\lambda_{i}}{\lambda-\lambda_{i}^*}
    \end{align*}
or
\begin{equation}\label{exp-a-2}
    \begin{split}
        \ln a(\lambda;t)&=
        \frac{1}{2\pi{\rm i}}\int_{-\infty}^{+\infty}\frac{\ln |\tilde{a}(\mu)|^{2}}{\mu-\lambda}\mathrm{d}\mu+
        \sum_{i=1}^{N}\left( \ln \left( 1-\frac{\lambda_{i}}{\lambda} \right) - \ln \left( 1-\frac{\lambda_{i}^*}{\lambda} \right) \right)\\
        &=\sum_{n=1}^{+\infty}\left(
            -\frac{1}{2\pi{\rm i}}\int_{-\infty}^{+\infty}\ln |\tilde{a}(\mu)|^{2}\mu^{n-1}\mathrm{d}\mu\right)
            \frac{1}{\lambda^{n}}-
            \sum_{n=1}^{+\infty}\frac{2{\rm i}}{n}\left(\sum_{i=1}^{N}\mathrm{Im} \lambda_{i}^{n}
        \right)\frac{1}{\lambda^{n}}.
    \end{split}
\end{equation}
Combining \eqref{hami} ,\eqref{exp-a-1} and \eqref{exp-a-2}, the conservation laws can be represented by the point spectrum and 
the essential spectrum of $a(\lambda;t)$:
\begin{equation}\label{trace-formula}
    \begin{split}
        H_{n-1}&=-\frac{2^{n-2}}{\pi}\int_{-\infty}^{+\infty}\ln |\tilde{a}(\mu)|^{2}\mu^{n-1}\mathrm{d}\mu
        +\frac{2^{n}}{n}
        \sum_{i=1}^{N}\mathrm{Im} \lambda_{i}^{n},\quad n\geq 1.
    \end{split}
\end{equation}
Then we can prove that the $N$-solitons satisfies an ODE. 

\begin{prop}
	\label{prop-app-A}
    The $N$-soliton solutions satisfy
    \begin{equation*}
    %\label{ODE1}
        \sum_{n=0}^{2N}\mu_{n}\frac{\delta H_{n}}{\delta \mathbf{q}}=0,
    \end{equation*}
    where the coefficients are given by
    \begin{equation*}
        \mathcal{P}_{N}(\lambda)=\prod_{i=1}^{N}(\lambda-\lambda_{i})(\lambda-{\lambda_{i}^{*}})=
        \sum_{n=0}^{2N}2^{n-2N}\mu_{n}
            \lambda^{n}
    \end{equation*}
\end{prop}

\begin{proof}
    By trace formula \eqref{trace-formula}, we have
    \begin{equation*}
    %\label{gradient-conservation-law}
        \frac{\delta H_{n}}{\delta \mathbf{q}}=2^{n}\cdot \frac{1}{2{\rm i}}
        \sum_{i=1}^{N}(\lambda_{i}^{n}\frac{\delta \lambda_{i}}{\delta \mathbf{q}}  -\lambda_{i}^{*n}
        \frac{\delta \lambda_{i}^{*}}{\delta \mathbf{q}} ),\quad n\geq 0. 
    \end{equation*}
    Hence the ODE
    \begin{equation*}
            \sum_{n=0}^{2N}\mu_{n}\frac{\delta H_{n}}{\delta \mathbf{q}}
            =2^{2N}\cdot \frac{1}{2{\rm i}}
            \sum_{i=1}^{N}(\mathcal{P}_{N}(\lambda_{i})\frac{\delta \lambda_{i}}{\delta \mathbf{q}} 
            -\mathcal{P}_{N}(\lambda_{i}^{*})\frac{\delta \lambda_{i}^{*}}{\delta \mathbf{q}} ) =0
    \end{equation*}
    since $\lambda_{i},\lambda_{i}^{*}$ are roots of $\mathcal{P}_{N}(\lambda)$. 
\end{proof}

\subsection{The integrable hierarchy and the matrix $\mathbf{L}$}
\label{app.A2}

Let
\begin{equation*}
    \mathbf{L}=\mathbf{\Phi}\sigma_{3}\mathbf{\Phi}^{-1},
\end{equation*}
then 
\begin{equation*}
    \mathbf{L}_{x}=[\mathbf{U},\mathbf{L}],\quad \mathbf{L}\mathbf{L}^{-1}=\mathbb{I}, \quad \lim_{\lambda\to\infty}\mathbf{L}=\sigma_{3}.
\end{equation*}
We use these conditions to calculate the integrable equation hierarchy. Setting
\begin{equation*}
%\label{L-matrix}
    \mathbf{L}=\sigma_{3}+\sum_{n=1}^{\infty}\frac{\mathbf{L}_{n}}{\lambda^{n}},
\end{equation*}
we have
\begin{align*}
        \mathbf{L}_{1}&=\mathbf{Q}, \\
        \mathbf{L}_{n+1}^{diag}&={\rm i}\partial_{x}^{-1}\mathrm{ad}_{\mathbf{Q}}\mathbf{L}_{n+1}^{off}=
        \frac{1}{2}\partial_{x}^{-1}\mathrm{ad}_{\mathbf{Q}}\sigma_3(\partial_{x}\mathbf{L}_{n}^{off}-
        {\rm i}\mathrm{ad}_{\mathbf{Q}}\mathbf{L}_{n}^{diag}),\\
        \mathbf{L}_{n+1}^{off}&=-\frac{{\rm i}}{2}\sigma_{3}\partial_{x}\mathbf{L}_{n}^{off}-
        \frac{1}{2}\sigma_{3}\mathrm{ad}_{\mathbf{Q}}\mathbf{L}_{n}^{diag}.
    \end{align*}
Moreover, we obtain the recursion
\begin{equation*}
    \mathbf{L}_{n+1}^{off}=-\frac{{\rm i}}{2}\sigma_{3}(\partial_{x}+ \mathrm{ad}_{\mathbf{Q}}\partial_{x}^{-1}
    \mathrm{ad}_{\mathbf{Q}})\mathbf{L}_{n}^{off}.
\end{equation*}
Setting 
\begin{equation*}
    \mathbf{L}_{n}^{off}=
    \begin{pmatrix}
        0 & \mathbf{W}_{n}^{T} \\
        \mathbf{V}_{n} & \mathbf{0}_{2\times 2}
    \end{pmatrix},
\end{equation*}
we obtain
\begin{equation*}
    \begin{pmatrix}
        \mathbf{V}_{n+1} \\ \mathbf{W}_{n+1}
    \end{pmatrix}=\frac{{\rm i}}{2}
    \begin{pmatrix}
        \partial_{x}\mathbf{V}_{n}+\partial_{x}^{-1}(\mathbf{V}_{n}\mathbf{r}^{T}-\mathbf{q}\mathbf{W}_{n}^{T})\mathbf{q}
        +\partial_{x}^{-1}(\mathbf{r}^{T}\mathbf{V}_{n}-\mathbf{W}_{n}^{T}\mathbf{q})\mathbf{q}
        \\ 
        -\partial_{x}\mathbf{W}_{n}-\partial_{x}^{-1}(\mathbf{W}_{n}\mathbf{q}^{T}-\mathbf{r}\mathbf{V}_{n}^{T})\mathbf{r}
        +\partial_{x}^{-1}(\mathbf{r}^{T}\mathbf{V}_{n}-\mathbf{W}_{n}^{T}\mathbf{q})\mathbf{r}
    \end{pmatrix},
\end{equation*}
where
\begin{equation*}
\partial_{x}^{-1}\cdot=
\int_{-\infty}^{x} (\cdot) dx,
\end{equation*}
Let the recursion operator be denoted by 
\begin{equation}\label{reu-op}
    \mathcal{L}_{r}\begin{pmatrix}
        \mathbf{g} \\ \mathbf{h}
    \end{pmatrix}=\frac{{\rm i}}{2}
    \begin{pmatrix}
        \partial_{x}\mathbf{g}+\partial_{x}^{-1}(\mathbf{g}\mathbf{r}^{T}+\mathbf{q}\mathbf{h}^{T})
        \mathbf{q}+\partial_{x}^{-1}(\mathbf{r}^{T}\mathbf{g}+\mathbf{h}^{T}\mathbf{q})\mathbf{q}
        \\
        -\partial_{x}\mathbf{h}-\partial_{x}^{-1}(\mathbf{h}\mathbf{q}^{T}+\mathbf{r}\mathbf{g}^{T})
        \mathbf{r}-\partial_{x}^{-1}(\mathbf{r}^{T}\mathbf{g}+\mathbf{h}^{T}\mathbf{q})\mathbf{r}
    \end{pmatrix}.
\end{equation}
Then we have
\begin{equation*}
    \begin{pmatrix}
        \mathbf{V}_{n+1} \\ -\mathbf{W}_{n+1}
    \end{pmatrix}=\mathcal{L}_{r}\begin{pmatrix}
        \mathbf{V}_{n} \\ -\mathbf{W}_{n}
    \end{pmatrix}.
\end{equation*}
By straightforward calculation, 
\begin{align*}
        \mathbf{L}_{0}&=\sigma_3, \\
        \mathbf{L}_{1}&=\mathbf{Q}, \\
        \mathbf{L}_{2}&=-\frac{1}{2}\sigma_3\mathbf{Q}^{2}-\frac{{\rm i}}{2}\sigma_{3}\mathbf{Q}_{x}, \\
        \mathbf{L}_{3}&=-\frac{{\rm i}}{4}(\mathbf{Q}\mathbf{Q}_{x}-\mathbf{Q}_{x}\mathbf{Q})-
        \frac{1}{4}(\mathbf{Q}_{xx}+2\mathbf{Q}^{3}), \\
        \mathbf{L}_{4}&=\frac{1}{8}\sigma_{3}
        (\mathbf{Q}_{xx}\mathbf{Q}+\mathbf{Q}\mathbf{Q}_{xx}-\mathbf{Q}_{x}^{2}+
        3\mathbf{Q}^{4})+\frac{{\rm i}}{8}\sigma_{3}
        (\mathbf{Q}_{xxx}+3\mathbf{Q}_{x}\mathbf{Q}^{2}+
        3\mathbf{Q}^{2}\mathbf{Q}_{x}).
    \end{align*}
The matrix $\mathbf{U}$ and $\mathbf{V}$ in \eqref{CNLS-lax} have representations
\begin{equation*}
    \mathbf{U}={\rm i}(\mathbf{L}_{0}\lambda+\mathbf{L}_{1}), \quad
    \mathbf{V}={\rm i}(\mathbf{L}_{0}\lambda^{2}+\mathbf{L}_{1}\lambda+\mathbf{L}_{2}).
\end{equation*}
In particular, we have the diagonal elements
\begin{align*}
        \mathbf{L}_{1}^{diag}&=0, \\
        \mathbf{L}_{2}^{diag}&=-\frac{1}{2}\sigma_3\mathbf{Q}^{2}, \\
        \mathbf{L}_{3}^{diag}&=-\frac{{\rm i}}{4}(\mathbf{Q}\mathbf{Q}_{x}-\mathbf{Q}_{x}\mathbf{Q}), \\
        \mathbf{L}_{4}^{diag}&=\frac{1}{8}
        (\mathbf{Q}_{xx}\mathbf{Q}+\mathbf{Q}\mathbf{Q}_{xx}-\mathbf{Q}_{x}^{2}+
        3\mathbf{Q}^{4})
    \end{align*}
and the off-diagonal elements
\begin{align*}
        \mathbf{L}_{1}^{off}&=\mathbf{Q}, \\
        \mathbf{L}_{2}^{off}&=-\frac{{\rm i}}{2}\sigma_{3}\mathbf{Q}_{x}, \\
        \mathbf{L}_{3}^{off}&=-\frac{1}{4}(\mathbf{Q}_{xx}+2\mathbf{Q}^{3}), \\
        \mathbf{L}_{4}^{off}&=\frac{{\rm i}}{8}\sigma_{3}
        (\mathbf{Q}_{xxx}+3\mathbf{Q}_{x}\mathbf{Q}^{2}+
        3\mathbf{Q}^{2}\mathbf{Q}_{x}), \\
        \mathbf{L}_{5}^{off}&=\frac{1}{16}
        (\mathbf{Q}_{xxxx}+4\mathbf{Q}_{xx}\mathbf{Q}^{2}+2\mathbf{Q}\mathbf{Q}_{xx}\mathbf{Q}+
        4\mathbf{Q}^{2}\mathbf{Q}_{xx}\\&\quad +
        2\mathbf{Q}_{x}^{2}\mathbf{Q}+6\mathbf{Q}_{x}\mathbf{Q}\mathbf{Q}_{x}+
        2\mathbf{Q}\mathbf{Q}_{x}^{2}+
        6\mathbf{Q}^{5}
        ). 
    \end{align*}

\section{\appendixname. The asymptotic expansion of $\mathbf{Q}_{ij}$}
\label{B.asy-exp}

\setcounter{equation}{0}

Assume that $0< b_{2}< b_{1}$, we have
    \begin{align*}
            M_{non}&\sim |c_{11}c_{22}|^{2}{\rm e}^{2(b_{1}+b_{2})\xi} \quad \xi\to\infty, \\
            M_{non}&\sim \frac{(b_{1}-b_{2})^{2}}{(b_{1}+b_{2})^{2}}{\rm e}^{-2(b_{1}+b_{2})\xi} \quad \xi\to-\infty.
        \end{align*}
    If $\lambda\notin \{\lambda_{1},\lambda_{2},\lambda_{1}^{*},\lambda_{2}^{*}\}$, 
    the asymptotics are given by
    \begin{equation*}
        \begin{pmatrix}
            Q_{11}(\xi) & 
            Q_{12}(\xi) &
            Q_{13}(\xi) \\
            Q_{21}(\xi) &
            Q_{22}(\xi) &
            Q_{23}(\xi) \\
            Q_{31}(\xi) &
            Q_{32}(\xi) &
            Q_{33}(\xi)
        \end{pmatrix}\in
        {\rm e}^{-\sigma_{3}\mathrm{Im}(\lambda)\xi}(L^{\infty}(\mathbb{R}))^{3\times 3}
        ,\quad 
        \xi \to \infty. 
    \end{equation*}
    For $\lambda=\lambda_{1}$, we have
    \begin{equation*}
        \begin{pmatrix}
            Q_{11}(\xi) & 
            Q_{12}(\xi) &
            Q_{13}(\xi) \\
            Q_{21}(\xi) &
            Q_{22}(\xi) &
            Q_{23}(\xi) \\
            Q_{31}(\xi) &
            Q_{32}(\xi) &
            Q_{33}(\xi)
        \end{pmatrix}\sim 
        \begin{pmatrix}
            {\rm e}^{-b_{1}\xi} & {\rm e}^{-b_{1}\xi} & {\rm e}^{-(2b_{2}-b_{1})\xi} \\ 
            {\rm e}^{-3b_{1}\xi} & {\rm e}^{-3b_{1}\xi} &  {\rm e}^{-(b_{1}+2b_{2})\xi}\\ 
            {\rm e}^{-(b_{1}+2b_{2})\xi} & {\rm e}^{-(b_{1}+2b_{2})\xi} & {\rm e}^{b_{1}\xi}
        \end{pmatrix},\quad 
        \xi \to +\infty
    \end{equation*}
    and
    \begin{equation*}
        \begin{pmatrix}
            Q_{11}(\xi) & 
            Q_{12}(\xi) &
            Q_{13}(\xi) \\
            Q_{21}(\xi) &
            Q_{22}(\xi) &
            Q_{23}(\xi) \\
            Q_{31}(\xi) &
            Q_{32}(\xi) &
            Q_{33}(\xi)
        \end{pmatrix}\sim 
        \begin{pmatrix}
            {\rm e}^{3b_{1}\xi} &  {\rm e}^{3b_{1}\xi} & {\rm e}^{(b_{1}+2b_{2})\xi}\\
            {\rm e}^{b_{1}\xi} &  {\rm e}^{b_{1}\xi} & {\rm e}^{(3b_{1}+2b_{2})\xi}\\
            {\rm e}^{(3b_{1}+2b_{2})\xi} & {\rm e}^{(3b_{1}+2b_{2})\xi} & {\rm e}^{b_{1}\xi}
        \end{pmatrix},\quad 
        \xi \to -\infty. 
    \end{equation*}
    For $\lambda=\lambda_{1}^{*}$, we have
    \begin{equation*}
        \begin{pmatrix}
            Q_{11}(\xi) & 
            Q_{12}(\xi) &
            Q_{13}(\xi) \\
            Q_{21}(\xi) &
            Q_{22}(\xi) &
            Q_{23}(\xi) \\
            Q_{31}(\xi) &
            Q_{32}(\xi) &
            Q_{33}(\xi)
        \end{pmatrix}\sim 
        \begin{pmatrix}
            {\rm e}^{-3b_{1}\xi} & {\rm e}^{-3b_{1}\xi} & {\rm e}^{-(b_{1}+2b_{2})\xi} \\ 
            {\rm e}^{-b_{1}\xi} & {\rm e}^{-b_{1}\xi} &  {\rm e}^{-(3b_{1}+2b_{2})\xi}\\ 
            {\rm e}^{-(3b_{1}+2b_{2})\xi} & {\rm e}^{-(3b_{1}+2b_{2})\xi} & {\rm e}^{-b_{1}\xi}
        \end{pmatrix},\quad 
        \xi \to +\infty
    \end{equation*}
    and
    \begin{equation*}
        \begin{pmatrix}
            Q_{11}(\xi) & 
            Q_{12}(\xi) &
            Q_{13}(\xi) \\
            Q_{21}(\xi) &
            Q_{22}(\xi) &
            Q_{23}(\xi) \\
            Q_{31}(\xi) &
            Q_{32}(\xi) &
            Q_{33}(\xi)
        \end{pmatrix}\sim 
        \begin{pmatrix}
            {\rm e}^{b_{1}\xi} &  {\rm e}^{b_{1}\xi} & {\rm e}^{(2b_{2}-b_{1})\xi}\\
            {\rm e}^{3b_{1}\xi} &  {\rm e}^{3b_{1}\xi} & {\rm e}^{(b_{1}+2b_{2})\xi}\\
            {\rm e}^{(b_{1}+2b_{2})\xi} & {\rm e}^{(b_{1}+2b_{2})\xi} & {\rm e}^{-b_{1}\xi}
        \end{pmatrix},\quad 
        \xi \to -\infty. 
    \end{equation*}
    Similarly, for $\lambda=\lambda_{2}$, we have
    \begin{equation*}
        \begin{pmatrix}
            Q_{11}(\xi) & 
            Q_{12}(\xi) &
            Q_{13}(\xi) \\
            Q_{21}(\xi) &
            Q_{22}(\xi) &
            Q_{23}(\xi) \\
            Q_{31}(\xi) &
            Q_{32}(\xi) &
            Q_{33}(\xi)
        \end{pmatrix}\sim 
        \begin{pmatrix}
            {\rm e}^{-b_{2}\xi} & {\rm e}^{-(2b_{1}-b_{2})\xi} & {\rm e}^{-b_{2}\xi} \\ 
            {\rm e}^{-(2b_{1}+b_{2})\xi} & {\rm e}^{b_{2}\xi} &  {\rm e}^{-(2b_{1}+b_{2})\xi}\\ 
            {\rm e}^{-3b_{2}\xi} & {\rm e}^{-(2b_{1}+b_{2})\xi} & {\rm e}^{-3b_{2}\xi}
        \end{pmatrix},\quad 
        \xi \to +\infty
    \end{equation*}
    and
    \begin{equation*}
        \begin{pmatrix}
            Q_{11}(\xi) & 
            Q_{12}(\xi) &
            Q_{13}(\xi) \\
            Q_{21}(\xi) &
            Q_{22}(\xi) &
            Q_{23}(\xi) \\
            Q_{31}(\xi) &
            Q_{32}(\xi) &
            Q_{33}(\xi)
        \end{pmatrix}\sim 
        \begin{pmatrix}
            {\rm e}^{3b_{2}\xi} &  {\rm e}^{(2b_{1}+b_{2})\xi} & {\rm e}^{3b_{2}\xi}\\
            {\rm e}^{(2b_{1}+3b_{2})\xi} &  {\rm e}^{b_{2}\xi} & {\rm e}^{(2b_{1}+3b_{2})\xi}\\
            {\rm e}^{b_{2}\xi} & {\rm e}^{(2b_{1}+3b_{2})\xi} & {\rm e}^{b_{2}\xi}
        \end{pmatrix},\quad 
        \xi \to -\infty. 
    \end{equation*}
    For $\lambda=\lambda_{2}^{*}$, we have
    \begin{equation*}
        \begin{pmatrix}
            Q_{11}(\xi) & 
            Q_{12}(\xi) &
            Q_{13}(\xi) \\
            Q_{21}(\xi) &
            Q_{22}(\xi) &
            Q_{23}(\xi) \\
            Q_{31}(\xi) &
            Q_{32}(\xi) &
            Q_{33}(\xi)
        \end{pmatrix}\sim 
        \begin{pmatrix}
            {\rm e}^{-3b_{2}\xi} & {\rm e}^{-(2b_{1}+b_{2})\xi} & {\rm e}^{-3b_{2}\xi} \\ 
            {\rm e}^{-(2b_{1}+3b_{2})\xi} & {\rm e}^{-b_{2}\xi} &  {\rm e}^{-(2b_{1}+3b_{2})\xi}\\ 
            {\rm e}^{-b_{2}\xi} & {\rm e}^{-(2b_{1}+3b_{2})\xi} & {\rm e}^{-b_{2}\xi}
        \end{pmatrix},\quad 
        \xi \to +\infty
    \end{equation*}
    and
    \begin{equation*}
        \begin{pmatrix}
            Q_{11}(\xi) & 
            Q_{12}(\xi) &
            Q_{13}(\xi) \\
            Q_{21}(\xi) &
            Q_{22}(\xi) &
            Q_{23}(\xi) \\
            Q_{31}(\xi) &
            Q_{32}(\xi) &
            Q_{33}(\xi)
        \end{pmatrix}\sim 
        \begin{pmatrix}
            {\rm e}^{b_{2}\xi} &  {\rm e}^{(2b_{1}+3b_{2})\xi} & {\rm e}^{b_{2}\xi}\\
            {\rm e}^{(2b_{1}+b_{2})\xi} &  {\rm e}^{-b_{2}\xi} & {\rm e}^{(2b_{1}+b_{2})\xi}\\
            {\rm e}^{3b_{2}\xi} & {\rm e}^{(2b_{1}+b_{2})\xi} & {\rm e}^{3b_{2}\xi}
        \end{pmatrix},\quad 
        \xi \to -\infty. 
    \end{equation*}
    In addition, to calculate the asymptotic expansion of $\mathbf{P}_{i}(\lambda)$, we also 
    need to calculate the asymptotic expansion of $R_{i}(\lambda)$ at $\lambda\in\{ \lambda_{1},
    \lambda_{1}^{*}, \lambda_{2}, \lambda_{2}^{*} \}$
    \begin{align*}
            &R_{1}(\lambda)\sim {\rm e}^{-\mathrm{Im}(\lambda)\xi},\quad \xi \to \pm\infty\\
            &R_{2}(\lambda)\sim {\rm e}^{-(2b_{1}+\mathrm{Im}(\lambda))\xi},\quad \xi \to +\infty,\quad 
            R_{2}(\lambda)\sim {\rm e}^{(2b_{1}-\mathrm{Im}(\lambda))\xi},\quad \xi \to -\infty,\\
            &R_{3}(\lambda)\sim {\rm e}^{-(2b_{2}+\mathrm{Im}(\lambda))\xi},\quad \xi \to +\infty,\quad
            R_{3}(\lambda)\sim {\rm e}^{(2b_{2}-\mathrm{Im}(\lambda))\xi},\quad \xi \to -\infty. 
    \end{align*}

\section{\appendixname. The closure relation}\label{C.colsure}
\setcounter{equation}{0}

By taking first variations in the first equation of the Lax pair, we obtain
\begin{equation*}
    \partial_{x}\delta \mathbf{\Phi}=\mathbf{U}\delta \mathbf{\Phi}+\delta \mathbf{U} \mathbf{\Phi}.
\end{equation*}
Using the condition $\mathbf{U}=\mathbf{\Phi}_{x}\mathbf{\Phi}^{-1}$, we obtain 
\begin{equation*}
        (\mathbf{\Phi}^{-1}\delta \mathbf{\Phi})_{x}=\mathbf{\Phi}^{-1}\delta \mathbf{U} \mathbf{\Phi},
\end{equation*}
hence 
\begin{equation*}
        \delta \mathbf{\Phi}^{\pm}(\lambda;x,t)=\int_{\pm\infty}^{x}\mathbf{\Phi}^{\pm}(\lambda;x,t){\mathbf{\Phi}^{\pm}}^{-1}(\lambda;y,t)\delta \mathbf{U}(\lambda;y,t) \mathbf{\Phi}^{\pm}(\lambda;y,t)\mathrm{d}y.
\end{equation*}

Using $\mathbf{S}(\lambda;t)={\mathbf{\Phi}^{+}}^{-1}(\lambda;x,t)\mathbf{\Phi}^{-}(\lambda;x,t)$, letting $x\to-\infty$, 
we obtain
\begin{equation*}
    \delta \mathbf{S}(\lambda;t)=
    \int_{-\infty}^{+\infty}{\mathbf{\Phi}^{+}}^{-1}(\lambda;y,t)\delta 
    \mathbf{U}(\lambda;y,t) \mathbf{\Phi}^{-}(\lambda;y,t)\mathrm{d}y.
\end{equation*}
Similarly, 
\begin{equation*}
    \delta \mathbf{S}^{-1}(\lambda;t)=-
    \int_{-\infty}^{+\infty}{\mathbf{\Phi}^{-}}^{-1}(\lambda;y,t)\delta 
    \mathbf{U}(\lambda;y,t) \mathbf{\Phi}^{+}(\lambda;y,t)\mathrm{d}y.
\end{equation*}

First, we assume that the point spectrum of the Lax pair \eqref{CNLS-lax} is an empty set. 
Denote 
\begin{align*}
    \mathbf{\Phi}^{\pm} &=(\phi_{ij}^{\pm})_{1\leq i,j \leq 3}=(\phi_{1}^{\pm}\ \phi_{2}^{\pm}\ \phi_{3}^{\pm}), \\
    (\mathbf{\Phi}^{\pm})^{-1} &=(\hat{\phi}_{ij}^{\pm})_{1\leq i,j \leq 3}
    =(\hat{\phi}_{1}^{\pm}\ \hat{\phi}_{2}^{\pm}\ \hat{\phi}_{3}^{\pm})^{T}, \\
    \mathbf{S}^{-1} &=(\hat{s}_{ij})_{1\leq i,j \leq 3}
\end{align*}
and define
\begin{equation*}
    \rho_{j}=\frac{s_{1,j+1}}{s_{11}}, \quad \hat{\rho}_{j}=\frac{\hat{s}_{j+1,1}}{\hat{s}_{11}}, j=1,2,
\end{equation*}
the variation of $\mathbf{S}, \mathbf{S}^{-1}$ induces that
\begin{equation}\label{drho}
        \begin{split}
            \delta\rho_{j}(\xi)=\frac{{\rm i}}{s_{11}^{2}(\xi)}\int_{\mathbb{R}}  
            \begin{pmatrix}
                \hat{\phi}_{12}^{+}\psi_{1,j+1}^{-} \\
                \hat{\phi}_{13}^{+}\psi_{1,j+1}^{-} \\
                \hat{\phi}_{11}^{+}\psi_{2,j+1}^{-} \\
                \hat{\phi}_{11}^{+}\psi_{3,j+1}^{-}
            \end{pmatrix}^{T}(x;\xi)
            \begin{pmatrix}
                \delta q_{1} \\
                \delta q_{2} \\
                \delta q_{1}^{*} \\
                \delta q_{2}^{*}
            \end{pmatrix}(x)
            \mathrm{d}x,
            \\
            \delta\hat{\rho}_{j}(\xi)=-\frac{{\rm i}}{\hat{s}_{11}^{2}(\xi)}\int_{\mathbb{R}}  
            \begin{pmatrix}
                \hat{\psi}_{j+1,2}^{-}\phi_{11}^{+} \\
                \hat{\psi}_{j+1,3}^{-}\phi_{11}^{+} \\
                \hat{\psi}_{j+1,1}^{-}\phi_{21}^{+} \\
                \hat{\psi}_{j+1,1}^{-}\phi_{31}^{+}
            \end{pmatrix}^{T}(x;\xi)
            \begin{pmatrix}
                \delta q_{1} \\
                \delta q_{2} \\
                \delta q_{1}^{*} \\
                \delta q_{2}^{*}
            \end{pmatrix}(x)
            \mathrm{d}x
        \end{split}
\end{equation}
where 
\begin{align*}
        \psi_{j+1}^{-} &=(\psi_{1,j+1}^{-}, \psi_{2,j+1}^{-}, \psi_{3,j+1}^{-} )^{T}=
        \phi_{j+1}^{-}s_{11}-\phi_{1}^{-}s_{1,j+1}, \\
        \hat{\psi}_{j+1}^{-} &=(\hat{\psi}_{j+1,1}^{-}, \hat{\psi}_{j+1,2}^{-}, \hat{\psi}_{j+1,3}^{-} )=
        \hat{\phi}_{j+1}^{-}\hat{s}_{11}-\hat{\phi}^{-}\hat{s}_{j+1,1}. 
\end{align*}
By \cite{yang_nonlinear_2010}, we obtain
\begin{equation}\label{dqr-1}
    \begin{pmatrix}
        \delta q_{1} ,
        \delta q_{2} ,
        \delta q_{1}^{*} ,
        \delta q_{2}^{*}
    \end{pmatrix}^{T}(x)=\frac{{\rm i}}{\pi}\int_{\mathbb{R}}\sum_{j=1}^{2}\left(
        \mathbf{O}_{j}(x;\xi)\delta\rho_{j}(\xi)+
        \mathbf{O}_{j+2}(x;\xi)\delta\hat{\rho}_{j}(\xi)
    \right)\mathrm{d}\xi, 
\end{equation}
where 
\begin{equation*}
        \mathbf{O}_{j}=
        \begin{pmatrix}
            \phi_{21}^{-}\hat{\phi}_{j+1,1}^{-} \\
            \phi_{31}^{-}\hat{\phi}_{j+1,1}^{-} \\
            -\phi_{11}^{-}\hat{\phi}_{j+1,2}^{-} \\
            -\phi_{11}^{-}\hat{\phi}_{j+1,3}^{-}
        \end{pmatrix},\quad 
        \mathbf{O}_{j+2}=\begin{pmatrix}
            \phi_{2,j+1}^{-}\hat{\phi}_{11}^{-} \\
            \phi_{3,j+1}^{-}\hat{\phi}_{11}^{-} \\
            -\phi_{1,j+1}^{-}\hat{\phi}_{12}^{-} \\
            -\phi_{1,j+1}^{-}\hat{\phi}_{13}^{-}
        \end{pmatrix} 
        ,\quad j=1,2. 
\end{equation*}

Under the symmetries $\mathbf{q}=\mathbf{r}^{*}$, then the symmetry
$\mathbf{S}^{-1}(\lambda;x,t)=\mathbf{S}^{\dagger}(\lambda^{*};x,t)$ is held, and it leads to
$\rho_{j}(\lambda^{*})=(\hat{\rho}(\lambda))^{*}$, hence
\begin{equation*}
    \delta\rho_{j}(\lambda^{*})=(\delta\hat{\rho}(\lambda))^{*}. 
\end{equation*}
Then for $\lambda\in\mathbb{R}$, take complex conjugate in both sides of \eqref{drho}, we obtain
\begin{equation}\label{drho-1}
    \begin{split}
        \delta\hat{\rho}_{j}(\xi) &=\frac{1}{\hat{s}_{11}^{2}(\xi)}\int_{\mathbb{R}}  
        \mathbf{R}_{j+2}^{\dagger}(x;\xi)\mathcal{J}
        \begin{pmatrix}
            \delta q_{1} \\
            \delta q_{2} \\
            \delta q_{1}^{*} \\
            \delta q_{2}^{*}
        \end{pmatrix}(x)
        \mathrm{d}x
        ,\\
        \delta\rho_{j}(\xi) &=-\frac{1}{s_{11}^{2}(\xi)}\int_{\mathbb{R}}  
        \mathbf{R}_{j}^{\dagger}(x;\xi)\mathcal{J}
        \begin{pmatrix}
            \delta q_{1} \\
            \delta q_{2} \\
            \delta q_{1}^{*} \\
            \delta q_{2}^{*}
        \end{pmatrix}(x)
        \mathrm{d}x,
    \end{split}
\end{equation}
where
\begin{equation*}
        \mathbf{R}_{j}=
        \begin{pmatrix}
            \hat{\psi}_{j+1,1}^{-}\phi_{21}^{+} \\
            \hat{\psi}_{j+1,1}^{-}\phi_{31}^{+} \\
            -\hat{\psi}_{j+1,2}^{-}\phi_{11}^{+} \\
            -\hat{\psi}_{j+1,3}^{-}\phi_{11}^{+} \\
        \end{pmatrix}, \quad 
        \mathbf{R}_{j+2}=
        \begin{pmatrix}
            \hat{\phi}_{11}^{+}\psi_{2,j+1}^{-} \\
            \hat{\phi}_{11}^{+}\psi_{3,j+1}^{-} \\
            -\hat{\phi}_{12}^{+}\psi_{1,j+1}^{-} \\
            -\hat{\phi}_{13}^{+}\psi_{1,j+1}^{-} \\
        \end{pmatrix},\quad j=1,2. 
\end{equation*}
Invoking \eqref{dqr-1} and \eqref{drho-1}, the closure relation is given by
\begin{equation*}
        \delta(x-y)\mathbb{I}_{4}=-\frac{{\rm i}}{\pi}\int_{\mathbb{R}}\sum_{j=1}^{2}\left(
            \frac{1}{s_{11}^{2}(\xi)}
            \mathbf{O}_{j}(x;\xi)\mathbf{R}_{j}^{\dagger}(y;\xi)\mathcal{J}
            -\frac{1}{\hat{s}_{11}^{2}(\xi)}
        \mathbf{O}_{j+2}(x;\xi)\mathbf{R}_{j+2}^{\dagger}(y;\xi)\mathcal{J}
    \right)\mathrm{d}\xi. 
\end{equation*}
Moreover, the orthogonality conditions hold
\begin{equation}\label{orth-cod-general}
    \begin{split}
        \int_{\mathbb{R}}  
        \mathbf{R}_{j}^{\dagger}(x;\xi)\mathcal{J}
        \mathbf{O}_{j}(x;\xi')
        \mathrm{d}x&={\rm i}\pi s_{11}^{2}(\xi)\delta(\xi-\xi'), \\
        \int_{\mathbb{R}}  
        \mathbf{R}_{j+2}^{\dagger}(x;\xi)\mathcal{J}
        \mathbf{O}_{j+2}(x;\xi')
        \mathrm{d}x&=-{\rm i}\pi \hat{s}_{11}^{2}(\xi)\delta(\xi-\xi').
    \end{split}
\end{equation}
For the case there exist discrete eigenvalues for Lax pair, 
the closure relation must contain the contributions of point spectrum. These  contributions are the residue of the functions in the above closure relation \cite{yang_nonlinear_2010}. 

\vspace{0.25cm}

{\bf Acknowledgements} Liming Ling is supported by National Natural Science Foundation of China (Nos. 12122105, 12471236).

%\subsection*{Data availability statement}
%Data sharing is not applicable to this article as no datasets were generated or
%analysed during the current study.
%
%\subsection*{Conflict of interest}
%On behalf of all authors, the corresponding author states that there is no conflict of
%interest.

\bibliographystyle{siam}
\bibliography{Ref-CNLS-stability}

\begin{thebibliography}{10}

\bibitem{APT2004}
{\sc M.~Ablowitz, B.~Prinari, and A.~Trubach}, {\em Discrete and Continuous
  Nonlinear Schrödinger Systems}, vol.~302 of London Mathematical Society
  Lecture Note Series, Cambridge University Press, 2003.

\bibitem{agrawal_nonlinear_2019}
{\sc G.~P. Agrawal}, {\em Nonlinear fiber optics}, Elsevier, 2019.

\bibitem{akhmediev1995phase}
{\sc N.~Akhmediev, A.~Buryak, J.~Soto-Crespo, and D.~Andersen}, {\em
  Phase-locked stationary soliton states in birefringent nonlinear optical
  fibers}, J. Opt. Sot. Am. B, 12 (1995), pp.~434--439.

\bibitem{Al2}
{\sc M.~Alejo}, {\em Nonlinear stability of gardner breathers}, J. Differential
  Equations, 264 (2018), pp.~1192--1230.

\bibitem{Al1}
{\sc M.~Alejo and C.~Munoz}, {\em Nonlinear stability of {MKdV} breathers},
  Comm. Math. Phys., 324 (2013), pp.~233--262.

\bibitem{alejo2021stability}
{\sc M.~A. Alejo, L.~Fanelli, and C.~Mu{\~n}oz}, {\em Stability and instability
  of breathers in the {$U (1)$Sasa--Satsuma} and nonlinear {Schr{\"o}dinger}
  models}, Nonlinearity, 34 (2021), p.~3429.

\bibitem{benjamin1972stability}
{\sc T.~B. Benjamin}, {\em The stability of solitary waves}, Proc. R. Soc.
  Lond. Ser. A Math. Phys. Eng. Sci., 328 (1972), pp.~153--183.

\bibitem{berkhoe_self_1970}
{\sc A.~L. Berkhoer and V.~E. Zakharov}, {\em Self excitation of waves with
  different polarizations in nonlinear media}, Soviet Phys. JETP, 31 (1970),
  pp.~486--490.

\bibitem{Prinari}
{\sc V.~Caudrelier, A.~Gkogkou, and B.~Prinari}, {\em Soliton interactions and
  {Yang-Baxter} maps for the complex coupled short-pulse equation}, Stud. Appl.
  Math., 151 (2023), pp.~285--351.

\bibitem{cazenave1983stable}
{\sc T.~Cazenave}, {\em Stable solutions of the logarithmic {Schr{\"o}dinger}
  equation}, Nonlinear Anal., 7 (1983), pp.~1127--1140.

\bibitem{cazenave_introduction_1989}
{\sc T.~{Cazenave}}, {\em An introduction to nonlinear {Schr\"odinger}
  equations}, Textos de Matodos Matematicos, 22 (1989).

\bibitem{cazenave_orbital_1982}
{\sc T.~Cazenave and P.~L. Lions}, {\em Orbital stability of standing waves for
  some nonlinear {Schr\"odinger} equations}, Comm. Math. Phys., 85 (1982),
  pp.~549--561.

\bibitem{cipolatti2000orbitally}
{\sc R.~Cipolatti and W.~Zumpichiatti}, {\em Orbitally stable standing waves
  for a system of coupled nonlinear {Schr{\"o}dinger} equations}, Nonlinear
  Anal., 42 (2000), pp.~445--462.

\bibitem{deconinck_orbital_2020}
{\sc B.~Deconinck and J.~Upsal}, {\em The orbital stability of elliptic
  solutions of the focusing nonlinear {Schrödinger} equation}, SIAM J. Math.
  Anal., 52 (2020), pp.~1--41.

\bibitem{faddeev1987hamiltonian}
{\sc L.~D. Faddeev and L.~A. Takhtajan}, {\em Hamiltonian methods in the theory
  of solitons}, Springer, 1987.

\bibitem{Ling}
{\sc B.~Feng and L.~Ling}, {\em Darboux transformation and solitonic solution
  to the coupled complex short pulse equation}, Physica D, 437 (2022),
  p.~133332.

\bibitem{GP-15}
{\sc T.~Gallay and D.~Pelinovsky}, {\em Orbital stability in the cubic
  defocusing {NLS} equation. {Part I: Cnoidal} periodic waves}, J. Differential
  Equations, 258 (2015), pp.~3607--3638.

\bibitem{GP-dark}
{\sc T.~Gallay and D.~{Pelinovsky}}, {\em Orbital stability in the cubic
  defocusing {NLS} equation. {Part II: The} black soliton}, J. Differential
  Equations, 258 (2015), pp.~3639--3660.

\bibitem{gerdjikov_generating_1981}
{\sc V.~S. Gerdjikov and P.~P. Kulish}, {\em The generating operator for the
  $n\times n$ linear system}, Phys. D, 3 (1981), pp.~549--564.

\bibitem{grillakis_analysis_1990}
{\sc M.~Grillakis}, {\em Analysis of the linearization around a critical point
  of an infinite dimensional {Hamiltonian} system}, Comm. Pure Appl. Math., 43
  (1990), pp.~299--333.

\bibitem{grillakis_stability_1987}
{\sc M.~Grillakis, J.~Shatah, and W.~Strauss}, {\em Stability theory of
  solitary waves in the presence of symmetry, {I}}, J. Funct. Anal., 74 (1987),
  pp.~160--197.

\bibitem{grillakis_stability_1990}
{\sc M.~Grillakis, J.~{Shatah}, and W.~Strauss}, {\em Stability theory of
  solitary waves in the presence of symmetry, {II}}, J. Funct. Anal., 94
  (1990), pp.~308--348.

\bibitem{haragus_spectra_2008}
{\sc M.~Ha\v{r}a\v{g}u\c{s} and T.~Kapitula}, {\em On the spectra of periodic
  waves for infinite-dimensional {Hamiltonian} systems}, Phys. D, 237 (2008),
  pp.~2649--2671.

\bibitem{LHP-17}
{\sc M.~H$\check{a}$r$\check{a}$gu\c{s}, J.~Li, and D.~Pelinovsky}, {\em
  Counting unstable eigenvalues in {Hamiltonian} spectral problems via
  commuting operators}, Comm. Math. Phys., 354 (2017), pp.~247--268.

\bibitem{Kap07}
{\sc T.~Kapitula}, {\em On the stability of {N-solitons} in integrable
  systems}, Nonlinearity, 20 (2007), pp.~879--907.

\bibitem{kapitula_counting_2004}
{\sc T.~Kapitula, P.~G. Kevrekidis, and B.~Sandstede}, {\em Counting
  eigenvalues via the {Krein} signature in infinite-dimensional {Hamiltonian}
  systems}, Phys. D, 195 (2004), pp.~263--282.

\bibitem{kapitula_spectral_2013}
{\sc T.~Kapitula and K.~Promislow}, {\em Spectral and dynamical stability of
  nonlinear waves}, vol.~185 of Appl. Math. Sci., Springer New York, 2013.

\bibitem{kaup_inverse_2009}
{\sc D.~J. Kaup and J.~Yang}, {\em The inverse scattering transform and squared
  eigenfunctions for a degenerate 3 × 3 operator}, Inverse Problems, 25
  (2009), p.~105010.

\bibitem{killip_orbital_2022}
{\sc R.~Killip and M.~Vişan}, {\em Orbital stability of {KdV} multisolitons in
  ${H}^{-1}$}, Comm. Math. Phys., 389 (2022), pp.~1445--1473.

\bibitem{koch2018conserved}
{\sc H.~Koch and D.~Tataru}, {\em Conserved energies for the cubic nonlinear
  {Schr{\"o}dinger} equation in one dimension}, Duke Math. J., 167 (2018),
  pp.~3207--3313.

\bibitem{koch2024multisolitons}
{\sc H.~{Koch} and D.~Tataru}, {\em Multisolitons for the cubic {NLS} in 1-d
  and their stability}, Publ. Math. Inst. Hautes Études Sci.,  (2024),
  pp.~1--116.

\bibitem{laurens_multisolitons_2023}
{\sc T.~Laurens}, {\em Multisolitons are the unique constrained minimizers of
  the {KdV} conserved quantities}, Calc. Var. Partial Differential Equations,
  62 (2023), p.~192.

\bibitem{LeCoz}
{\sc S.~Le~Coz and Z.~Wang}, {\em Stability of the multi-solitons of the
  modified {Korteweg–de Vries} equation}, Nonlinearity, 34 (2021),
  pp.~7109--7143.

\bibitem{li_structural_1998}
{\sc Y.~A. Li and K.~Promislow}, {\em Structural stability of non-ground state
  traveling waves of coupled nonlinear {Schrödinger} equations}, Phys. D, 124
  (1998), pp.~137--165.

\bibitem{li_mechanism_2000}
{\sc Y.~A. Li and K.~Promislow}, {\em The mechanism of the polarizational mode
  instability in birefringent fiber optics}, SIAM J. Math. Anal., 31 (2000),
  pp.~1351--1373.

\bibitem{ling_darboux_2015}
{\sc L.~Ling, L.-C. Zhao, and B.~Guo}, {\em Darboux transformation and
  multi-dark soliton for \textit{{N}}-component nonlinear {Schrödinger}
  equations}, Nonlinearity, 28 (2015), pp.~3243--3261.

\bibitem{ling_darboux_2016}
{\sc L.~{Ling}, L.-C. Zhao, and B.~Guo}, {\em Darboux transformation and
  classification of solution for mixed coupled nonlinear {Schr\"odinger}
  equations}, Commun. Nonlinear Sci. Numer. Simul., 32 (2016), pp.~285--304.

\bibitem{MS93}
{\sc J.~Maddocks and R.~Sachs}, {\em On the stability of {KdV} multi-solitons},
  Comm. Pure Appl. Math., 46 (1993), pp.~867--901.

\bibitem{manako_theory_1974}
{\sc S.~V. Manakov}, {\em On the theory of two-dimensional stationary
  self-focusing of electromagnetic waves}, Soviet Phys. JETP, 38 (1974),
  pp.~248--253.

\bibitem{mesentsev1992stability}
{\sc V.~Mesentsev and S.~K. Turitsyn}, {\em Stability of vector solitons in
  optical fibers}, Opt. Lett., 17 (1992), pp.~1497--1499.

\bibitem{nguyen2011orbital}
{\sc N.~V. Nguyen and Z.-Q. Wang}, {\em Orbital stability of solitary waves for
  a nonlinear {Schr{\"o}dinger} system}, Adv. Differential Equations, 16
  (2011), pp.~977--1000.

\bibitem{nguyen_existence_2015}
{\sc N.~V. Nguyen and Z.-Q. Wang}, {\em Existence and stability of a
  two-parameter family of solitary waves for a 2-coupled nonlinear
  {Schrödinger} system}, Discrete Contin. Dyn. Syst., 36 (2015),
  pp.~1005--1021.

\bibitem{ohta_stability_1996}
{\sc M.~Ohta}, {\em Stability of solitary waves for coupled nonlinear
  {Schr\"odinger} equations}, Nonlinear Anal., 26 (1996), pp.~933--939.

\bibitem{PY14}
{\sc D.~Pelinovsky and Y.~Shimabukuro}, {\em Orbital stability of dirac
  solitons}, Lett. Math. Phys., 104 (2014), pp.~21--41.

\bibitem{pelinovsky_inertia_2005}
{\sc D.~E. Pelinovsky}, {\em Inertia law for spectral stability of solitary
  waves in coupled nonlinear {Schrödinger} equations}, Proc. R. Soc. A, 461
  (2005), pp.~783--812.

\bibitem{pelinovsky_instabilities_2005}
{\sc D.~E. Pelinovsky and J.~Yang}, {\em Instabilities of multihump vector
  solitons in coupled nonlinear {Schrödinger} equations}, Stud. Appl. Math.,
  115 (2005), pp.~109--137.

\bibitem{qin_nondegenerate_2019}
{\sc Y.-H. Qin, L.-C. Zhao, and L.~Ling}, {\em Nondegenerate bound-state
  solitons in multicomponent {Bose}-{Einstein} condensates}, Phys. Rev. E, 100
  (2019), p.~022212.

\bibitem{radhakrishnan1995bright}
{\sc R.~Radhakrishnan and M.~Lakshmanan}, {\em Bright and dark soliton
  solutions to coupled nonlinear {Schr{\"o}dinger} equations}, J. Phys. A, 28
  (1995), p.~2683.

\bibitem{ramakrishnan2020nondegenerate}
{\sc R.~Ramakrishnan, S.~Stalin, and M.~Lakshmanan}, {\em Nondegenerate
  solitons and their collisions in {Manakov} systems}, Phys. Rev. E, 102
  (2020), p.~042212.

\bibitem{roskes_nonlinear_1976}
{\sc G.~J. Roskes}, {\em Some nonlinear multiphase interactions}, Stud. Appl.
  Math., 55 (1976), pp.~231--238.

\bibitem{silberberg1995rotating}
{\sc Y.~Silberberg and Y.~Barad}, {\em Rotating vector solitary waves in
  isotropic fibers}, Opt. Lett., 20 (1995), pp.~246--248.

\bibitem{stalin2020nondegenerate}
{\sc S.~Stalin, R.~Ramakrishnan, and M.~Lakshmanan}, {\em Nondegenerate soliton
  solutions in certain coupled nonlinear {Schr{\"o}dinger} systems}, Phys.
  Lett. A, 384 (2020), p.~126201.

\bibitem{stalin2019prl}
{\sc S.~Stalin, R.~Ramakrishnan, M.~Senthilvelan, and M.~Lakshmanan}, {\em
  Nondegenerate solitons in {Manakov} system}, Phys. Rev. Lett., 122 (2019),
  p.~043901.

\bibitem{Sipe}
{\sc M.~Tratnik and J.~Sipe}, {\em Bound solitary waves in a birefringent
  optical fiber}, Phys. Rev. A, 38 (1988), pp.~2011--2017.

\bibitem{upsal_real_2020}
{\sc J.~Upsal and B.~Deconinck}, {\em Real {Lax} spectrum implies spectral
  stability}, Stud. Appl. Math., 145 (2020), pp.~765--790.

\bibitem{wang_integrable_2010}
{\sc D.-S. Wang, D.-J. Zhang, and J.~Yang}, {\em Integrable properties of the
  general coupled nonlinear {Schrödinger} equations}, J. Math. Phys., 51
  (2010), p.~023510.

\bibitem{Wang22}
{\sc Z.~Wang}, {\em Isoinertial operators around the {KdV} multi-solitons},
  Nonlinear Anal., 219 (2022), p.~112820 (24 pp).

\bibitem{weinstein_lyapunov_1986}
{\sc M.~I. Weinstein}, {\em Lyapunov stability of ground states of nonlinear
  dispersive evolution equations}, Comm. Pure Appl. Math., 39 (1986),
  pp.~51--67.

\bibitem{yagasaki_bifurcations_2023}
{\sc K.~Yagasaki and S.~Yamazoe}, {\em Bifurcations and spectral stability of
  solitary waves in coupled nonlinear {Schrödinger} equations}, J.
  Differential Equations, 372 (2023), pp.~348--401.

\bibitem{yang_nonlinear_2010}
{\sc J.~Yang}, {\em Nonlinear waves in integrable and nonintegrable systems},
  SIAM, 2010.

\bibitem{yang_squared_2009}
{\sc J.~Yang and D.~J. Kaup}, {\em Squared eigenfunctions for the
  {Sasa}–{Satsuma} equation}, J. Math. Phys., 50 (2009), p.~023504.

\end{thebibliography}

\end{document}